\documentclass[11pt]{article}

\usepackage{template}
\usepackage{setspace}
\usepackage{comment}
\usepackage{multirow}
\usepackage{subcaption}

\title{Covariate-assisted graph matching}
\author{Trisha Dawn\thanks{\href{mailto:trdawn@ucsc.edu}{trdawn@ucsc.edu}}, Jes\'us Arroyo\thanks{\href{mailto:jarroyo@tamu.edu}{jarroyo@tamu.edu}}\\
\normalsize{*Department of Statistics, University of California, Santa Cruz \\
\textdagger Department of Statistics, Texas A\&M University}}

\begin{document}

\date{}
\maketitle

\begin{abstract}

Data integration often requires  aligning datasets without unique identifiers; for network data this is commonly addressed through graph matching.
Despite the widespread availability of node and edge covariates in many applications, most existing methods rely solely on the structural similarity between the two graphs, limiting their flexibility and robustness to non-identifiable matches and structural noise. In this article, we study the seeded graph matching problem,  where a partial alignment for a set
of nodes (seeds) is known, in the presence of informative node and edge covariates. We propose two covariate-assisted methods: a global approach based on a quadratic assignment problem (QAP), and a scalable local approach that leverages the neighborhood of the seeds. Both methods are based on a flexible conditional modeling framework that enables covariate integration using predictive models ranging from generalized linear models to machine learning methods.
We provide estimation error rates for the generalized linear model, and sufficient conditions for exact recovery of the true correspondences via the QAP. We further specialize these results to Erd\H{o}s-R\'enyi model with covariates. The effectiveness in achieving improved performance by leveraging additional covariates is demonstrated through simulations and in an application to matching the statistics academic genealogy and the collaboration networks. 
\end{abstract}

\tableofcontents

\section{Introduction}\label{Intro}

Networks (or graphs) are a structured representation of objects or agents and the relationships between them. The nodes or vertices of a graph denote the units, and edges encode their interactions. These data representations are widely used across various application areas, including social networks 
protein interactions 
and power grids,
 among many others.
Network data analysis encompasses a collection of techniques for examining these relationships and characterizing the overall structure of the system. Extensive surveys for statistical models and methods can be found for example in \cite{newman2018networks, sengupta2025statistical}. 

In recent years, there has been a growing interest in integrating of multiple data modalities for the analysis of networks. Examples include multilayer networks \citep{kivela2014multilayer,peixoto2015inferring}, multi-view data \citep{salter2017latent,gao2022testing} or network data with covariates \citep{fosdick2015testing,binkiewicz2017covariate,james2024learning}. Analysis of such datasets offers the opportunity to uncover new insights by combining complementary sources of information. Because these modalities are obtained from different sources, a crucial step in their joint analysis is to establish the correspondence between the entities associated with the nodes across the different datasets.  

Building on the need for accurate cross dataset alignment, this article focuses on the problem of matching nodes across two different network data, which is known as the problem of graph matching. This problem arises when the node correspondences across networks are unknown. 
Some applications where graph matching plays a vital role include video analysis \citep{caetano2009learning}, social network re-identification \citep{pedarsani2011privacy}, alignment of biological networks \citep{pedigo2023bisected} and fingerprint recognition \citep{cui2024contactless}.
Comprehensive surveys of graph matching methodologies along with applications can be found in \cite{foggia2014graph, yan2016short}.  

In many practical scenarios, network data is accompanied by auxiliary information such as edge or node covariates. Examples include user-profile data in social networks or cellular function of nodes in gene regulatory networks. These features can serve as a powerful tool for different inference tasks, including community detection \citep{yang2013community}, node embeddings \citep{ma2020universal}, and vertex nomination \citep{levin2020vertex}. However, in the graph matching problem, most of the existing methodologies are limited to using only network structural information. When additional features are incorporated, this is often done through pre-computed similarity matrices from the features, limiting the flexibility of these methodologies.

In this paper, we study the problem of graph matching with auxiliary edge and node covariates. In particular, we focus on the seeded version of this problem, where a subset of nodes has a known correspondence in advance, referred to as the seed nodes. In some applications, the seeds arise naturally. For example, in social network re-identification, the users connecting their accounts across various social media platforms can serve as seeds \citep{narayanan2009anonymizing}. Even the knowledge of a few seeds has proven to substantially enhance matching outcomes \citep{fishkind2019seeded}. In our setup, the use of seeds is crucial for inferring the effect of auxiliary covariates in the matching problem. Below we summarize our main contributions.



\begin{itemize}
    \item We propose a framework for covariate-assisted graph matching, based on a model that expresses one graph in terms of a function of the other graph and the available covariate information to guide the matching process.
     While we focus on a generalized linear regression model, we show how this framework is flexible and can be extended to other types of prediction methods, such as random forests.
    \item We propose two methods for covariate-assisted seeded graph matching. The first method is based on approximately solving a Quadratic Assignment Problem (QAP) across the entire collection of nodes and edges, and the second one leverages the neighborhood structure of the seed nodes for computational scalability. While the QAP-based method achieves better matching accuracy, its computational cost is significantly higher, and hence, the suitability of the methods depend on the specific application as well as on the practitioner.
    \item We establish a non-asymptotic estimation error bound for the fitted prediction values of the model, and provide sufficient conditions for exact matching of the estimated vertex correspondence under the optimal solution of the proposed QAP. Furthermore, we specialize these results to an Erd\H{o}s--R\'enyi graph model with covariates. 
    \item We evaluate the performance of our proposed methods through simulations, comparing graph matching accuracy and run time. We also demonstrate their performance on a pair of real networks, showing significant improvement over existing approaches that do not incorporate covariates.
\end{itemize}

To motivate our methodology and illustrate the role of covariate information in practice, we study an application of matching an academic genealogy network and a collaboration network (Figure~\ref{fig: networksAB} in Supplementary Material), each containing $8,627$ nodes, representing statisticians. Matching these two networks not only can help reveal how academic lineage translates into collaboration, but also can provide insights into broader fundamental questions arising from the ``science of science'' field \citep{fortunato2018science,wang2021science}, such as the development of individual scientific careers, the evolution of various disciplines, and factors shaping scholarly impact. The genealogy network encodes doctoral advisor\mbox{-}student relationships based on academic lineage, while the collaboration network captures coauthorship relationships with edges denoting scholarly collaborations. The former is constructed from the Mathematics Genealogy Project and the latter from the Multi\mbox{-}Attribute Dataset on Statisticians (MADStat) \citep{DVN_V7VUFO_2022}. As these two networks are obtained from different sources, a crucial step for joint analysis is to match the individuals across the networks. Use of auxiliary information associated with the genealogy network, including graduation institution, country and year, offers substantial improvement in the matching accuracy.  Unlike traditional graph matching approaches that ignore auxiliary features, our method explicitly incorporates both node and edge-level covariates for matching the two networks.

The rest of the article is organized as follows. Section~\ref{Litreview} provides a review of the existing literature. Section~\ref{method} introduces the general formulation of the problem, the proposed model, and the proposed methods. Section~\ref{theory} establishes theoretical guarantees for model estimation and a perfect matching result for the solution of the matching problem.  Section~\ref{sim} reports results from numerical simulations, comparing the performance and run time of different methods. Finally, Section~\ref{dataApp} applies the proposed and competing methods to match the academic genealogy and collaboration networks, followed by a discussion in Section~\ref{discussion}.


\subsection{Related work}\label{Litreview}

The goal of this paper is to match entities across two different data sources including network data with covariates. One of the main techniques for entity resolution is record linkage, which consolidates information from various sources to enable joint statistical inference. 
For extensive overview of the available statistical methods, refer to \cite{dey2010efficient, winkler2014matching}. These approaches are often focused on matching individual observations based on covariate data itself, but are unable to account for the network structure across the observations.
Although the literature on graph matching is usually treated separately from general record linkage, it can be viewed as a special case of record linkage when the data sources to be matched are graph valued objects.

Graph matching has a long and rich history. The existing literature can be viewed both from an algorithmic and a modeling perspective. The algorithms available often include optimization techniques and different search strategies (see \cite{qiao2021igraphmatch} for an overview). A popular approach consists of finding the permutation that minimizes a matching error or graph edit distance across two graphs \citep{pedarsani2011privacy}, which can be written as a quadratic assignment problem (QAP) \citep{edwards1980branch}
Given two adjacency matrices $A,\widetilde{B}\in\{0,1\}^{n\times n}$ for graphs with the same number of vertices, the problem of graph matching is usually defined by finding a permutation matrix $\widehat{Q}\in\Pi_n$ (with $\Pi_n$ the set of $n\times n$ permutation matrices) that makes the rows and columns of $A$ and $\Tilde{B}$ structurally most similar, by solving a QAP given by
\begin{equation}\label{eq: classicGM}
   \widehat{Q} = \argmin_{Q \in \Pi_n} \| A - Q^\top \widetilde{B} Q \|^2_F = \argmax_{Q\in\Pi_n} \operatorname{Tr}\left(AQ^\top \widetilde{B}Q\right).
\end{equation}
 Here, $\|\cdot\|_F$ denotes the Frobenius norm of a matrix and $\operatorname{Tr}(\cdot)$ denotes the trace of a matrix.
Unfortunately, this optimization problem is NP hard \citep{finke1987quadratic}. Even the special case of exact graph matching, also known as graph isomorphism, remains an open problem with respect to computational complexity, where the best known algorithm runs in quasi-polynomial time \citep{babai2016graph}. However, several relaxation-based efficient approximations exist to serve as a  practical tool \citep{vogelstein2015fast, fishkind2019seeded, lyzinski2015graph}.

In addition to these optimization approaches, other prominent classes of matching methods utilize special properties of the pair of graphs to perform the alignment. These methods include spectral-based approaches \citep{umeyama2002eigendecomposition, fan2019spectral}, methods based on degree-profile information \citep{ding2021efficient} or subtree-counts \citep{mao2023random}. These methods can perform seedless graph matching in polynomial time at the cost of stronger signal strength requirements. A different line of work involving machine learning methods includes graph neural networks \citep{jiang2022glmnet} using a semi-supervised or supervised learning framework. Seeded graph matching methods are usually based on percolation, where the matches propagate starting from some initial seed set \citep{yartseva2013performance, kazemi2015growing}, or are based on restricted optimization problems to the neighborhood of the seed nodes \citep{lyzinski2014seeded,mossel2020seeded}.

The methods discussed so far only deal with graph structural information but do not incorporate additional features, such as edge or node covariates. The use of these features is relatively unexplored in the context of graph matching. Notable existing approaches include learning a similarity matrix instead using features \citep{zhang2018attributed}, or constructing seed nodes based on node features \citep{zhang2024attributed}. These approaches implicitly assume that the features contribute to the similarity between graphs, but are unable to leverage the effect of different attributes and graph information.

Alongside algorithmic developments, significant advances have also been made in developing various random graph models to explore the difficulty and viability of matching accuracy. The proposed models in literature include correlated Erd\H{o}s--R\'enyi models \citep{pedarsani2011privacy, lyzinski2014seeded} with some extensions to inhomogeneous Bernoulli graphs \citep{onaran2016optimal, lyzinski2018information,arroyo2021maximum}.
Several existing algorithms, including many of the ones previously mentioned, possess theoretical guarantees for recovery of the latent vertex alignment under these models.  In particular, the optimization problem in Equation~\eqref{eq: classicGM} is guaranteed to consistently recover the matching permutation under the correlated Erd\H{o}s--R\'enyi model \citep{lyzinski2014seeded,Cullina2016,cullina2017exact} with sharp information-theoretic thresholds for matching recovery already established \citep{wu2022settling}. 

A practical limitation of the correlated Erd\H{o}s-R\'enyi modeling framework is the underlying assumption that the graphs have identical marginal distributions, which can be unrealistic in practice
\citep{lyzinski2020matchability}. In our approach, we use the edge and node covariate information to model the heterogeneity in the marginal distributions for the two graphs. This often arises in practice where the two graphs to be matched are generated from two different mechanisms but share common vertices. For example, users across different social media platforms might behave differently, which generates different network structures to be matched for social network re-identification. We consider a similar example in our motivating application (see Section~\ref{dataApp}) where the genealogy network presents a tree-like structure while the collaboration network has many loops and a higher density.

To leverage available node or edge covariate information in addition to graph structure, we consider a regression model for network responses. This problem has been widely studied in the literature, with several approaches including latent class models \citep{nowicki2001estimation}, latent distance models \citep{hoff2002latent, wu2017generalized, zhang2022joint}, latent eigenmodels \citep{loyal2025spike} and community structure \cite{pan2021inference, huang2024pcabm} to incorporate block structures. In this work, we develop a unified framework that builds on the strengths of existing graph matching algorithms and regression models for network responses to incorporate node and edge level covariates to improve graph matching accuracy in practical data integration tasks.




\section{Methodology}\label{method}
The goal of this paper is to find a correspondence between the vertices of two graphs, generally known as the graph matching problem. A graph $G=(V, E)$ consists  of the set of nodes (or vertices) $V$ and the set of edges between the nodes $E\subset V\times V$. In this paper, we consider two such graphs $G_1 = (V_1, E_1)$ and $G_2 = (V_2, E_2)$. We assume the graphs have $n$ nodes each and there is a one-to-one correspondence between the node sets $V_1$ and $V_2$, in which case we can write their labels as $V_l = [n]=\{1, \cdots, n\},$ $l = 1, 2$, but the edges $E_1$ and $E_2$ are potentially different. For simplicity, we consider simple graphs, meaning the edges are unweighted, undirected and with no self-loops, but extensions to other settings are discussed. The graphs are represented by their corresponding adjacency matrices $A, {B} \in \{0, 1\}^{n \times n}$ for $G_1, G_2$ respectively, which are symmetric with zeros in their diagonals. Presence of an edge between nodes $i$ and $j$ in $G_1$ (resp. $G_2$) is denoted by $A_{ij} = 1$ (resp. ${B}_{ij} = 1)$ and no edge is denoted by zero. We will use $A, {B}$ interchangeably with $G_1, G_2$ in the subsequent sections. 

To represent the fact that the correspondence between the nodes is unknown, we denote by $\widetilde{B}\in\{0,1\}^{n\times n}$ to the observed adjacency matrix of $G_2$, with its vertices shuffled by an unknown permutation matrix $Q^\ast\in\Pi_n$, where $\Pi_n$ is the set of all permutation matrices of size $n$, in which case $\widetilde{B} = Q^{{\ast}^\top} B Q^\ast$. Hence, the observed data consist of matrices $A$ and $\widetilde{B}$, and the goal is to infer $Q^\ast$ to determine the unknown alignment. 
 
In addition to graph data, we also consider edge or node covariate information coming from one of the graphs to inform the graph matching problem. 
Edge covariates are encoded in vectors denoted by $Y_{ij} = (Y_{ij}^{(1)}, \ldots, Y^{(d_1)}_{ij})\in\mathbb{R}^{d_1}$ for $i,j\in[n]$, and let $Y=\{Y_{ij}:i,j\in[n]\}$. Similarly, node covariates are represented by vectors $Z_{i} = (Z_i^{(1)}, \ldots, Z_i^{(d_2)})\in\real^{d_2},$ $i = 1, \cdots, n$, and let $Z = \{Z_{i}:i\in[n]\}$. For each node covariate $k\in[d_2]$, we assume that there is a known deterministic function $h_k:\real^2\rightarrow \real$ to transform this covariate into an edge covariate, such as a similarity or distance function. For example, this function can denote the absolute difference $h_k(Z^{(k)}_i, Z^{(k)}_j) = |Z^{(1)}_i - Z^{(1)}_j|$, or an indicator variable $h_k(Z_i^{(k)}, Z_j^{(k)}) = \mathbbm{1}(Z_i^{(k)} = Z_j^{(k)})$.

To make the graph matching problem with covariates more tractable, we assume that the vertex correspondence for a subset of nodes $\mathcal{S}\subset[n]$ is known. This set $\mS$ is referred to as the \emph{seed nodes}, with $s=|\mathcal{S}|$ denoting the total number of seeds. In real-world applications, the set $\mathcal{S}$ is either known a priori or can be obtained using available information, such as unique identifiers or manual inspection. Without loss of generality, we assume that $Q^\ast_{\mS,\mS} = \mathrm{I}_s$  where $ \mathrm{I}_s$ be the $s\times s$ identity matrix, which ensures that the rows and columns of $A$ and $\widetilde{B}$ with indexes in $\mS$ are aligned.

To summarize, given an adjacency matrix from a graph $A$ with edge and node covariate information $Y,Z$, a second observed graph with its vertices partially permuted $\widetilde{B}$, and a set of seeds $\mathcal{S}\subset[n]$ of nodes with known correspondence, our goal is to infer the correspondence of the remaining $\wbmS = [n]\setminus \mathcal{S}$ nodes.

\subsection{Modeling graph matching with covariates}\label{ProposedModel}

We consider a conditional model for the graph $B$ given the graph $A$ and its covariates.
In particular, we assume that the entries
of $B$ are independent Bernoulli random variables with probabilities $P_{ij}$, $i,j\in[n]$,  which are a function of the graph $A$ and the covariates. Specifically, we assume that the edges of the graph $B$ follow a generalized linear model \citep{mccullagh1989generalized} that is a function of $A$ and the covariate data $Y, Z$ of the form
\begin{equation}
    \begin{split}
    B_{ij} | A, Y, Z \sim \text{Ber}(P_{ij}),\quad\quad i,j\in[n], i>j,\label{eq:glm-for-B}\\
    g(P_{ij}) = \theta_0^\ast + A_{ij}\theta_1^\ast +  Y_{ij}^\top\theta^\ast_{\text{edge}} + h(Z_i, Z_j)^\top\theta^\ast_{\text{node}},
    \end{split}
\end{equation}
where $h(Z_i, Z_j) = (h_1(Z_i^{(1)}, Z_j^{(1)}), \ldots, h_{{d_2}}(Z_i^{(d_2)}, Z_j^{(d_2)}))$. Here, $g$ is the link function corresponding to the specific model, with possible choices including logistic or linear. The parameters $\theta_0^\ast,\theta_1^\ast\in\real$ are the intercept and network coefficients, respectively, $\theta^\ast_{\text{edge}}\in\real^{d_1}$ is a vector of the edge covariate coefficients, and $\theta^\ast_{\text{node}}\in\real^{d_2}$ is the node covariates coefficient vector. Letting $\theta^\ast = (\theta_0^\ast, \theta_1^\ast, \theta_{\text{edge}}^\ast, \theta_{\text{node}}^\ast)\in\real^{d}$, $d = d_1 + d_2 + 2$, be the coefficient vector, $X_{ij} = (1, A_{ij}, Y_{ij}, h(Z_{i},Z_{j}))$ be the vector of predictors, and $\mu = g^{-1}$ be the inverse link or mean function, we write the model in a more compact notation as
\begin{equation*}
    B_{ij}|X \sim\text{Ber}(P_{ij}), \quad\quad\quad P_{ij} =  \mu(X_{ij}^\top\theta^\ast).
\end{equation*}
Due to the graphs being undirected with no self-loops, we assume that $P_{ij}=P_{ji}$, $X_{ij} = X_{ji}$, and $P_{ii}=0$, $X_{ii} = \mathrm{0}_d$,where $0_d$ is the zero vector of length $d$, for all $i,j\in[n]$.

The model in Equation~\eqref{eq:glm-for-B} is related to the literature of regression models for network-valued responses (see Section~\ref{Intro}). Since our goal here is to match the graphs, our model choice intends to be flexible enough to capture variation across the graphs, while remaining feasible to estimate in the presence of unknown node correspondences across the graphs. Since $\widetilde{B} = Q^*{B}Q^{*\top}$, this model in turn implies that the elements of the observed matrix $\widetilde{B}$ are also conditionally modeled given elements of $A$, $Y$ and $Z$ under some unknown permutation matrix $Q^\ast\in\Pi_n$. 

Notably, our modeling framework encompasses the popular correlated Erd\H{o}s--R\'enyi (ER) model \citep{pedarsani2011privacy,lyzinski2014seeded} as a special case. In this model, with edge probability $p$ and correlation $\varrho$, the distribution of the edges of $A$ and $B$ satisfies $\p(A_{ij}=1) = \p(B_{ij}=1) = p$ and $\text{Corr}(A_{ij}, B_{ij}) = \varrho \geq 0$, independently for each $i > j$. Equivalently, this can be written as
\begin{align*}
     \p(B_{ij} = 1 \mid A_{ij}) = (1-\varrho)p + \varrho A_{ij},
\end{align*}
which corresponds to a special case of our model with the identity link function and with coefficients $\theta_0^\ast = (1-\varrho)p$ and $\theta_1^\ast = \varrho$.

\begin{remark} The model formulation in Equation~\eqref{eq:glm-for-B} can easily handle extensions to different graph structures. For instance, directed networks might often have different values for $P_{ij}$ and $P_{ji}$. The choice of a Bernoulli distribution for the response can also be modified to handle numerical or categorical variables with other distributions. More generally, our methodology can be extended to any prediction framework with $B$ as the response and $A,Y,Z$ as covariates. Examples include, nonparametric methods such as random forests \citep{breiman2001random} or penalized estimators to handle high-dimensional covariates \citep{hastie2015statistical}. Our choices here are motivated to strike a balance between model flexibility and theoretical analysis, but we illustrate extensions of this framework in Section~\ref{dataApp}.
\end{remark}

\subsection{Matching via a quadratic assignment problem}\label{algo1detailed}

Recall that our goal is to find the permutation matrix $Q^* \in \Pi_n$ such that the vertices of $Q^\ast \widetilde{B} Q^{*\top}=B$ correspond to the ones of $A$. 
As $P = \mathbbm{E}[B|A, Y, Z]$, by properties of the expectation, it is clear that the expected mean squared error $\mathbbm{E}[\|B - M\|_F^2|A, Y, Z]$ is minimized with respect to $M\in\real^{n\times n}$ if $M=P$. Hence, this property implies that the permutation $Q^\ast$ satisfies
\begin{equation}\label{eq: QuadOptim-Expectation}
    {Q}^\ast \in \argmin_{Q \in \Pi_n} \mathbbm{E}\left[ \| P - Q \widetilde{B} Q^\top \|^2_F \Big|A,Y,Z\right].
\end{equation}
Moreover, the solution to the above optimization problem is unique whenever $P \neq QPQ^\top$ for every $Q\in \Pi_n,$ $Q\neq I_n$, that is,
there are no automorphisms of $P$ other than the identity. However, this optimization problem in Equation~\eqref{eq: QuadOptim-Expectation} requires to compute an expectation with respect to $B$ and needs $P$ to be known and demands a suitable estimator of $P$ to work with. Thus, given some estimator $\widehat{P}$ of $P$ (to be defined later), motivated by \eqref{eq: QuadOptim-Expectation} we propose to estimate $Q^\ast$ using the following optimization problem
\begin{equation}\label{eq: QuadOptim}
    \widehat{Q} = \argmin_{Q \in \Pi_n} \| \widehat{P} - Q \widetilde{B} Q^\top \|^2_F,\quad\quad\quad\text{subject to }Q_{\mS,\mS}=I_s.
\end{equation}
This problem, equivalent to a quadratic assignment problem, 
attempts to find the permutation matrix $\widehat{Q}$ which minimizes the sum of squared errors between a permuted version of $\widetilde{B}$ and an estimator of its expectation. This optimization problem is conceptually similar to the classic formulation of graph matching with a quadratic assignment problem as in Equation~\eqref{eq: classicGM}, but the adjacency matrix $A$ is replaced with $\widehat{P}$. The constraint $Q_{\mS, \mS} = I_s$ incorporates the knowledge of the node correspondence in the seeds. While an exact solution of this problem is computationally hard \citep{finke1987quadratic}, several existing approaches can approximate its optimal value; we defer the discussion of this point to the end of the section.

Finding a suitable estimator of $P$ under the model in Equation~\eqref{eq:glm-for-B} can be performed when the knowledge of some correspondences between the edges of $A$ and $\widetilde{B}$ is available. We thus use the set of seeds $\mathcal{S}\subset[n]$ to fit a regression model in order to estimate the coefficient vector ${\theta^\ast}$ in Equation~\eqref{eq:glm-for-B}, and denote this estimator as $\widehat{\theta}$. As we assume that the nodes in $\mathcal{S}$ are correctly aligned across the graphs ($Q^\ast_{\mathcal{S}, \mathcal{S}} = I_s$), this estimator only requires the submatrices $A_{\mS, \mS}, \widetilde{B}_{\mS, \mS}$ and the covariate information for the seed nodes $Y_{\mS, \mS} = \{Y_{ij}: i,j\in\mS\},$ $Z_{\mS} = \{Z_i: i\in\mS\}$. The estimator for $\theta$, obtained via minimization of the generalized linear model loss function, is given by
\begin{equation}\widehat{\theta} = \argmin_{\theta\in\real^{d}}  L(\theta), \quad\quad \text{where }L(\theta) = \sum_{i,j\in \mathcal{S}, i>j}\left\{-\widetilde{B}_{ij}(X_{ij}^\top \theta) + \psi(X^\top_{ij}\theta)\right\}.
\label{eq:glm-loss-function}
\end{equation}
Here, $\psi:\real\rightarrow\real$ is the cumulant function of the generalized linear model and satisfies $\psi'(t) = \mu(t)$. Examples include $\psi(t) = t^2/2$ for ordinary least squares, and $\psi(t) = \log(1+e^t)$ for logistic regression.

Once the estimator for the coefficient vector is obtained, this is extended to estimate the probability matrix $P$ on the entire graph via the regression model as
\begin{equation}\label{eq:estimatorP}
    \widehat{P}_{ij} = \mu\left(X^\top_{ij}\widehat{\theta} \right) = \mu\left(\widehat{\theta}_0 + A_{ij}\widehat{\theta}_1  + Y^\top_{ij}\widehat{\theta}_{\text{edge}} + h(Z_i, Z_j)^\top\widehat{\theta}_{\text{node}}\right), \quad \quad i,j\in[n], i\neq j,
\end{equation}
and $\widehat{P}_{ii}=0$. Given $\widehat{P}$, the solution to the quadratic assignment problem in Equation~\eqref{eq: QuadOptim} is used to estimate the unshuffling permutation $Q^\ast$. This process is summarized in Algorithm~\ref{alg1}. 

Solving the quadratic assignment problem in the last step of the method is a well-known NP hard problem, and therefore, there are no known polynomial-time algorithms available. While exact solutions are only possible for very small graphs, relaxations of this problem are shown to be accurate and perform well in practice \citep{zaslavskiy2008path, lyzinski2015graph}. These methods relax the permutation matrix constraint to its convex envelope and then use a continuous optimization algorithm, such as gradient descent. In our implementation, we use the Fast Approximate Quadratic Programming (FAQ) method with seeds \citep{vogelstein2015fast,fishkind2019seeded}, which is based on the Frank-Wolfe algorithm.


\begin{algorithm}
\caption{QAP-based Covariate Assisted Graph Matching with Seeds (CovQAP)}
\label{alg1}
\begin{algorithmic}
 \State \textbf{Input:} Adjacency matrices $A$, $\widetilde{B}$, edge covariates $Y$, node covariates $Z$, seed nodes $\mS$.
 \Procedure{\State 1. Extract the information from the seed nodes $\mathcal{S}$, denoted by $A_{\mathcal{S}, \mathcal{S}}, \widetilde{B}_{\mathcal{S}, \mathcal{S}}$, $Y_{\mathcal{S}, \mathcal{S}}, Z_{\mathcal{S}}$.
\State 2. Fit the regression model as in Eq.~\eqref{eq:glm-for-B} using the seed nodes to obtain $\widehat{\theta} = (\widehat{\theta}_0, \widehat{\theta}_1,  \widehat{\theta}_\text{edge}, \widehat{\theta}_\text{node})$.
 \State 3. Obtain estimate of $P$ on the entire graph, denoted by $\widehat{P}_{ij} = \mu(X_{ij}^\top \widehat{\theta})$ as in Eq.~\eqref{eq:estimatorP}.
 \State 4. Solve the quadratic assignment problem in Eq.~\eqref{eq: QuadOptim} using $\widehat{P}$, $\widetilde{B}$ and $\mS$ to obtain $\widehat{Q}$.
 \State \textbf{Output:} Permutation matrix $\widehat{Q}\in\Pi_n$.}
\EndProcedure
\end{algorithmic}
\end{algorithm}

\subsection{An efficient algorithm via seeded neighborhood matching}
We now present a computationally efficient method to solve the graph matching problem with covariates. This proposed algorithm is based on only using the local neighborhood of the seed nodes on the unseeded ones. The idea of utilizing these neighborhoods has been explored previously in the seeded graph matching literature without covariates \citep{lyzinski2014seeded,mossel2020seeded}. This idea involves matching the non-seed vertices of the given graphs according to the similarity of their connections to the nodes in the local neighborhood  seeded set of nodes $\mS$. Our algorithm follows this idea and incorporates the estimation of $P$ in order to handle the covariates.

Specifically, for a pair of non-seed vertices $i,j\in\wbmS = [n]\setminus \mS$, corresponding to the graphs $A$ and $\widetilde{B}$ respectively, we compare the neighborhood of these vertices among the seed nodes in $\mS$. To motivate the idea, assuming first that $P$ is known, for two unmatched vertices $i$ and $j$ we can compare the values of $P_{ik}$ and $\widetilde{B}_{jk}$ for all $k\in\mathcal{S}$ to determine if $i$ and $j$ are a correct match. In particular, if the pair $i$ and $j$ is the correct correspondence, then $\mathbbm{E}[\widetilde{B}_{jk}|A,Y,Z] = P_{ik}$, which implies that the expected square distance between the vectors $P_{i,\mS}, \widetilde{B}_{j,\mS}\in\real^s$, given by
$\mathbbm{E}\left[\|P_{i,\mS} - B_{j,\mS}\|^2 |A,Y,Z\right] = \sum_{k\in\mathcal{S}} \mathbbm{E}(P_{ik} - \widetilde{B}_{jk}|A,Y,Z)^2,$
is minimized when $i$ and $j$ are a match. Using this idea, we estimate the permutation matrix $\widehat{Q}$ by minimizing a measure of the quality of a match that compares the rows of the matrices $\widehat{P}_{\wbmS, \mS},\widetilde{B}_{\wbmS, \mS}\in\real^{(n-s)\times s}$, formed by the entries of $\widehat{P}$ and $\widetilde{B}$ in $\wbmS\times \mS$, according to the squared error given by 
\begin{equation}
\widehat{Q}_{\wbmS, \wbmS} = \argmin_{Q\in\Pi_{n-s}} \left\|\widehat{P}_{\wbmS, \mS} - Q \widetilde{B}_{\wbmS, \mS}\right\|^2_F =
    \argmax_{Q\in\Pi_{n-s}}\Tr(Q\widetilde{B}_{\wbmS,\mS}\widehat{P}_{\wbmS,\mS}^\top), \label{eq:lin_assign_prob}
\end{equation}
and we set $\widehat{Q}_{\mS, \mS} = I_s$. This formulation is thus equivalent to a linear assignment problem, which can be solved exactly in polynomial time by the Hungarian algorithm \citep{Kuhn1955}.

The estimation procedure of our second method using the neighborhood of the seed nodes is summarized in Algorithm~\ref{alg2}.
The following proposition justifies this approach by showing that the true permutation minimizes the expected loss function based on the true probability matrix $P$.

\begin{proposition}\label{prop:linearassignmentsol}
Suppose that $B,P\in\real^{n\times n}$ are matrices such that $\mathbbm{E}[B] = P$, and let $\widetilde{B} = Q^\ast B Q^{\ast^\top}$, where $Q^\ast\in\Pi_n$ is a permutation matrix. For any given set of seeds $\mS\subset[n]$ and $\wbmS = [n]\setminus\mS$, the permutation $Q^\ast$ satisfies
    \begin{equation}\label{eq:prop-linearassignment}
        Q^\ast_{\wbmS, \wbmS} \in \argmax_{Q\in\Pi_{n-s}} \mathbbm{E}\left[\Tr\left(Q \left(\widetilde{B}_{\wbmS, \mS} P^\top_{\wbmS, \mS} \right)\right)\right].
    \end{equation}
    Moreover, $Q^\ast$ is the unique solution of the above optimization problem if $P_{\wbmS, \mS}$ does not have repeated rows.
\end{proposition}
According to the proposition, the correct permutation $Q^\ast$ is the solution to the expected loss function using $P$. However, in contrast with the solution in Equation~\eqref{eq: QuadOptim-Expectation}, the uniqueness of the solution here depends on the particular value of $\mS$ to make the rows of $P_{\wbmS, \mS}$ distinguishable. In addition, the size of $\mS$ affects the quality of the estimation error. Compared to the optimization problem in Equation~\eqref{eq: QuadOptim}, which uses all the $O(n^2 - s^2)$ edges available to perform the matching, the neighborhood method uses only a submatrix of size $O((n-s)s)$. This can result in some statistical efficiency loss, but with the advantage of gaining computational scalability as algorithms with computational complexity $O((n-s)^3)$ can obtain the exact solution of the problem \citep{vogelstein2015fast}. We compare these trade-offs in simulated data in Section~\ref{sim}.


\begin{algorithm}[ht]
\caption{Seeded Neighborhood-based Covariate Assisted Graph Matching (CovNeigh)}
\label{alg2}
\begin{algorithmic}
 \State \textbf{Input:} Adjacency matrices $A$, $\widetilde{B}$, edge covariates $Y$, node covariates $Z$, seed nodes $\mS$.
 \Procedure{\State 1. Extract the information from the seed nodes $\mathcal{S}$, denoted by $A_{\mathcal{S}, \mathcal{S}}, \widetilde{B}_{\mathcal{S}, \mathcal{S}}$, $Y_{\mathcal{S}, \mathcal{S}}, Z_{\mathcal{S}}$.
\State 2. Fit the regression model as in Eq.~\eqref{eq:glm-for-B} using the seed nodes to obtain $\widehat{\theta} = (\widehat{\theta}_0, \widehat{\theta}_1,  \widehat{\theta}_\text{edge}, \widehat{\theta}_\text{node})$.
 \State 3. Obtain estimate of $P$ on the entire graph, denoted by $\widehat{P}_{ij} = \mu(X_{ij}^\top \widehat{\theta})$ as in Eq.~\eqref{eq:estimatorP}.
 \State 4. Solve the linear assignment problem for $\widetilde{B}_{\Bar{\mathcal{S}}, \mathcal{S}}^\top \widehat{P}_{\Bar{\mathcal{S}}, \mathcal{S}}$ as in Eq.~\eqref{eq:lin_assign_prob}.
 \State \textbf{Output:} Permutation matrix $\widehat{Q} \in \Pi_n$.}
\EndProcedure
\end{algorithmic}
\end{algorithm}

\section{Theory}\label{theory}

In this section, we study the theoretical properties of the estimator obtained from the optimization problem in Equation~\eqref{eq: QuadOptim} under the generalized linear model in Equation~\eqref{eq:glm-for-B}. Our proposed algorithm for graph matching relies on the estimator of $P$. Hence, in the following, we first investigate the estimation error 
based on $\widehat{P}$ using only the seeds. This is followed by an analysis of the accuracy of the inferred permutation matrix obtained as a solution of the quadratic assignment problem. We provide conditions under which the estimated permutation matches all the vertices correctly with high probability and study the implications of this result in the Erd\H{o}s-R\'enyi model.

\subsection{Model estimation error}\label{linearmodel}

In order to solve our graph matching problem, it is essential to obtain an accurate estimate of $P$. Here, we study the error of the estimator $\widehat{P}$ defined in Equation~\eqref{eq:estimatorP}, by providing a probabilistic upper bound on the difference between the entries of $\widehat{P}$ and $P$. 
Recall that, the estimate $\widehat{P}$ is constructed based on $\widehat{\theta}$ which is in turn constructed using the information corresponding to the seeds $\mathcal{S}$. 
In order to obtain a bound on the estimation error for $P$, we consider some assumptions on the Hessian matrix $\nabla^2 L(\theta^\ast)\in\real^{d\times d}$, which is given by
\begin{equation*}
    \nabla^2 L(\theta^\ast)  = \sum_{i,j\in\mS, i>j}\mu'(X_{ij}^\top \theta^\ast) X_{ij}X_{ij}^\top.
\end{equation*}
We also introduce a quantity $\rho_s\in(0,1)$ to measure the variance of the subgraph defined by the edges within the seeded vertices, defined as
$$\rho_s = \frac{2}{s(s-1)} \sum_{i,j\in \mS, i>j} P_{ij}(1-P_{ij}).$$
\begin{assumption}\label{assump:eigenval-hessian}
    The eigenvalues of the Hessian matrix $\nabla^2L(\theta^\ast)$ are positive and bounded away from zero. Specifically, there exists a positive constant $C$
    such that $\lambda_{\min}(\nabla^2 L(\theta^\ast)) \geq C L_s s^2$, where $\lambda_{\min}(\cdot)$ denotes the minimum eigenvalue for a given matrix, $s$ is the cardinality of $\mathcal{S}$ and $L_s$ satisfies $L_s \geq \rho_s$.
\end{assumption}

Essentially, Assumption~\ref{assump:eigenval-hessian} ensures that the collinearity between the network $A$ and its covariates, or between the covariates themselves, is not too high.
In addition, we assume that the magnitude of the covariate vectors $X_{ij}$ is bounded, which allows to obtain a uniform bound over all the entries of $\widehat{P}$.

 \begin{assumption}\label{assump:Bdd-norm-covariates}
    The Euclidean norm of each $X_{ij}\in\real^{d}$, $i>j$ 
    is bounded. That is, there exists some positive constant $C_d$, dependent on the number of covariates $d$, such that $\left\|X_{ij}\right\| \leq C_d$ for all $i,j\in[n], i>j$. 
\end{assumption}

Additionally, we impose certain conditions on the inverse link function of the generalized linear model. These conditions are sufficient to ensure that the loss function satisfies a local strong convexity property around $\theta^\ast$. The following assumption is mild, and it is satisfied by several models, including linear ($\mu(t) = t$) and logistic ($\mu(t) = e^{t}/(1+e^{t})$).

\begin{assumption}\label{assumption:linkfunction}
The inverse link function $\mu = g^{-1}$ is strictly convex ($\mu'(t)>0$ for all $t\in\real$), and has bounded first and second derivatives, that is, there exist $M_1, M_2\in\real$ such that 
    $|\mu'(t)| \leq M_1$ and $|\mu''(t)| \leq M_2$ for all $t$.
\end{assumption}

As the estimator of the regression coefficient vector $\widehat{\theta}$ is constructed using only the information contained on the edges between the seeds in $\mathcal{S}$, the estimation error of $\widehat{P}_{ij} = g^{-1}(X_{ij}^\top \widehat\theta)$ depends on the coefficient vector $\widehat{\theta}$ obtained using the seeded set $\mathcal{S}$. Naturally, we would like the estimation of $P_{ij}$ to be  accurate in order to solve the graph matching problem. 
The following theorem ensures that
error bound for the maximum absolute error $\left| \widehat{P}_{ij} - P_{ij}\right|$ over all edge pairs $i,j \in [n]$ is of the order $O_{\mathbbm{P}}({\sqrt{\log s}}/({L^{1/2}_ss}))$, which allows to control the estimation accuracy as a function of the number of seeds.

\begin{theorem}\label{thm:estimation-error-P}
    Given two graphs $A$, $\widetilde{B}$, the associated edge and node covariate information $Y, Z$, and a set of seed nodes $\mathcal{S}$, let $\widehat{P}$ be the estimator defined in Equation~\eqref{eq:estimatorP}. Under Assumptions~\ref{assump:eigenval-hessian}, \ref{assump:Bdd-norm-covariates} and \ref{assumption:linkfunction}, there exists constants $c_1$, $c_2$ and $c_3$ that depend on the constants in the assumptions, such that, if $s^2L_s \geq c_3 \log s$, then with probability at least $1 - c_2 s^{-1}$,
\begin{equation}\label{eq: thm1statement}
\max_{i, j\in [n]: j>i} \left| \widehat{P}_{ij} - P_{ij} \right| \leq  c_1  \frac{\sqrt{\log s}}{L_s^{1/2} s}.
\end{equation}
\end{theorem}

\subsection{Exact graph matching recovery for the QAP method}\label{sec:Recovglobal}

In order to solve the graph matching problem, our goal is to find the true permutation matrix $Q^\ast$ that aligns the vertices across the graphs, according to Equation~\eqref{eq: QuadOptim-Expectation}. Here, we show that the solution of the problem in Equation~\eqref{eq: QuadOptim} ensures that, when the signal is strong enough, $\widehat{Q} = Q^\ast$ with high probability, implying that all the vertices are matched correctly. This is equivalent to saying that $\|\widehat{P} - Q^\ast \widetilde{B} {Q^\ast}^\top\|_F$ is smaller than $\| \widehat{P} -  Q^\prime \widetilde{B} Q^{\prime^\top}\|_F$ for any $Q^\prime \neq Q^*$ with $Q'_{ii}=1$ for $i\in\mS$. Thus, in order to quantify the signal, we define some quantities related to the population version of the loss function. 

For a given $k\geq 2$, let $\Pi_{n,k}$ be the set of permutation matrices that permute exactly $k$ rows/columns. Given $Q\in\Pi_{n,k}$,  we define $M_P(Q) = \| P - QPQ^\top \|_F^2$ to reflect how different $P$ is from any permuted version of it. If $\mathcal{K} = \{i\in[n]: Q_{ii}=0\}$ is the set of node labels permuted by $Q$, denote by
    $\mathcal{Q}\subset [n]\times[n]$ to the set of edge pairs $(i,j)$ permuted by $Q$, that is,
\begin{equation}\label{eq:set-Q}
\mathcal{Q} = \{(i,j)\in[n]\times[n]: i\in\mathcal{K}\text{ or }j \in \mathcal{K}, i>j\}.
\end{equation}
Observe that $|\mathcal{Q}| = O(nk)$. The following two assumptions are introduced to ensure enough signal in the entries of $P$ to achieve correct graph matching.
\begin{assumption}\label{assump:Sparsity} For every $k\geq 2$ and $Q\in\Pi_{n,k}$, with $Q_{ii}=1$ for $i\in\mS$,
       $$ \sum_{(i,j) \in \mathcal{Q}} P_{ij} \geq c_1^* k \operatorname{log}\left(n-s\right),$$ 
     where $c_1^\ast>0$ is a constant.
\end{assumption}
Since admissible permutations fix all $s$ seed vertices, a permutation moving exactly $k$ vertices must satisfy $k \leq n-s$, although $k$ may exceed $s$ and the number of such permutations is at most $\binom{n-s}{k}k!\le(n-s)^k$. 
Assumption~\ref{assump:Sparsity} ensures that the graph $B$ is sufficiently dense. In particular,  Assumption~\ref{assump:Sparsity}holds when most of the entries of $P$ satisfy $P_{ij} \geq c\,\frac{\log (n-s)}{n}$ for some constant $c$, which is a mild assumption.

\begin{assumption}\label{assump:Diff-P-permutedP} 
For every $Q\in\Pi_{n,k}$, $k\geq 2 $, with $Q_{ii}=1$ for $i\in\mS$,
\begin{equation}\label{eq:assumption-mpq}
    M_P(Q) = \| P - QPQ^\top \|_F^2 \geq c_2^\ast \max \left\{  k \log (n-s), \frac{\sqrt{\log s}  }{L_s^{1/2} s}\sum_{(i, j) \in \mathcal{Q}} P_{ij} \right\},
\end{equation}
where $c_2^\ast>0$ is a constant.
\end{assumption}

 Assumption~\ref{assump:Diff-P-permutedP} specifies a lower bound for the dissimilarity  between $P$ and a permuted version $QPQ^\top$ to ensure that the loss function is able to distinguish $Q$ from $Q^\ast$. The first term on the right-hand side of Equation~\eqref{eq:assumption-mpq} depends only on the model matrix $P$. Similar assumptions have been used in the graph matching setting with no covariates, and are shown to hold under different random graph models for $A$ \citep{arroyo2021maximum}. The second term depends on the model estimation error from Theorem~\ref{eq: thm1statement}, 
 This term also gives an idea about the number of seeds required to achieve accurate graph matching according to the  signal strength involving $P$ on the left-hand side of \eqref{eq:assumption-mpq}.
 In the next section, we verify this assumption under the popular  Erd\H{o}s-R\'enyi random graph model. 

The following theorem provides a consistency guarantee for the solution of the problem in Equation~\eqref{eq: QuadOptim}. 

\begin{theorem}\label{thm:consistency}
Given two graphs $A, \widetilde{B}$ of size $n$, the associated edge and node covariate information Y, Z, and a set of seed nodes $\mathcal{S}$, let $\widehat{Q}$ be the estimator obtained from the optimization problem in \eqref{eq: QuadOptim}.  Under the assumptions of Theorem~\ref{thm:estimation-error-P}, in addition to  Assumptions~\ref{assump:Sparsity} and \ref{assump:Diff-P-permutedP}, we have that
\begin{equation}
    \mathbbm{P}(\widehat{Q} = Q^*|A,Y,Z) \geq  1 - c \left(\frac{1}{(n-s)^2} + \frac{1}{s}\right),
\end{equation}
where $c>0$ is a constant that depends on the constants in the assumptions.
\end{theorem}

Theorem~\ref{thm:consistency} presents the conditions required for the exact estimation of the  permutation matrix $Q^\ast$. This is equivalent to saying that the number of incorrectly matched vertices is zero with high probability, provided that the signal as quantified in Assumptions~\ref{assump:Sparsity} and \ref{assump:Diff-P-permutedP} is strong enough. 

In order for the method to consistently recover the correct permutation, the number of seeds $s$ needs to grow with $n$. The exact rate at which $s$ needs to grow depends on the condition in Assumption~\ref{assump:Diff-P-permutedP}, which depends on the estimation error. In the sparse regime where most of the entries of $P$ satisfy $P_{ij}\geq c \log (n-s)/n$, a rate of $s\asymp \sqrt{n}$ reduces the right-hand side of Equation~\eqref{eq:assumption-mpq} to $c^\ast_2 k \log n$. We study this result in detail under a specific model for $P$ in Section~\ref{sec:ERmodel}.

\subsection{Erd\H{o}s-R\'enyi model with covariates}\label{sec:ERmodel}

The theoretical guarantees for model estimation error and exact recovery rely on several assumptions. Validation of such assumptions, in particular, Assumption~\ref{assump:Diff-P-permutedP}, is non-trivial, with some existing results available under some random graph models with no covariates \citep{lyzinski2015graph,arroyo2021graph}. In this section, we consider the specific case of the Erd\H{o}s-R\'enyi random graph model for $A$. Recall that for a graph with adjacency matrix $A$, we say that $A$ is distributed according to the ER model, denoted by $A\sim\text{ER}(p)$ with $p\in(0,1)$, if $\mathbbm{P}(A_{ij}=1) = p$ independently for all $i>j$.

We consider the ER model with probability $p_n$ for the first graph $A\sim \text{ER}(p_n)$. In addition, we model a single edge covariate $Y\in\real^{n\times n}$ as a random binary variable independent of $A$ with probability $q_n$, in which case $Y\sim \text{ER}(q_n)$. For simplicity, we consider a linear model with link function $g(t) = t$. Following the model in Equation~\ref{eq:glm-for-B}, the distribution of the second graph satisfies
\begin{align*}
    P_{ij} = \theta_0^\ast + \theta_1^\ast A_{ij} + \theta_2^\ast Y_{ij} \hspace{5pt} \text{with} \hspace{5pt} B_{ij}|A,Y \sim \text{Ber}(P_{ij}), 
\end{align*}
for some coefficients $\theta_0^\ast, \theta_1^\ast, \theta_2^\ast$.

\begin{proposition}\label{prop: eigenval-hessian}
    For graphs of size $n$, suppose the graph $A \sim \text{ER}(p_n)$ and the edge covariate $Y \sim \text{ER}(q_n)$.
    Then under the conditions (i) $\sqrt{p_n} + \sqrt{q_n} \leq 1 - {\nu}$ for some ${\nu} >0$, (ii) $\max \{p_n, q_n\} = o\left( s^2 \min\{p^2_n, q^2_n\} \right),$ and (iii) $\theta_0^\ast + \theta_1^\ast p_n + \theta_2^\ast q_n = O(\min\{p_n, q_n\}),$
Assumption~\ref{assump:eigenval-hessian} holds with high probability as $n$ increases. In particular, for some constant $c>0$, we have
\begin{align*}
      \mathbbm{P}\left(\lambda_{\min}\left(\nabla^2 L(\theta^*) \right) \geq c s^2 \min\{p_n, q_n\}\geq s^2 \rho_s \right) = 1 - o(1).
\end{align*}
\end{proposition}

Proposition~\ref{prop: eigenval-hessian} illustrates that the eigenvalues of $\nabla^2 L(\theta^*)$ are positive and bound away from zero with high probability under some conditions of the Erd\H{o}s-R\'enyi parameters and the size of the seed set $s$. This proposition ensures the local strong convexity of the Hessian matrix. Condition (i) is introduced for technical reasons, and it is mild, as it only excludes some dense graph regimes. Conditions (ii) and (iii) are introduced to ensure that the magnitudes of the parameters in the regression model are relatively similar. In particular, condition (ii) implies that $s^2\min\{p_n, q_n\}\to\infty$. Next we proceed to verify  Assumption~\ref{assump:Diff-P-permutedP}.

\begin{proposition}\label{prop:Bd-Diff-P-permutedP}
    Define $\tau_n = \min\{p_n, q_n\}$, and suppose that the following conditions hold along with condition (i) of Proposition~\ref{prop: eigenval-hessian}:
     $$(i)\ \theta_0^\ast + \theta_1^\ast p_n + \theta_2^\ast q_n\geq \alpha_1\frac{\log (n-s)}{n}, \quad\quad \quad (ii) \ \frac{\tau_n^2}{p_n + q_n} \geq \alpha_2 \frac{\log (n-s)}{n}$$
     $$(iii )\ {{\theta^*_1}^2} + {{\theta^*_2}^2} \geq \alpha_3 \max\left\{ \frac{\log (n-s)}{\tau_n n}, \frac{\sqrt{\log s}(\theta_0^\ast + \theta_1^\ast p_n + \theta_2^\ast q_n)}{\tau_n^{3/2} s} \right\},$$
     where $\alpha_1, \alpha_2, \alpha_3$ are some positive constants. Then,     Assumption~\ref{assump:Diff-P-permutedP} holds with high probability.  That is, with probability at least $1-o(1)$, we have that for all $Q\in\Pi_n$, $Q_{ii} = 1, i\in\mS$,
$$ M_P(Q) \geq c_2^\ast \max \left\{ k \log (n-s), \frac{ \sqrt{\log s} }{\tau_n^{1/2} s} \sum_{(i, j) \in \mathcal{Q}} P_{ij} \right\}.$$
       
\end{proposition}

Now using Propositions~\ref{prop: eigenval-hessian} and \ref{prop:Bd-Diff-P-permutedP}, we verify Assumptions~\ref{assump:eigenval-hessian} and \ref{assump:Diff-P-permutedP} respectively for the case of Erd\H{o}s-R\'enyi model, required for the graph matching consistency as mentioned in Theorem~\ref{thm:consistency}. The rest of the assumptions can be satisfied by appropriately choosing the link function, model parameters and normalizing the covariates. In the following theorem, we ensure the consistency holds for the Erd\H{o}s-R\'enyi case.   
\begin{theorem}\label{thm:ERconsistency}
   Suppose that the conditions of Propositions~\ref{prop: eigenval-hessian} and \ref{prop:Bd-Diff-P-permutedP} hold.
   Then, under the assumption $n \gg s$ and as $n,s$ go to infinity, we have that
    \begin{equation*}
    \mathbbm{P}(\widehat{Q} = Q^*) = 1 - o(1).
\end{equation*}
\end{theorem}

The conditions required by Theorem~\ref{thm:ERconsistency} impose constraints in the relationship between the number of seeds and the edge probabilities for $A$, $Y$ and $B$. When $s\ll n$, $p_n \asymp q_n$ and $\theta_1^\ast \asymp \theta_2^\ast\asymp 1$, the conditions reduce to $s^2 p_n\to \infty$, $\theta_0^\ast = O(p_n)$ and $\frac{\log n}{n} \lesssim p_n$, which are mild. For instance, in the sparse regime when $p_n \asymp \frac{\log n}{n}$, the number of seeds required by the theorem to achieve exact recovery is of the order $\sqrt{n}$, whereas in the dense regime with $p_n \asymp 1$, a growing number of seeds $s\to\infty$ is enough.

\subsection{{Exact graph matching recovery for the local neighborhood method}}\label{sec:Recovlocal}
In this section, we establish an exact recovery guarantee for the local neighborhood matching method presented in Algorithm~\ref{alg2}. As opposed to the QAP method for which a computationally efficient algorithm to solve the problem does not exist, the local neighborhood method is computationally fast, yielding a theoretical guarantee for a computable solution. In particular, we show that the solution to the problem in Equation~\eqref{eq:lin_assign_prob} satisfies $\widehat{Q}=Q^\ast$ with high probability when the signal is sufficiently strong. This implies that all non-seed vertices are matched correctly.

For the loss function in the linear assignment problem on Equation \eqref{eq:lin_assign_prob}, exact recovery is achieved if $
\left\|\widehat{P}_{\bar{\mS}, \mS} - Q^\ast \widetilde{B}_{\bar{\mS}, \mS} \right\|_F <
\left\| \widehat{P}_{\bar{\mS}, \mS} - Q^\prime \widetilde{B}_{\bar{\mS}, \mS}\right\|_F
$ for any $Q^\prime \neq Q^\ast$ where $Q^\prime \in \Pi_{n-s,k}$ and any $k\geq 2$. Thus, we show that this event holds with high probability as the nodes and the size of the seeds increase. Under Assumptions~\ref{assump:Sparsity} and \ref{assump:Diff-P-permutedP} which are also used to establish global consistency, the following theorem establishes consistency of the solution to the optimization problem in Equation~\eqref{eq:lin_assign_prob}.

\begin{theorem}\label{thm:consistency_local}
Given two graphs $A, \widetilde{B}$ of size $n$, the associated edge and node covariate information Y, Z, and a set of seed nodes $\mathcal{S}$, let $\widehat{Q}$ be the estimator obtained from the optimization problem in \eqref{eq:lin_assign_prob}.  Underthe assumptions of Theorem~\ref{thm:estimation-error-P}, in addition to  Assumptions~\ref{assump:Sparsity} and \ref{assump:Diff-P-permutedP}, we have that
\begin{equation}
    \mathbbm{P}(\widehat{Q} = Q^*|A,Y,Z) \geq  1 - c_\ell \left(\frac{1}{(n-s)^2} + \frac{1}{s}\right),
\end{equation}
where $c_\ell>0$ is a constant that depends on the constants in the assumptions.
\end{theorem}

Theorem~\ref{thm:consistency_local} gives sufficient conditions for exact recovery of the permutation matrix $Q^\ast$. Equivalently, the number of incorrectly matched non-seed vertices is zero with high probability when the signal quantified by Assumptions~\ref{assump:Sparsity} and \ref{assump:Diff-P-permutedP} is sufficiently strong.


\section{Simulations}\label{sim}
In this section, we present simulated experiments and compare the results obtained by our methods with alternative approaches for graph matching. The simulation study considers a comparison of five different methods implemented. First, our proposed methods that incorporate covariates, Algorithm~\ref{alg1} based on the QAP (``CovQAP'') and Algorithm~\ref{alg2} based on the neighborhood of the seed nodes (``CovNeigh''). We also consider the analog methods that do not incorporate covariates, namely, a non-convex relaxation of the QAP method (``NoCovQAP'') \citep{vogelstein2015fast}, and the LAP method using the neighborhood of the seeded nodes (``NoCovNeigh'') \citep{lyzinski2014seeded,mossel2020seeded}. Finally, we combine the network $A$ and the edge covariate matrix $Y$ 
by constructing a similarity matrix $M =\frac 12 (A+Y)$, and we match this matrix to $\widetilde{B}$ via the non-convex relaxation of the QAP (``AvgSim''). In our implementation, for the QAP-based algorithms, we use the implementation of the Frank-Wolfe methodology from the R package \texttt{iGraphMatch} \citep{qiao2021igraphmatch} to optimize the nonconvex loss function.

In all simulation scenarios, we generate the first graph following an  Erd\H{o}s-R\'enyi model, that is, $A \sim G(n, p)$, with $n = 500$ being the number of vertices and $p = 0.1$ being the probability of edge inclusion. We consider a binary $(0/1)$ $n\times n$ edge covariate matrix $Y$, similarly generated independently from the Erd\H{o}s-R\'enyi model, that is $Y \sim G(n, q)$ with $q = 0.1$. We then sample the adjacency matrix of the second graph $B$ following a Bernoulli distribution from the corresponding probability matrix $P$ given by $P = \theta_0 + \theta_1  A + \theta_2 Y$  for real fixed values of $\theta_0, \theta_1$ and $\theta_2$, where $\theta_1$ and $\theta_2$ depend on parameters $\alpha$ and $\gamma$, representing the signal strength and the amount of covariate effect, respectively, defined as follows. For Figures~\ref{fig:easysetup100} and \ref{fig:difficultsetup100} we have $\theta_0 = 0.01, \theta_1 = \alpha(1-\gamma)$ and $\theta_0 = 0.6, \theta_1 = -\alpha(1-\gamma)$ respectively, while $\theta_2 = \alpha \gamma$ is the same for all figures. Given a fixed number of seeds and seed set $\mathcal{S}$, we construct $\widetilde{B}$ keeping the seeded vertices untouched and shuffling the remaining vertices, $n - |\mathcal{S}|$ of $B$ through a permutation matrix $Q^*$.

Using the above method, we generate the covariate matrix $Y$ and the observed graphs $A$, $\widetilde{B}$ for some fixed values of $\theta_0, \alpha, \gamma$ and number of seeds $|\mathcal{S}|$.
For each of these combinations, the experiment is repeated $50$ times and the average proportion of incorrectly matched vertices between $A$ and $\widetilde{B}$ is reported. We report the mean proportion of mismatches considering $\gamma \in\{0.05, 0.45, 0.85\}$ and for each of the $\gamma$ values, we provide the plot of proportion of mismatches as a function of $\alpha \in\{0.1, 0.15, \cdots, 0.55, 0.6\}$. This encompasses a whole range of situations, including different degrees of influence of the observed matrix $A$ and covariate matrix $Y$. We also repeat the simulations with two choices of seeds, $\{100, 250\}$, to study the combined interplay of the seed numbers and covariates. 

  
  
  

\begin{figure}[ht!]
  \centering
  
  \begin{subfigure}{0.8\textwidth}
    \centering
    \includegraphics[width=\textwidth]
    {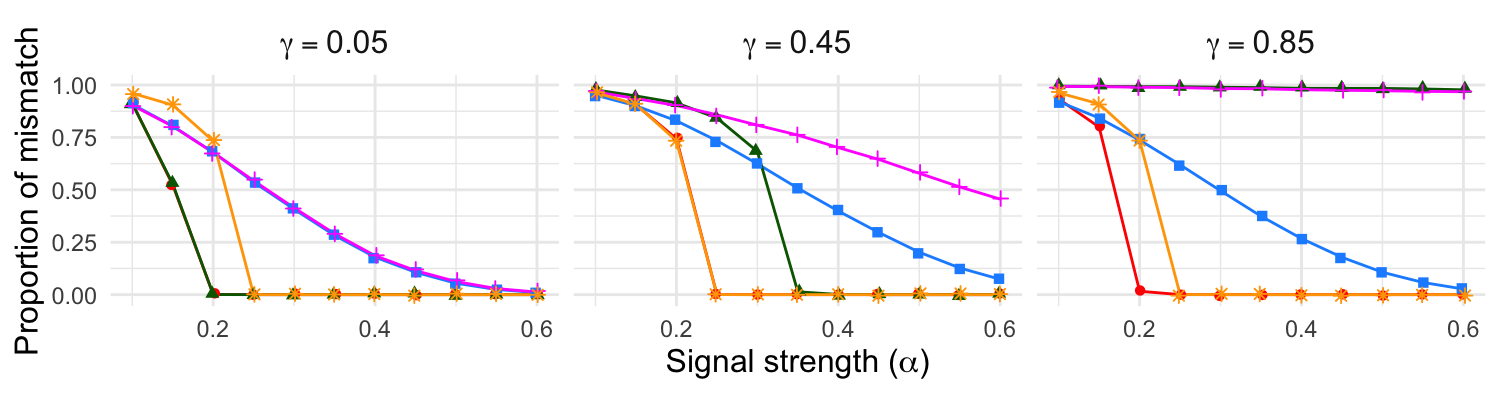}
    \caption{Both the network and the covariate have an effect in the same direction.}
     \label{fig:easysetup100}
  \end{subfigure}
  
  \begin{subfigure}{0.8\textwidth}
    \centering
    \includegraphics[width=\textwidth]{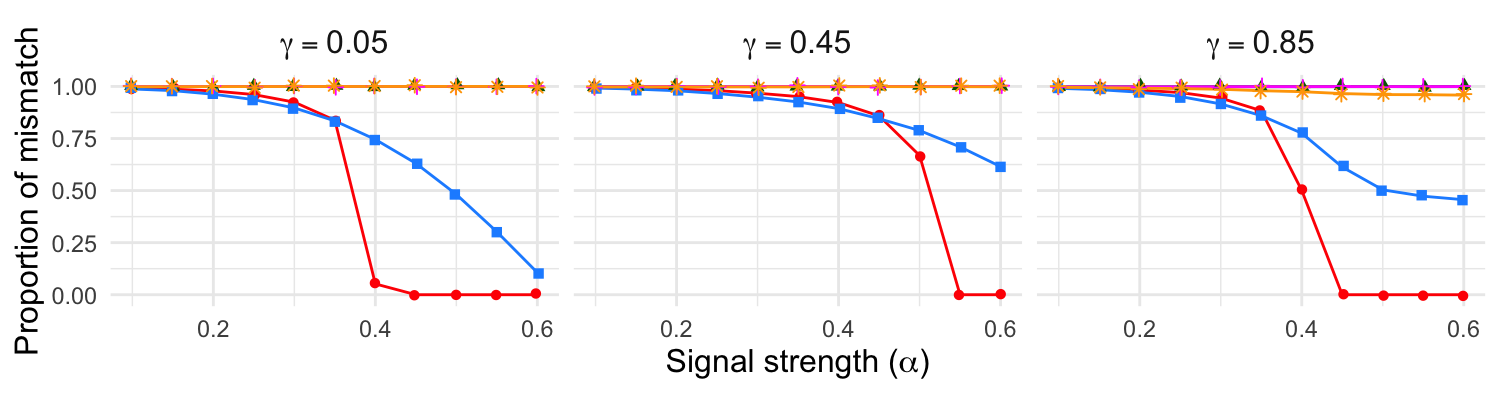}
    \caption{The given network and the covariate have an effect in the opposite direction.}
     \label{fig:difficultsetup100}
  \end{subfigure}
  
  \caption{Comparison of different graph matching algorithms with varying covariate effect $(\gamma)$ and signal strength $(\alpha)$. 
  Using $100$ seeds out of $500$ vertices.}
  \label{fig:setup100}
\end{figure}

In Figure~\ref{fig:easysetup100}, we explore a reasonably easy setup to demonstrate situations where the competing methods are likely to work well. As expected, the proportion of mismatched vertices is high for very low values of $\alpha$ because of the lack of signal. This is the case where $P$ behaves like a constant matrix ($P_{ij} \approx \theta_0$ for all $i,j$) and all the methods fail due to a lack of information in common between $P$ and $\widetilde{B}$. When $\gamma$ is low, the graph $A$ has the most influence, and the no covariate methods perform almost as good as our proposed algorithms. But as $\gamma$ increases, the signal due to the covariates becomes more significant and, as expected, the no covariate methods worsen their performance in such situations. The method using the average similarity matrix performed reasonably well in all scenarios, as the covariate and the network have a similar effect, which is captured by $M$. The algorithms based on the seed node neighborhood, CovNeigh and NoCovNeigh, performed worse than their QAP counterparts.

  
  
  

  
  
  

Figure~\ref{fig:difficultsetup100} considers a relatively more difficult setup for matching, where the network and the covariate have opposite effects on $\widetilde{B}$. Compared to Figure~\ref{fig:easysetup100}, our proposed methods show a similar behavior, whereas the alternative approaches 
For the low $\gamma$ case, the degree of influence of $A$ dominates the information from the covariates. In such cases, 
our proposed methods gradually achieved zero proportion of error in matching as the $\alpha$ values increased. However, in contrast with Figure~\ref{fig:easysetup100}, the alternative approaches are not able to achieve a good matching performance. In fact, the QAP and the neighborhood-based methods with no covariates perform no better than random guessing, as the network $A$ has a negative effect in this scenario. Meanwhile, AvgSim can detect some signal when the effect of the adjacency matrix $A$ is low. 
A larger number of seeded nodes helps achieve better performance for our methods. Especially, the improvement of the performance is noticed for Algorithm~\ref{alg2}, which is expected, as the method primarily utilizes the information from the seeded nodes in the neighborhood of the non-seeded ones. Refer to Supplementary material for results with $250$ seeds.

\begin{figure}[ht!]
  \centering
  
  \begin{subfigure}{0.8\textwidth}
    \centering
    \includegraphics[width=\textwidth]{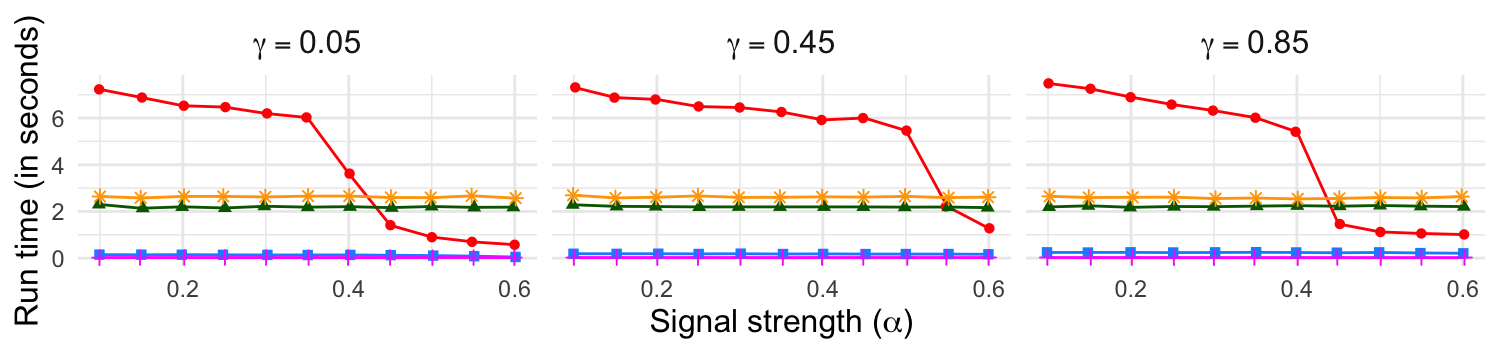}
    \caption{Time comparison under varying covariate effect $(\gamma)$ and signal strength $(\alpha)$.}
    \label{fig:timeSignal}
  \end{subfigure}
  
  \begin{subfigure}{0.8\textwidth}
    \centering
    \includegraphics[width=\textwidth]{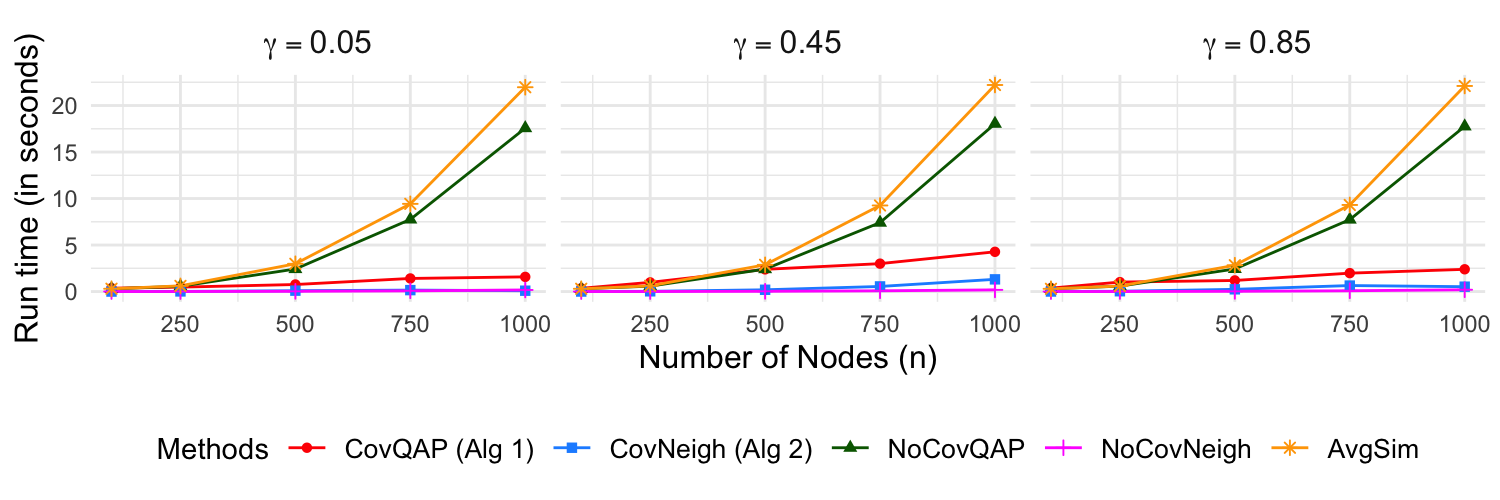}
    \caption{Time comparison under varying covariate effect $(\gamma)$ and graph size $(n)$.}
    \label{fig:timeNodes}
  \end{subfigure}
  
  \caption{Comparison of different graph matching algorithms with respect to run time (in seconds) with varying covariate effect $(\gamma)$.}
  \label{fig:time}
\end{figure}

Next, we present run-time comparison results measured in seconds. Figure~\ref{fig:timeSignal} shows the run time for different graph matching algorithms under varying covariate effects and signal strength. This result is obtained under the same setup as in Figure~\ref{fig:difficultsetup100}. Among the methods, the neighborhood-based methods are orders of magnitude faster than the QAP-based ones, but with a worse performance in terms of matching accuracy as seen in Figure~\ref{fig:difficultsetup100}. Although CovQAP performs the worst in terms of run time under low signal, the speed increases significantly as the signal increases, and it provides better matching accuracy overall. The second comparison examines run time with increasing graph size, where $20\%$ of the nodes are used as seeds. Figure~\ref{fig:timeNodes} presents the results for varying covariate effect and number of nodes with signal strength fixed at $\alpha = 0.55$. The other parameter setup is the same as Figure~\ref{fig:difficultsetup100}. As expected, run time grows with graph size. Our proposed methods, along with NoCovNeigh, run faster than NoCovQAP and AvgSim. Among our methods, CovNeigh is the quickest, performing orders of magnitude faster than CovQAP, trading off accuracy as noted before in Figures~\ref{fig:easysetup100} and \ref{fig:difficultsetup100}.

\section{Matching collaboration and genealogy networks}\label{dataApp}

In this section, we consider the practical application of matching individual records of researchers in statistics coming from two different data sources. The Multi-Attribute Dataset on Statisticians (MADStat) \citep{DVN_V7VUFO_2022} contains records of statistics research papers from 1975 to 2015 published on 36 different journals, from which we obtained a collaboration network between the more than 40,000 authors in the data. In addition, we used data from the Mathematics Genealogy Project (MGP) (\url{https://www.genealogy.math.ndsu.nodak.edu/}) to reconstruct the academic genealogy (doctoral student-advisor relationships) of a subset of statisticians from the MADStat data.  
Thus, the matrix $A$, representing the genealogy network, has entries $A_{ij}=1$ if $i$ is the doctoral advisor of $j$ or viceversa, and 0 otherwise. The matrix $\widetilde{B}$ encodes collaboration relationships, with $\widetilde{B}_{ij}=1$ indicating that the authors $i$ and $j$ have written a paper together. We focus on a subset of statisticians who appear both in the MADStats and MGP data with known correspondence across the two datasets.  This subset, containing $n = 8,627$ nodes,  is selected based on the largest connected component of the collaboration network once the nodes with known correspondence across the two datasets are considered. 
Figure~\ref{fig: networksAB} (see Supplementary Material) shows a visualization of the two networks.

In addition to the genealogy information, the MGP dataset also contains covariates related to the doctoral studies of each node, such as graduation year and institution, that we use in our methodology for covariate-assisted graph matching. We consider four  covariate matrices $Y^{(i)} \in \mathbbm{R}^{n \times n}, i = 1, 2, 3, 4$, derived from node covariates and the network information as follows. For a pair of nodes $i,j$,  $Y^{(1)}_{ij}\in\{0,1\}$ indicates if $i$ and $j$ graduated from the same institution, $Y^{(2)}_{ij} \in\{0,1\}$ indicates if the graduate institutions of $i$ and $j$ are situated in the same country, $Y^{(3)}_{ij} \in\{0,1\}$ indicates if $i$ and $j$ have two degrees of academic separation, that is they are academic siblings or grandparent-grandchild ($A_{ik}A_{jk}=1$ for some $k$), and  $Y^{(4)}_{ij}$ is the absolute difference between the graduation years of $i$ and $j$.

Our proposed methods require the estimation of a prediction model for $B$ conditional on the covariates and the network $A$. To estimate ${P}$ as in Equation~\eqref{eq:estimatorP},  we use a logistic regression model with an intercept, the genealogy network and the above covariates. Table~\ref{FittedCoeff} shows the fitted coefficients in one replication of the experiments, with $500$ non-seeds selected at random; we observed that the estimated coefficients and the corresponding standard errors remain similar in other replications or with different numbers of seeds. As expected, advisor-student relationship is the strongest predictor of collaboration, and we obtained positive estimated coefficients for $Y^{(1)}$, $Y^{(2)}$ and $Y^{(3)}$, suggesting that individuals that graduated from the same institution or the same country, or being closely related in the academic family, have higher chances of collaboration compared to individuals that do not share these features. Meanwhile, a large academic age gap in terms of the difference between the graduation years decreases the chances of collaboration.

\begin{table}[ht]
\centering
\begin{tabular}{|clcc|}
\hline
\multicolumn{1}{|c|}{Term}     & \multicolumn{1}{c|}{Variable description}   & \multicolumn{1}{c|}{\begin{tabular}[c]{@{}c@{}}Coefficient\end{tabular}} & \begin{tabular}[c]{@{}c@{}}Std. error\end{tabular} \\ \hline
\multicolumn{1}{|c|}{Intercept} & \multicolumn{1}{l|}{Intercept term.}   & \multicolumn{1}{c|}{-7.986}   & 0.017   \\ \hline
\multicolumn{1}{|c|}{$A$}    & \multicolumn{1}{l|}{Advisor-student indicator (genealogy network).}  & \multicolumn{1}{c|}{8.703}   & 0.031 \\ \hline
\multicolumn{1}{|c|}{$Y^{(1)}$}        & \multicolumn{1}{l|}{\begin{tabular}[c]{@{}l@{}} Same institution of Ph.D. graduation indicator. 
\end{tabular}}          & \multicolumn{1}{c|}{1.105}   & 0.041  \\ \hline
\multicolumn{1}{|c|}{$Y^{(2)}$}        & \multicolumn{1}{l|}{\begin{tabular}[c]{@{}l@{}} Same country of Ph.D. graduation indicator. \end{tabular}}  & \multicolumn{1}{c|}{0.746}   & 0.019  \\ \hline
\multicolumn{1}{|c|}{$Y^{(3)}$}        & \multicolumn{1}{l|}{\begin{tabular}[c]{@{}l@{}} Academic siblings or grandparent/child).\end{tabular}} & \multicolumn{1}{c|}{3.276}  & 0.042   \\ \hline
\multicolumn{1}{|c|}{$Y^{(4)}$}   & \multicolumn{1}{l|}{\begin{tabular}[c]{@{}l@{}} Absolute difference between graduation years.
\end{tabular}}  & \multicolumn{1}{c|}{-0.028} & 0.001  \\ \hline
\end{tabular}
\caption{Estimated coefficients, descriptions along with standard error to obtain $\widehat{P}$. These coefficients were obtained by fitting the logistic regression model to one realization of the experiments with 500 non-seed nodes.}
\label{FittedCoeff}
\end{table}

\begin{table}[ht!]
\begin{center}
\begin{tabular}{ |c|c|c|c|c| } 
 \hline
 No. of non-seeds & CovQAP (Alg 1) & NoCovQAP  & CovNeigh (Alg 2) & NoCovNeigh  \\ 
 \hline
 $500$ & $\textbf{65.76} (2.47)$ & $61.82 (2.64)$ & $64.45 (2.32)$ & $59.28 (2.40)$ \\
 \hline
 $1000$ & $\textbf{58.87} (1.80)$ & $55.30 (1.65)$ & $55.48 (1.61)$ & $51.22 (1.72)$ \\
 \hline
  $2000$ & $\textbf{48.66} (1.29)$ & $46.25 (1.27)$ & $43.08 (1.18)$ & $39.60 (1.09)$ \\
  \hline
  $4000$ & $34.08 (0.70)$ & $32.74 (0.79)$ & $26.60 (0.55)$ & $24.06 (0.63)$ \\
 \hline
\end{tabular}
\end{center}
\caption{Average matching accuracy (along with the standard deviation), measured by the percentage of non-seed nodes correctly matched by the corresponding algorithm over 50 replications with random seeds.}
\label{RealDataTab}
\end{table}

We compare the performance of four different graph matching methods considered in the previous section, namely, our proposed algorithms (CovQAP and CovNeigh), and their analog versions for the no-covariate case (NoCovQAP, NoCovNeigh). The performance is measured by randomly selecting seed nodes and calculating the percentage of non-seed nodes that are correctly matched by the algorithm. We consider $500, 1000, 2000$ and $4000$ as the number of non-seed nodes, and the rest of the nodes are considered as seeds. In Table~\ref{RealDataTab}, we summarize the average percentage of matching along with standard deviation over $50$ replications of random seed nodes, where each time we shuffle the matrix $B$ to get different $\widetilde{B}$ to match with $A$. The results clearly demonstrate the capability of all methods to correctly match a large portion of the nodes in the graphs. Moreover, as expected, the percentage of matches decreases as the number of non-seeds increases. Our proposed methods with covariates show a significant improved performance when compared to the corresponding no covariate versions. 
For illustration purposes, we also include results with $4000$ non-seed nodes to demonstrate performance in a relatively larger matching problem. However, due to the increased computational cost, we use $10$ replications for this setting.

To compare this performance improvement more precisely, we also performed a paired t-test between the percentage of matches obtained using our method and the corresponding no-covariate version. In all cases, the performance of our methods is statistically significantly better. In Table~\ref{TestTab} we summarize the estimated mean difference in the percentage of matched nodes and the corresponding $95\%$ confidence interval. Overall, our algorithms are able to obtain between 
$1.34 \%$ to $5.16$\% more nodes correctly matched than the methods that do not use covariates.


Moreover, our proposed algorithms are flexible to implement alternative predictive machine learning methods to obtain the estimates for the entries in $\widehat{P}$ as in Equation~\ref{eq:estimatorP}, which is then used by the proposed algorithms. For the purpose of illustration, we consider a random forest classifier \citep{breiman2001random} in a single replication with $500$ and $1000$ non-seed nodes. We use the set of seed nodes to fit the method, and create the matrix $\widehat{P}$ from the predicted values, which is then used in the remaining steps of our algorithms. 

\begin{table}[ht!]
\centering
\begin{tabular}{|c|c|c|c|}
\hline
No. of non-seed nodes & No. of trees & CovQAP (Alg 1) & CovNeigh (Alg 2) \\ \hline
$500$ & $300$ & $63.27$   & $59.53$ \\ \hline
$1000$  & $600$ & $54.3$   & $49.9$  \\ \hline
\end{tabular}
\caption{Matching accuracy measured by the percentage of non-seed nodes correctly matched by estimating $\widehat{P}$ using random forest and the corresponding proposed algorithms for a single replication.}
\label{tab:RFRealdata}
\end{table}


Using a random forest with the \texttt{ranger} R package  \citep{wright2020ranger}, the percentages of correctly matched nodes by Algorithm~\ref{alg1} and \ref{alg2} are mentioned in Table~\ref{tab:RFRealdata}. While these numbers already show the flexibility of our methodology and the potential of other machine learning methods, by tuning various parameters of the random forest or exploring other classification methods, it might be possible to achieve better matching accuracy, but we leave this exploration as future work.

\section{Discussion}\label{discussion}

In this paper, we addressed the inexact graph matching problem incorporating the use of node and edge covariates, which are often available in practice. Our methods rely on estimating the probability matrix $P$ for one of the graphs based on seeds in a manner similar to regression models with network responses, followed by solving an optimization problem to match the nodes in the entire graph. The conditional framework we introduced extends to a variety of random graph models while effectively incorporating covariate information to improve matching accuracy. Within this framework, we derived bounds for estimation error, established consistency of the estimated latent node correspondences and provided validation for Erd\H{o}s-R\'enyi random graph model. Our proposed methods demonstrated good empirical performance in simulations and real data. In summary, our proposed framework offers a new perspective on utilizing useful covariate information to enhance the graph matching performance across a variety of applications.

There are specific extensions of interest. In this project, we restricted attention to pairs of graphs with the same number of nodes. It is possible to extend the framework to settings where one graph is larger. 
Another potential extension is to consider alternative prediction methods. We have empirically explored the use of random forest, showing promising results. On the theoretical side,  our results are based on the optimal solution of the QAP, which is intractable, but it provides a  baseline to compare with existing work. A comprehensive theoretical study of covariate-assisted graph matching under computationally efficient algorithms is left as future work. These several directions would be interesting and valuable to explore in the future.

The current approach presents certain limitations. First, our methodology assumes the existence of some seed nodes. While these seeds can be manually obtained in many applications, this may not always be possible. An alternative approach that does not require seeds would need to simultaneously optimize over the model fit and node matching (for example, as in \cite{arroyo2021graph}), but this can vastly increase the computational cost. Second, the network regression model used in the paper does not incorporate node latent effects, which are common in the literature \citep{ma2020universal}, as the estimation of these effects requires knowing the node correspondence. Nevertheless, the real data results show that even without these effects the method already shows significant improvements. Finally, while we incorporated covariates within a regression framework, we did not explore alternative strategies, such as string matching, which could be relevant when textual attributes are available for each node.

\section*{Acknowledgments}
Jes\'us Arroyo acknowledges support from the National Science Foundation under grant DMS-2413553. The authors are grateful to students Adam Authur, Gage Sanders, Xingyu Liu and Yufan Li, who contributed to collecting the data. The authors also acknowledge that portions of this research were conducted with the advanced computing resources provided by Texas A\&M Department of Statistics Arseven Computing Cluster.

\bibliographystyle{apalike}
\bibliography{biblio}
\newpage

\appendix

\section{Technical Details}

\subsection{Proof of Proposition~\ref{prop:linearassignmentsol}}

\begin{proof}
Without loss of generality, suppose that $\mS = \{1, \ldots, s\}$ and that $Q^\ast = I_n$, that is, the nodes of the two graphs are correctly aligned. 
In this case, $\widetilde{B} = B$ and we can write
\begin{align*}
\e\left[\Tr\left(\widetilde{B}_{\wbmS, \mS} P^\top_{\wbmS, \mS}\right)\right] & = \sum_{i=s+1}^n \sum_{k=1}^s \e[\widetilde{B}_{ik}]P_{ik}\\
    & = \sum_{i=s+1}^n \sum_{k=1}^s P_{ik}^2.
\end{align*}
Now, let $Q$ be a different permutation matrix, and $\pi:[n-s]\rightarrow[n-s]$ its corresponding permutation function such that $Q_{ij}=1$ iff $\pi(i) = j$. Then,
\begin{align*}
\e\left[\Tr\left(Q\widetilde{B}_{\wbmS, \mS} P^\top_{\wbmS, \mS}\right)\right] & = \sum_{i=s+1}^n \sum_{k=1}^s \e[Q_{(i-s),[n-s]}\widetilde{B}_{\wbmS, k}]P_{ik}\\
    & = \sum_{i=s+1}^n \sum_{k=1}^s \e[\widetilde{B}_{\pi(i-s)+s, k}]P_{ik}\\
    & = \sum_{i=s+1}^n \sum_{k=1}^s P_{\pi(i-s)+s,k}P_{ik}.
\end{align*}
    Observe that the variables $P_{\pi(i-s)+s,k}, i\in\wbmS, k\in \mS$ are a rearrangement of $P_{ik}, i\in\wbmS, k\in \mS$. Hence, by the rearrangement inequality \citep{hardy1952inequalities},
     \begin{equation}\label{eq:proof-prop-linear asgnment}
         \e\left[\Tr\left(\widetilde{B}_{\wbmS, \mS} P^\top_{\wbmS, \mS}\right)\right] \geq   \e\left[\Tr\left(Q\widetilde{B}_{\wbmS, \mS} P^\top_{\wbmS, \mS}\right)\right],
     \end{equation}
    which proves that $Q^\ast = I_n$ satisfies 
    Equation (2.7).
    Moreover, equality only holds if $P_{\pi(i-s)+s,k} = P_{ik}$ for all $i\in\{s+1, \ldots, n\}$ and $k\in[s]$. If $Q\neq I_n$, then there exists some $i'\in[n-s]$ such that $\pi(i') \neq i'$, and if $P_{\wbmS, \mS}$ does not have repeated rows, then there exists some $k$ such that $P_{\pi(i')+s,k} \neq P_{i,k}$, which shows that the inequality in \eqref{eq:proof-prop-linear asgnment} is strict, and $Q^\ast$ is the unique solution.
\end{proof}

\subsection{Proof of Theorem~\ref{thm:estimation-error-P}}
Throughout the proofs, to simplify notation, we use $\widetilde{\mathcal{S}}$ to denote the set of edge pairs between nodes contained within the set of seeds, that is,
$$\widetilde{\mathcal{S}} = \{(i,j)\in[n]^2: i,j\in \widetilde{\mathcal{S}}, i>j\}.$$
Observe that $\widetilde{B}_{ij} = B_{ij}$ for $(i,j)\in \widetilde{\mathcal{S}}$, and that $\{B_{ij}: (i,j)\in \widetilde{\mathcal{S}}\}$ are mutually independent random variables.


Before proving Theorem~3.1, we introduce some lemmas that are used to show that the loss function satisfies a local strong convexity property. Recall that for the loss function defined in Equation~(2.4),
the gradient and the Hessian matrix are given by
\begin{equation*}
\label{eq:gradient-lossfunction}
    \nabla L(\theta) = \sum_{(i,j)\in \widetilde{\mathcal{S}}}-\left((\widetilde{B}_{ij} - \mu(X_{ij}^\top \theta)) X_{ij}\right),
\end{equation*}
\begin{equation}
    \nabla^2 L(\theta)  = \sum_{(i,j)\in \widetilde{\mathcal{S}}}\mu'(X_{ij}^\top {\theta}) X_{ij}X_{ij}^\top.
    \label{eq:Hessian-matrix-definition}
\end{equation}
In addition, by properties of the link function in a generalized linear model, $\mu(X_{ij}^\top \theta^\ast) = \e[\widetilde{B}_{ij}] = P_{ij}$. We also write $P_{ij}(\theta) = \mu(X_{ij}^\top \theta)$ and $\nabla P_{ij}(\theta) = \mu'(X_{ij}^\top\theta)X_{ij}$.

The following lemma shows that the loss function is strongly convex in a neighborhood of $\theta^\ast$. The proof follows a similar truncation argument as Proposition 2 of the long version of the paper by \cite{negahban2012unified}.

\begin{lemma} \label{lemma:glm-restricted-strong-convexity}
    Suppose that the link function $\mu$ satisfies Assumption~\ref{assumption:linkfunction}. 
    Let $\mu'_{\min}$ be the minimum value of $\mu'$  over the set $\mathcal{B}({C_d\|\theta^\ast\|}) = \{t\in\real:|t|\leq C_d\|\theta^\ast\|\}$, with $C_d$ as defined in Assumption~\ref{assump:Bdd-norm-covariates}. 
    That is,
     $$\mu'_{\min}= \min_{t\in\mathcal{B}({C_d\|\theta^\ast\|})} \mu'(t).$$
     Given $c\in(0,1)$, let $\Delta\in\real^d$ be a vector such that $\|\Delta\| \leq \frac{(1-c) \mu'_{\min}}{M_2C_d}$. Then, for any $\widetilde{\theta}\in\real^d$ such that $\|\widetilde{\theta}  - \theta^\ast\| \leq \|\Delta\|$, it holds that
     \begin{equation}
         \Delta^\top \nabla^2L(\widetilde{\theta}) \Delta \geq c \Delta^\top\nabla^2L(\theta^\ast)\Delta \geq  c\lambda_{\min}\left(\nabla^2L(\theta^\ast) \right)\|\Delta\|^2. \label{eq:lemma-glm-restrictedstrongconvexity}
     \end{equation}
\end{lemma}

\begin{proof}
    First, observe that the strict convexity of $\mu$ implies that $\mu'(t)>0$ for all $t\in\real$, and hence $\mu'_{\min}>0$, so the upper bound for $\|\Delta\|$ is well defined. 
   We now proceed by analyzing  each $\mu'(X_{ij}^\top \widetilde{\theta})$ appearing on Equation~\eqref{eq:Hessian-matrix-definition} to show that these values are lower bounded by $c\mu'(X_{ij}^\top\theta^\ast)$ in a neighborhood
    of $X_{ij}^\top\theta^\ast$.
    
    By Assumption~\ref{assumption:linkfunction},
    $|\mu''(t)|\leq M_2$, which implies  $\mu'$ is Lipschitz continuous with Lipschitz constant bounded by $M_2$. Hence, by Lipschitz continuity and the Cauchy-Schwarz inequality,
    \begin{align*}
        \mu'(X_{ij}^\top \widetilde{\theta}) - \mu'(X_{ij}^\top\theta^\ast) & \geq -M_2|X_{ij}^\top (\widetilde{\theta} - \theta^\ast)|\\
        & \geq -M_2\|X_{ij}\|\|\widetilde{\theta} - \theta^\ast\|\\
        & \geq -M_2C_d \|\Delta\|,
     \end{align*}
    where the last inequality holds by Assumption~\ref{assump:Bdd-norm-covariates}.
    Therefore, for any $c\in(0,1)$, it holds that
    \begin{align*}
        \mu'(X_{ij}^\top \widetilde{\theta}) - c\mu'(X_{ij}\theta^\ast) & \geq (1-c) \mu'(X_{ij}^\top\theta^\ast)  - M_2C_d \|\Delta\|.
     \end{align*}
     Observe that, by Assumption~\ref{assump:Bdd-norm-covariates}, 
     $|X_{ij}^\top \theta^\ast| \leq C_d \|\theta^\ast\|$, hence $\mu'(X_{ij}^\top \theta^\ast)\geq \mu'_{\min}$.  Therefore, under the assumption that $\|\Delta\| \leq \frac{(1-c)\mu'_{\min}}{M_2C_d}$, it holds that 
     \begin{equation*}
         \mu'(X_{ij}^\top \widetilde{\theta}) - c\mu'(X^\top_{ij}\theta^\ast) \geq 0.
         \label{eq:showing-positivedefinite}
     \end{equation*}
     Now, to establish \eqref{eq:lemma-glm-restrictedstrongconvexity}, observe that
     \begin{align*}
         \Delta^\top \nabla^2L(\widetilde{\theta}) \Delta & = \sum_{(i,j)\in \widetilde{\mathcal{S}}} \mu'(X_{ij}^\top \widetilde{\theta}) (\Delta^\top X_{ij})^2\\
         & \geq c \sum_{(i,j)\in \widetilde{\mathcal{S}}} \mu'(X_{ij}^\top {\theta}^\ast) (\Delta^\top X_{ij})^2\\
         & = c \Delta^\top \nabla^2 L(\theta^\ast) \Delta\\
         & \geq c\lambda_{\min}\left(\nabla^2L(\theta^\ast) \right) \|\Delta\|^2.
     \end{align*}
\end{proof}


The next lemma shows that the minimum of the loss function is located in a neighborhood of $\theta^\ast$ with high probability, and hence, the event defined in Lemma~\ref{lemma:glm-restricted-strong-convexity} for $\Delta = \widehat{\theta} - \theta^\ast$ holds.

\begin{lemma}\label{lemma:thetahat-close-to-theta}
Under the conditions of Lemma~\ref{lemma:glm-restricted-strong-convexity}, set $\kappa = \frac{(1-c)\mu_{\min}'}{M_2C_d}$ 
Then, 
    there exists some constant $C$ that only depends on the constants in the assumptions, such that
    \begin{equation}\label{eq:lemma-ltheta+delta-ltheta}
    \p\left(\|\widehat{\theta} - \theta^\ast\| \leq \kappa\right) \geq 1 - (d+1) \exp( - C \lambda_{\min}(\nabla^2 L(\theta^\ast) )).
    \end{equation}
\end{lemma}

\begin{proof}
    By a second-order Taylor expansion, for any $\Delta\in\real^d$ there exists some $v\in[0,1]$ such that
    \begin{equation*}
        L(\theta^\ast + \Delta) - L(\theta^\ast) = \Delta^\top \nabla L(\theta^\ast) + \frac{1}{2}\Delta^\top \nabla^2L(\theta^\ast + v{\Delta}) \Delta.
    \end{equation*}
    We show that if $\|\Delta\| = \kappa$, then the above quantity is positive with high probability. By convexity of $L$, this implies that $\widehat{\theta}$, the minimizer of $L$, is located within the ball centered at $\theta^\ast$ and with radius $\kappa$, and hence, $\|\widehat{\theta} - \theta^\ast\| \leq \kappa$.
    
    Following the same arguments in the proof of Lemma~\ref{lemma:glm-restricted-strong-convexity}, we can show that
    \begin{align*}
        \Delta^\top \nabla^2L(\theta^\ast + v{\Delta}) \Delta & \geq c \Delta^\top \nabla^2L(\theta^\ast) \Delta\\
        & \geq c\lambda_{\min}(\nabla^2 L(\theta^\ast) ) \|\Delta\|^2.
    \end{align*}
    On the other hand, by Cauchy-Schwarz inequality, 
    $$|\Delta^\top \nabla L(\theta^\ast)| \leq \|\Delta \| \| \nabla L(\theta^\ast)\|.$$
    Therefore, 
    \begin{equation*}
        L(\theta^\ast + \Delta)- L(\theta^\ast) \geq \frac{c}{2}\lambda_{\min}(\nabla^2 L(\theta^\ast) )\|\Delta\|^2 -  \|\Delta \| \| \nabla L(\theta^\ast)\|.
    \end{equation*}
    Hence, to show Equation~\eqref{eq:lemma-ltheta+delta-ltheta}, it is enough to show that
    $$\frac{c}{2}\lambda_{\min}(\nabla^2 L(\theta^\ast) ) \|\Delta\|^2 > \|\Delta \| \| \nabla L(\theta^\ast)\| $$
    with high probability, which implies that $L(\theta^\ast) < L(\theta^\ast + \Delta)$, and hence the minimizer $\widehat{\theta}$ is contained within this ball. By Lemma~\ref{UpperBddP_probabilistic_lemma},
    \begin{align*}
        \p\left(  \| \nabla L(\theta^\ast)\| \geq \frac{c}{2}\lambda_{\min}(\nabla^2 L(\theta^\ast)) \|\Delta\| \right) & \leq (d+1)\exp\left(\frac{ -c^2 \lambda_{\min}^2 (\nabla^2 L(\theta^\ast) ) \|\Delta\|^2/8}{V + M_1 C_d c\lambda_{\min}(\nabla^2 L(\theta^\ast) ) \|\Delta\|/6} \right)\\
        & = (d+1) \exp\left( \frac{-c^2\kappa^2/8}{V/\lambda_{\min}(\nabla^2 L(\theta^\ast) ) + c M_1 C_d \kappa/6} \lambda_{\min}(\nabla^2 L(\theta^\ast) )\right)\\
        & \leq (d+1) \exp\left( \frac{-c^2\kappa^2/8}{C' + c M_1 C_d \kappa/6} \lambda_{\min}(\nabla^2 L(\theta^\ast) )\right)\\
        & \leq (d+1) \exp( -C \lambda_{\min}(\nabla^2 L(\theta^\ast) )),
    \end{align*}
  for some constant $C>0$, under the assumption that $V \leq C'\lambda_{\min}(\nabla^2 L(\theta^\ast) )$ for some $C'>0$ that depends on the constants in 
  Assumption~\ref{assump:eigenval-hessian} 
  and the definition of $V$ in Lemma~\ref{UpperBddP_probabilistic_lemma}. Therefore, Equation~\eqref{eq:lemma-ltheta+delta-ltheta} holds.
\end{proof}

    
The following lemma provides an upper bound in the probabilistic sense of the gradient norm using the exponential concentration inequality. Specifically, we use the Matrix Bernstein inequality in our case, which uses the spectral norm for the matrices. However, the spectral norm is the same as the Euclidean norm for vectors. 

\begin{lemma}\label{UpperBddP_probabilistic_lemma}
    Let $\nabla L(\theta)$ be the gradient of the loss function as defined in Equation~\eqref{eq:gradient-lossfunction}.
    Suppose that 
    Assumptions~\ref{assump:Bdd-norm-covariates} and \ref{assumption:linkfunction}
    hold, and define the quantity 
    \begin{equation*}
        V = C_d^2M_1^2\sum_{(i,j)\in \widetilde{\mathcal{S}}} P_{ij}(1-P_{ij})
    \end{equation*}
    Then, it holds that
  \begin{equation} \label{eq: UpperBddP_probabilistic_L}
   \mathbbm{P} \left( \| \nabla L(\theta^\ast)\| > \epsilon\right) \leq (d+1) \exp{\left( \frac{-\frac{\epsilon^2}{2}}{V + \frac{M_1C_d\epsilon}{3}} \right)} \hspace{5pt} \text{for all} \hspace{5pt} \epsilon \geq 0.
\end{equation}
\end{lemma}

\begin{proof} We use the matrix Bernstein inequality 
\citep{MAL-048}
to obtain the bound.
The  gradient of $L(\theta)$ evaluated at $\theta^\ast$, as defined in Equation~\eqref{eq:Hessian-matrix-definition},  is given by
$$\nabla L(\theta^\ast)= \sum_{(i,j)\in \widetilde{\mathcal{S}}} -\left(\widetilde{B}_{ij} - \mu(X^\top_{ij} \theta^\ast) \right)X_{ij}.$$
Define $W_{ij}\in\real^d$ as  
$W_{ij}= \left(\widetilde B_{ij} - \mu(X^\top_{ij} \theta^\ast) \right)X_{ij}$ for each $(i,j)\in \widetilde{\mathcal{S}}$.
 Note that for a vector, the Euclidean norm is the same as the spectral norm. Hence, below we verify the conditions required for the matrix Bernstein inequality.
\begin{itemize}
    \item For each $i, j\in\mS$, $i\neq j$, $W_{ij}$ is a mean zero random variable. That is, 
    \begin{align*}
        \e[W_{ij}] = & \left(  \e[\widetilde{B}_{ij}] - \mu(X_{ij}\theta^\ast)\right)X_{ij} = 0,
    \end{align*}
     since the expectation is taken conditional on $X$, and  $\e[\widetilde{B}_{ij}] = \e[B_{ij}] = P_{ij} = \mu(X^\top_{ij}\theta^\ast)$ for $i,j \in \mathcal{S}$.
\item For each $i, j\in\mS, i\neq j,$ the norm of  $W_{ij}$ is bounded. That is,
\begin{align*}
        \|W_{ij}\| & = 
        \left\|\left(\widetilde B_{ij} - \mu(X_{ij}^\top\theta^\ast) \right)X_{ij} \right\|
        \\
        & \leq C_d  \max \left\{ |1-P_{ij}|, |-P_{ij}| \right\} \leq 
        C_d,
    \end{align*}
    where the upper bound holds by Assumption~\ref{assump:Bdd-norm-covariates}.

\item Finally we define $v$ as follows,
    \begin{align*}
        v& := \max\left\{ \left\|\e\left[\sum_{(i,j)\in \widetilde{\mathcal{S}}} W_{ij}W_{ij}^\top \right]\right\|, \left\|\e\left[\sum_{(i,j)\in \widetilde{\mathcal{S}}}W_{ij}^\top W_{ij}\right\| \right]\right\}\\
        & = \max \left\{ \left\| \e \left[ \sum_{(i,j)\in \widetilde{\mathcal{S}}}  \left(\widetilde{B}_{ij} - P_{ij} \right)^2 X_{ij}X_{ij}^\top\right]\right\|, \left\| \e\left[ \sum_{(i,j)\in \widetilde{\mathcal{S}}}  \left(\widetilde{B}_{ij} - P_{ij} \right)^2X_{ij}^\top X_{ij}\right]\right\| \right\}.
    \end{align*}
    The second term in the expression satisfies
\begin{align*}
      \left\| \e \left[ \sum_{(i,j)\in \widetilde{\mathcal{S}}}  \left(\widetilde{B}_{ij} - P_{ij} \right)^2 X_{ij}^\top X_{ij}\right]\right\| 
 & = \left\|   \sum_{(i,j)\in \widetilde{\mathcal{S}}} X_{ij}^\top X_{ij} P_{ij}\left(1 - P_{ij} \right) \right\|\\
 &\leq C_d^2 \sum_{(i,j)\in \widetilde{\mathcal{S}}}  P_{ij} (1 - P_{ij})= V.
\end{align*}
Meanwhile, the first term satisfies
\begin{align*}
     \left\| \e \left[ \sum_{(i,j)\in \widetilde{\mathcal{S}}} \left(\widetilde{B}_{ij} - P_{ij} \right)^2 X_{ij} X_{ij}^\top\right]\right\| & = \left\|   \sum_{i,j \in \mathcal{S}} P_{ij}\left(1 - P_{ij}\right) X_{ij}X_{ij}^\top   \right\|\\
 & \leq  \sum_{(i,j)\in \widetilde{\mathcal{S}}} P_{ij}\left(1 - P_{ij}\right)  \left\| X_{ij}X_{ij}^\top   \right\|\\
& \leq C^2_d  \sum_{(i,j)\in \widetilde{\mathcal{S}}}  P_{ij} (1 - P_{ij}) =V,
\end{align*}
 
by using the fact that, $X_{ij}$ being a vector, the spectral norm of $\|X_{ij} X_{ij}^\top\|$ is the same as the Euclidean norm of $\|X_{ij}\|^2 \leq C^2_d$. Note that $V$, being the sum of at most $s(s-1)/2$ many different quantities, is of at most order $s^2$.
\end{itemize}

Under the above conditions, the matrix Bernstein inequality (Theorem 1.6.2 of \cite{MAL-048}) can be expressed in the following form:

\begin{align*}\label{eq: MatrixBern1}
\begin{split}
     \mathbbm{P} \left( \| \nabla L(\theta)\| > \epsilon\right) & \leq (d+1) \exp{\left( \frac{-\frac{\epsilon^2}{2}}{V + \frac{C_d\epsilon}{3}} \right)} \hspace{5pt} \text{for all} \hspace{5pt} \epsilon \geq 0.
\end{split}
\end{align*}

\end{proof}




With these lemmas, we are now ready to proof 
Theorem~\ref{thm:estimation-error-P}.

\begin{proof}
By a second-order Taylor expansion, there exists some $\widetilde{\theta}$ between $\theta^\ast$ and $\widehat{\theta}$ such that
$$L(\widehat{\theta}) = L(\theta^\ast) + (\widehat{\theta} - \theta^\ast)^\top \nabla L(\theta^\ast) +  \frac{1}{2}(\widehat{\theta} - \theta^\ast)^\top\nabla^2 L(\widetilde{\theta})  (\widehat{\theta} - \theta^\ast).$$
Since $L(\widehat{\theta}) \leq L(\theta)$ for any  $\theta\in\real^d$, 
\begin{align*}
    0 & \geq L(\widehat{\theta)} - L(\theta^\ast)\\
    & =  (\widehat{\theta} - \theta^\ast)^\top \nabla L(\theta^\ast) +  \frac{1}{2}(\widehat{\theta} - \theta^\ast)^\top\nabla^2 L(\widetilde{\theta})  (\widehat{\theta} - \theta^\ast).
\end{align*}
Given $c\in(0,1)$, let $\kappa= \frac{(1-c)\mu'_{\min}}{M_2C_d}$ as defined in Lemma~\ref{lemma:glm-restricted-strong-convexity}. Consider the event
$$\mathcal{E} = \{\|\widehat{\theta}-\theta^\ast\| \leq \kappa\}. $$ 
Under the event $\mathcal{E}$, by Lemma~\ref{lemma:glm-restricted-strong-convexity} and 
Assumption~\ref{assump:eigenval-hessian},
we obtain,

\begin{equation*}
    \nabla L(\theta^\ast)^\top (\widehat\theta - {\theta^\ast}) + \frac{CL_s s^2}{2}\| \theta^\ast - \widehat{\theta} \|^2 \leq 0. \hspace{5pt} 
\end{equation*}
Following the application of the Cauchy-Schwarz inequality, the upper bound on the norm difference of $\widehat{\theta}$ and $\theta^\ast$ takes the following form
\begin{equation}\label{eq: UpperBddTheta}
    \frac{CL_s s^2}{2}\left\|\theta^\ast - \widehat{\theta} \right\|^2 \leq \left\|\nabla L(\theta^\ast)\right\| \left\|\theta^\ast - \widehat{\theta}\right\| \implies 
        \left \| \theta^\ast - \widehat{\theta}\right \| \leq \frac{2}{CL_s s^2} \left\|\nabla L(\theta^\ast)\right \|.
\end{equation}
Note that, given the graph $A$, the node and edge covariates 
$Y$ and $Z$, respectively,
the error in estimating $P_{ij}$ by $\widehat{P}_{ij}$ solely depends on the error in estimating $\theta^\ast$ by $\widehat{\theta}$. 
Therefore, using Equation~\eqref{eq: UpperBddTheta} we can obtain an upper bound for the absolute error regarding $P_{ij}$. In the following lemma, we demonstrate the upper bound for $\left|\widehat{P}_{ij} - P_{ij}\right|$. Note that, $\widehat{P}_{ij} \equiv P_{ij}(\widehat{\theta})$ and $P_{ij} \equiv P_{ij}(\theta^\ast)$. Hence, using Taylor series expansion, we have,
$$\widehat{P}_{ij} = P_{ij} + (\widehat{\theta} - \theta^\ast)^\top \nabla P_{ij}(\bar{\theta}),$$
for some $\bar{\theta}$ between $\theta^\ast$ and $\widehat{\theta}$. Hence, 
\begin{equation}\label{eq: UpperBddP_Step1}
    \begin{split}
        \left| \widehat{P}_{ij} - P_{ij} \right| 
        & \leq \left\| (\widehat{\theta} - \theta^\ast)^\top \nabla P_{ij}(\bar{\theta}) \right\| 
        \leq \left\| \theta^\ast - \widehat{\theta} \right\| \left\| \nabla P_{ij}(\bar{\theta})\right\| ,
    \end{split}
\end{equation}
where the rightmost inequality follows from the Cauchy-Schwarz inequality. Using the fact that $\nabla P_{ij}(\bar\theta) = \mu'(X_{ij}^\top\bar{\theta}) X_{ij}$, we can rewrite Equation~\eqref{eq: UpperBddP_Step1} as  $ \left| \widehat{P}_{ij} - P_{ij} \right| \leq \left\|  \theta^\ast - \widehat{\theta} \right\|  |\mu'(X_{ij}^\top\bar{\theta})|\left\| X_{ij} \right\|$.  Moreover, according to 
Assumption~\ref{assump:Bdd-norm-covariates}
the Euclidean norm of the vectors $X_{ij}$'s is bounded by $C_d$, and by 
Assumption~\ref{assumption:linkfunction},
the derivative of $\mu$ is bounded by $M_1$. 
Hence,
    \begin{equation*}
        \left| \widehat{P}_{ij} - P_{ij} \right| \leq M_1C_d \left\| \widehat{\theta} - \theta^\ast \right\|. 
    \end{equation*}
    Note that the above bound holds for every $i,j\in[n], i < j$, and hence, taking the maximum over this set, we get
        \begin{equation*}
        \max_{i,j\in[n]:j>i} \left| \widehat{P}_{ij} - P_{ij} \right| \leq M_1C_d \left\| \widehat{\theta} - \theta \right\|. 
    \end{equation*}
    Further, using  the inequality in Equation~\eqref{eq: UpperBddTheta}, we obtain 
    \begin{equation*}
         \max_{i,j\in[n]:j>i}\left| \widehat{P}_{ij} - P_{ij} \right| \leq \frac{2 M_1 C_d}{CL_s s^2} \left\| \nabla L(\theta^\ast) \right\|.
    \end{equation*}
    The above inequality holds under the event $\mathcal{E}$. For a given $\epsilon>0$, using the law of total probability, we can obtain
\begin{align}
    \mathbbm{P} \left( \max_{i,j\in[n]:j>i} \left| \widehat{P}_{ij} - P_{ij} \right| > \epsilon\right) & \leq \mathbbm{P} \left( \left\{\max_{i,j\in[n]:j>i} \left| \widehat{P}_{ij} - P_{ij} \right| > \epsilon\right\}\cap \mathcal{E}\right) + \mathbbm{P}(\mathcal{E}^c) \nonumber\\
    & \leq  \mathbbm{P} \left( \left\{\frac{2 M_1 C_d}{CL_s s^2} \left\| \nabla L(\theta^\ast) \right\| > \epsilon\right\}\cap \mathcal{E}\right) + \mathbbm{P}(\mathcal{E}^c) \nonumber\\
     & \leq  \mathbbm{P} \left( \frac{2 M_1 C_d}{CL_s s^2} \left\| \nabla L(\theta^\ast) \right\| > \epsilon\right) + \mathbbm{P}(\mathcal{E}^c).\label{eq:theorem-proof-P_ij>epsilon}
\end{align}
The gradient $\nabla L(\theta^\ast)$ is random.
Therefore, using Lemma~\ref{UpperBddP_probabilistic_lemma}, we observe that, 
\begin{align}\label{eq: UpperBddP_probabilistic_P}
      \mathbbm{P} \left( \frac{2 M_1 C_d}{CL_s s^2} \left\| \nabla L(\theta^\ast) \right\| > \epsilon\right)
      & \leq (d+1) \exp{\left( \frac{-\frac{C^2 L^2_s s^4\epsilon^2}{8C_d^2M_1^2}}{V + \frac{CL_s s^2\epsilon}{6}} \right)}\nonumber\\
      & \leq (d+1) \exp{\left( \frac{-\frac{CL_s s^2\epsilon^2}{8C_d^2M_1^2}}{\frac{V}{CL_s s^2} + \frac{\epsilon}{6}} \right)}.
\end{align}
Specific choices of $\epsilon$ as a function of the size of the seeded set $\mathcal{S}$ ensure the right-hand side of Equation~\eqref{eq: UpperBddP_probabilistic_P} decays polynomially as a function of $s$. In particular, the choice of $\epsilon = c_1\frac{\sqrt{\log s}}{L_s^{1/2} s}$ for an appropriate choice of $c_1$ ensures that the right-hand side satisfies
\begin{align*}
    (d+1) \exp{\left( \frac{-\frac{C^2L ^2_s s^4 \epsilon^2}{8C_d^2M_1^2}}{\frac{V}{CL_s s^2} + \frac{\epsilon}{6}} \right)} \leq  (d+1) \exp{\left( \frac{-\frac{c_1C L_s s^2}{8C_d^2M_1^2\rho s^2} \log(s)}{\frac{V}{CL_s s^2} + \frac{\epsilon}{6}} \right)} \leq \frac{d+1}{s}.
\end{align*}
The last inequality uses the fact that $V = C_d^2 M_1^2 \frac{s(s-1)}{2} \rho_s$ and $L_s \geq \rho_s$ by 
Assumption~\ref{assump:eigenval-hessian}.
Hence, we have that
\begin{align}\label{eq:upperbound-pij}
      \mathbbm{P}\left( \max_{i, j\in [n]: j>i} \left| \widehat{P}_{ij} - P_{ij} \right| > c_1  \frac{\sqrt{\log s}}{L_s^{1/2} s} \right)
     \leq \frac{d+1}{s} + \mathbbm{P}(\mathcal{E}^c).
\end{align}
On the other hand, using Lemma~\ref{lemma:thetahat-close-to-theta}, we see that there exists a constant $C'>0$ that only depends on the other constants in the assumptions, such that
\begin{align}\label{eq:proofthm-boundforE}
    \mathbbm{P}(\mathcal{E}^c) & \leq (d+1) \exp \left( -C' \lambda_{\min}(\nabla^2 L(\theta^\ast)) \right)\nonumber\\
    & \leq (d+1) \exp \left( -C' C L_s s^2 \right)\nonumber\\
    & \leq (d+1) \exp(-C'' L_ss^2) \leq (d+1) \exp( -  \log(s)),
\end{align}
under the assumption that $L_s \geq  \log(s)/(C''s^2)$. Therefore, combining Equations~\eqref{eq:upperbound-pij} and \eqref{eq:proofthm-boundforE}, we obtain that
\begin{equation*}
    \mathbbm{P}\left( \max_{i, j\in [n]: j>i} \left| \widehat{P}_{ij} - P_{ij} \right| > c_1  \frac{\sqrt{\log s}}{L_s^{1/2} s} \right) \leq \frac{2(d+1)}{s}.
\end{equation*}
\end{proof}




\subsection{Proof of 
Theorem~\ref{thm:consistency}}

Similar to the previous theorem, before proving 
Theorem~\ref{thm:consistency},
we introduce some lemmas required in the subsequent proof. 
We define the functions $M_P(\cdot), \widehat{M}_P(\cdot): \Pi_n \rightarrow \mathbbm{R}$ using the Frobenius norm as 
\begin{equation*}
    M_P(\Psi^*) = \|P - \Psi^*P{\Psi^*}^\top\|_F^2,
\end{equation*}
\begin{equation*}
    \widehat{M}_P(\Psi^*) = \|P - \Psi^*B{\Psi^*}^\top\|_F^2 - \|P - B\|_F^2,
\end{equation*}
for  $\Psi^* \in \Pi_n$. The function $M_P(\Psi^*)$ computes the difference between a matrix $P$ and a permuted version of it. Meanwhile, the function $\widehat{M}_P(\cdot)$ is used to compare the value of the loss function at $P$ and $\widehat{P}$, that is,
\begin{equation*}\label{eq: MQ_quantities}
     \widehat{M}_P(Q) = \|P - QBQ^\top\|_F^2 - \|P - B\|_F^2,
\end{equation*}
\begin{equation*}
    \widehat{M}_{\widehat{P}}(Q) = \|\widehat{P} - QBQ^\top\|_F^2 - \|\widehat{P} - B\|_F^2.
\end{equation*} 
Observe that the value of $M_P(Q)$ is the value of the loss function in 
Equation~\eqref{eq: QuadOptim-Expectation} 
when $P$ is known and $B$ is substituted by its expectation. Meanwhile, $\widehat{M}_{\widehat{P}}(Q)$ computes the difference between the loss function evaluated at $Q$ and $M_{\widehat{P}}(Q^\ast)$, i.e., the loss function evaluated at the true permutation matrix $Q^\ast$. The following lemma shows the relationship between these quantities. 
   
\begin{lemma}\label{HatM_Phat_and_P_lemma}
The quantities $M_P(Q), \widehat{M}_{\widehat{P}}(Q)$ and $\widehat{M}_P(Q)$ satisfy
\begin{align*}
     \widehat{M}_P(Q) = M_P(Q) - \Delta_1, 
\end{align*}
\begin{equation*}
    \widehat{M}_{\widehat{P}}(Q) = \widehat{M}_P(Q) + \Delta_2,
\end{equation*}
where
\begin{equation}
    \Delta_1 = 4 \sum_{(i,j)\in\mathcal{Q}}\left( [Q B Q^\top]_{ij} - [Q P Q^\top]_{ij} \right)\left( P_{ij} - [QPQ^\top]_{ij} \right),\label{eq:def_delta_1}
\end{equation}
\begin{equation}
    \Delta_2 = 4 \sum_{(i,j)\in\mathcal{Q}}(\widehat{P}_{ij} - P_{ij})\left( B_{ij} - [Q B Q^\top]_{ij}\right),\label{eq:def_delta_2}
\end{equation}
where $\mathcal{Q}$ is defined in 
Equation~\eqref{eq:set-Q}.
\end{lemma}

\begin{proof}
First, we prove the relationship between $\widehat{M}_P(Q)$ and $M_P(Q)$ as follows,
\begin{equation*}
    \begin{split}
        \widehat{M}_P(Q) & = \| P - Q B Q^\top \|_F^2 - \|P - B\|_F^2 \\
        & = \|\left(P - QPQ^\top\right) + \left(QPQ^\top - QBQ^\top\right)\|_F^2 - \|P - B\|_F^2\\
        & = \|P - QPQ^\top\|_F^2 -2 \sum_{i,j}\left( [Q B Q^\top]_{ij} - [Q P Q^\top]_{ij} \right)\left( P_{ij} - [QPQ^\top]_{ij} \right) \\
        & \ \ \ \ + \|Q(P -B)Q^\top\|_F^2 - \|P-B\|_F^2\\
        & = M_P(Q) - \Delta_1,
    \end{split}
\end{equation*}
where we used the fact that $\|Q(P-B)Q^\top\|_F^2 = \|P-B\|_F^2$ and that
$$\sum_{i,j}\left( [Q B Q^\top]_{ij} - [Q P Q^\top]_{ij} \right)\left( P_{ij} - [QPQ^\top]_{ij} \right) = 2 \sum_{(i,j)\in\mathcal{Q}}\left( [Q B Q^\top]_{ij} - [Q P Q^\top]_{ij} \right)\left( P_{ij} - [QPQ^\top]_{ij} \right), $$
since $P_{ij} = [QPQ^\top]_{ij}$ for $(i,j)\notin \mathcal{Q}$. $i>j$.

Now, we provide the connection between $\widehat{M}_{\widehat{P}}(Q)$ and $\widehat{M}_P(Q)$.
\begin{equation*}
    \begin{split}
        \widehat{M}_{\widehat{P}}(Q) & = \| \widehat{P} - Q B Q^\top \|_F^2 - \|\widehat{P} - B\|_F^2 \\
       & = \|(\widehat{P} - P) + (P - Q B Q^\top)\|_F^2 - \| (\widehat{P} - P) + (P - B) \|_F^2 \\
       & = \|P - Q B Q^\top\|_F^2 + 2\sum_{i,j}(\widehat{P}_{ij} - P_{ij})(P_{ij} - [Q B Q^\top]_{ij}) - \|P - B\|_F^2\\
       &\ \ \ \ - 2\sum_{i,j}(\widehat{P}_{ij} - P_{ij})(P_{ij} - B_{ij}) \\
       & = \widehat{M}_P(Q) + 2 \sum_{i,j}(\widehat{P}_{ij} - P_{ij})\left( (P_{ij} - [Q B Q^\top]_{ij} - (P_{ij} - B_{ij}) \right) \\
       & = \widehat{M}_P(Q) + 2 \sum_{i,j}(\widehat{P}_{ij} - P_{ij})\left(  B_{ij}  - [Q B Q^\top]_{ij} \right)\\
       & = \widehat{M}_P(Q) + \Delta_2.
    \end{split}
\end{equation*}
Again, we use the fact that $B_{ij}  = [Q B Q^\top]_{ij} $ for $(i,j)\notin\mathcal{Q}, i>j$, and hence,
$$2 \sum_{i,j}(\widehat{P}_{ij} - P_{ij})\left(  B_{ij}  - [Q B Q^\top]_{ij} \right) = 4  \sum_{(i,j)\in\mathcal{Q}}(\widehat{P}_{ij} - P_{ij})\left(  B_{ij}  - [Q B Q^\top]_{ij} \right).$$
\end{proof}

From the previous relationships, it holds that $Q^\ast$ is the optimal solution to the problem in \eqref{eq: QuadOptim} whenever $\widehat{M}_{\widehat{P}}(Q)> 0$ for all $Q\neq Q^\ast$. The next lemmas provides a probabilistic upper bound for the event $\left\{\widehat{M}_{\widehat{P}}(Q) \leq 0\right\}$.

\begin{lemma}\label{ProbAsSumProbs_lemma}Given a permutation matrix $Q\in\Pi_n$, 
the event $\left\{\widehat{M}_{\widehat{P}}(Q) \leq 0\right\}$ can be written as a subset of a union of events as follows, 
\begin{equation}\label{eq: ProbAsSumProbs_lemma}
     \left\{\widehat{M}_{\widehat{P}}(Q) \leq 0 \right\} \subseteq  \Theta_{1, Q} \cup \Theta_{2, Q} \cup \Theta_3,
\end{equation}
where $\Theta_{1,Q}, \Theta_{2,Q}$ and $\Theta_3$ are defined as
$$\Theta_{1,Q} = \left\{ \Delta_1 \geq \frac{M_P(Q)}{2} \right\},$$ 
$$\Theta_{2, Q} = \left\{\sum_{(i,j) \in\mathcal{Q}} B_{ij} \geq \frac{M_P(Q)sL_s^{1/2}}{16c_1\sqrt{\log s} }\right\},$$ 
$$\Theta_3 =  \left\{ \max_{i,j} \left|\widehat{P}_{ij} - P_{ij} \right| \geq c_1  \frac{\sqrt{\log s}}{s L_s^{1/2}} \right\}.$$
\end{lemma}

\begin{proof}
Using Lemma~\ref{HatM_Phat_and_P_lemma}, we can conclude that $\widehat{M}_{\widehat{P}}(Q) = M_P(Q) - \Delta_1 + \Delta_2$. From the definition of $\Delta_2$ in Equation~\eqref{eq:def_delta_2}, we have
\begin{align*}
    \begin{split}
        |\Delta_2| = 2 \left|\sum_{i,j} \left( \widehat{P}_{ij} - P_{ij} \right) \left( B_{ij} - [QBQ^\top]_{ij} \right)\right| & \leq 2 \max_{i,j} \left| \widehat{P}_{ij} - P_{ij} \right| \sum_{i, j} \left|B_{ij} - [Q B Q^\top]_{ij} \right| \\
        & \leq 8 \max_{i,j} \left| \widehat{P}_{ij} - P_{ij} \right| \sum_{(i,j)\in\mathcal{Q}} B_{ij}.
    \end{split}
\end{align*}
The last inequality comes from the fact that $B_{ij} = [QBQ^\top]_{ij}$ for all $i,j$ with $i\in\mathcal{K}^C$ and $j\in\mathcal{K}^C$, and hence, using the definition of $\mathcal{Q}$ and the symmetry of the matrix $B$,
\begin{align*}
    \sum_{i, j} \left| B_{ij} - [QBQ^\top]_{ij}\right| & = \sum_{j\in\mathcal{K}}\sum_{i\in\mathcal{K}^C}\left| B_{ij} - [QBQ^\top]_{ij}\right| + \sum_{i\in\mathcal{K}}\sum_{j\in\mathcal{K}^C}  \left| B_{ij} - [QBQ^\top]_{ij}\right| \\
    & \quad\quad + \sum_{j\in\mathcal{K}}\sum_{i\in\mathcal{K}}  \left| B_{ij} - [QBQ^\top]_{ij}\right|\\ 
    & = 2\sum_{(i,j)\in\mathcal{Q}}\left| B_{ij} - [QBQ^\top]_{ij}\right|\\
    & \leq 2\sum_{(i,j)\in\mathcal{Q}}\left(\left| B_{ij}\right| + \left|[QBQ^\top]_{ij}\right|\right)\\
    & \leq 4\sum_{(i,j)\in\mathcal{Q}}B_{ij}.\\
\end{align*}
 Combining these, in the following, we write the event  $\left\{\widehat{M}_{\widehat{P}}(Q) \leq 0 \right\}$ as union of multiple events:
\begin{equation*}
    \begin{split}
       \left\{ \widehat{M}_{\widehat{P}}(Q) \leq 0 \right\} & = \left\{ M_P(Q) - \Delta_1 + \Delta_2 \leq 0 \right\}\\
       & \subseteq \left\{-\Delta_1  + \frac{M_P(Q)}{2} \leq 0\right\} \cup \left\{\Delta_2  + \frac{M_P(Q)}{2} \leq 0\right\}\\
       & \subseteq \left\{ \Delta_1\geq \frac{M_P(Q)}{2} \right\} \cup \left\{ |\Delta_2| \geq \frac{M_P(Q)}{2} \right\}\\
       & \subseteq \left\{ \Delta_1 \geq \frac{M_P(Q)}{2} \right\} \cup \left\{\max_{i,j} \left| \widehat{P}_{ij} - P_{ij} \right| \sum_{(i,j)\in\mathcal{Q}} B_{ij} \geq \frac{M_P(Q)}{16}  \right\} \\
       & \subseteq \left\{ \Delta_1 \geq \frac{M_P(Q)}{2} \right\} \cup \left\{ \sum_{(i,j)\in\mathcal{Q}} B_{ij} \geq \frac{M_P(Q) s L_s^{1/2}}{16c_1  \sqrt{\log s}} \right\}\\
       & \ \ \ \ \cup \left\{ \max_{i,j} \left| \widehat{P}_{ij} - P_{ij}\right| \geq c_1 \frac{\sqrt{\log s}}{L_s^{1/2}s}\right\}\\
       & = \Theta_{1, Q} \cup \Theta_{2, Q} \cup \Theta_3.
\end{split}
\end{equation*}
\end{proof}


In the following two lemmas we provide individual upper bounds for the probability of the events $\Theta_{1, Q}$ and $\Theta_{2, Q}$. Specifically, depending on $M_P(Q)$, Lemma~\ref{ProbtermFirst_lemma} provides an upper bound for 
$\mathbbm{P}\left(\Delta_1 > \frac{M_P(Q)}{2} \right)$.

\begin{lemma}\label{ProbtermFirst_lemma}
For the event $\Theta_{1, Q}$ as defined in Lemma~\ref{ProbAsSumProbs_lemma}, its probability can be upper bounded as
\begin{equation}\label{eq: ProbtermFirst}
        \mathbbm{P}\left(\Delta_1 \geq \frac{M_P(Q)}{2} \right) \leq \exp{\left( -\frac{3}{400} M_P(Q) \right)}.
\end{equation}
\end{lemma}

\begin{proof}
    Recall from Equation~\eqref{eq:def_delta_1} that 
    $$\Delta_1  
    = 4 \sum_{(i,j)\in\mathcal{Q}}\left( [Q (B-P) Q^\top]_{ij}\right)\left( P_{ij} - [QPQ^\top]_{ij} \right).$$ The above is a sum of independent random variables $U_{ij} = 4[Q(B-P)Q^\top]_{ij},$ for 
    $i>j$. Since $\mathbbm{E}(B) = P$, the variables $U_{ij}$ have zero-mean for all $i, j$. This implies $\mathbbm{E}(\Delta_1) = 0$. Additionally we also have, $-1 \leq P_{ij} - [Q P Q^\top]_{ij} \leq 1$ and $-1 \leq [Q B Q^\top]_{ij} - [Q P Q^\top]_{ij} \leq 1$ which leads to 
    $$\left|\left( [Q B Q^\top]_{ij} - [Q P Q^\top]_{ij} \right)\left( P_{ij} - [QPQ^\top]_{ij} \right)\right| \leq 1.$$ 
    Therefore, by applying the Bernstein's inequality, we have 
\begin{equation}\label{firstterm}
\begin{split}
    \mathbbm{P}\left(\Delta_1 \geq \frac{M_P(Q)}{2} \right) & \leq \exp{\left( \frac{-\frac {1}{2} \frac{M^2_Q(P)}{4}}{\sum_{i,j}\mathbbm{E}\left[ U^2_{ij} \left( P_{ij} - [QPQ^\top]_{ij} \right)^2\right] + \frac 43  \frac{M_P(Q)}{2}} \right)} \\
    & \leq \exp{\left( \frac{- M^2_Q(P)}{ 128 M_P(Q) + \frac{16}{3}  M_P(Q)}\right)} \\
   & =  \exp{\left( -\frac{3}{400} M_P(Q) \right)},
\end{split}
\end{equation}
 where the last inequality holds due to the fact that $[QBQ^\top]_{ij}$ are Bernoulli random variables with mean $[QPQ^\top]_{ij}$ and variance $[QPQ^\top]_{ij}(1 - [QPQ^\top]_{ij})$. Therefore, the following holds,
\begin{align*}
\begin{split}
    &  \mathbbm{E}\left[ U^2_{ij}\left( P_{ij} - [QPQ^\top]_{ij} \right)^2\right] \leq 16\left( P_{ij} - [QPQ^\top]_{ij} \right)^2.
\end{split}
\end{align*}
Hence, it follows that
\begin{align*}
    \sum_{i,j} \mathbbm{E}\left[ U^2_{ij} \left( P_{ij} - [QPQ^\top]_{ij} \right)^2\right] & \leq 16\sum_{i,j} \left( P_{ij} - [QPQ^\top]_{ij} \right)^2\\
    & = 16\|P - QPQ^\top\|_F^2 = 16M_P(Q).
\end{align*}
\end{proof}


Similar to Lemma~\ref{ProbtermFirst_lemma},  Lemma~\ref{ProbtermSecond_lemma} provides an upper bound for the probability of the event $\Theta_{2, Q}$.

\begin{lemma}\label{ProbtermSecond_lemma}
 Under 
 Assumption~\ref{assump:Diff-P-permutedP}, 
 the upper bound of the probability of the event $\Theta_{2, Q}$ is given by the following,
\begin{equation}\label{eq: ProbtermSecond}
\begin{split}
     \mathbbm{P}\left( \sum_{(i,j)\in\mathcal{Q}} B_{ij} \geq \frac{sL_s^{1/2}M_P(Q)}{16c_1\sqrt{\log s}} \right) \leq \exp\left(-\frac{3}{8}\sum_{(i,j)\in\mathcal{Q}} {P_{ij}}\right).
\end{split}
    \end{equation}
\end{lemma}

\begin{proof}
Observe that the entries of $B_{ij}$ on the set $(i,j)\in\mathcal{Q}$ are independent, as $i>j$. Define $$\xi =  \frac{s L_s^{1/2} M_P(Q)}{16 c_1  \sqrt{\log s}} - \sum_{(i,j)\in\mathcal{Q}} P_{ij}.$$ Note that, by 
Assumption~\ref{assump:Diff-P-permutedP},  
$\xi >0$. Hence, by using Bernstein's inequality,
\begin{equation}\label{eq: sumB_termtwo}
\begin{split}\mathbbm{P}\left(\sum_{(i,j)\in\mathcal{Q}} B_{ij} \geq  \frac{s L_s^{1/2} M_P(Q)}{16 c_1  \sqrt{\log s}} \right) & = \mathbbm{P}\left(\sum_{(i,j)\in\mathcal{Q}} B_{ij} \geq \sum_{(i,j)\in\mathcal{Q}} P_{ij} + \xi \right) \\
    & \leq \exp \left( \frac{-\frac 12 \xi^2}{\sum_{(i,j)\in\mathcal{Q}} P_{ij}(1-P_{ij}) + \frac 13 \xi} \right),\\
    & = \exp \left( \frac{- \xi}{\frac{2}{\xi}\sum_{(i,j)\in\mathcal{Q}} P_{ij}(1-P_{ij}) + \frac 23} \right),\\
    & \leq \exp\left(\frac{-\xi}{2 + \frac{2}{3}}\right) = \exp\left(-\frac{3}{8}\xi\right)\\
    & \leq \exp\left(-\frac{3}{8}\sum_{(i,j)\in\mathcal{Q}} {P_{ij}}\right).
\end{split}
\end{equation}
The last two lines follows from 
Assumption~\ref{assump:Diff-P-permutedP}. 
In particular, taking $c_2^\ast>32 c_1$, we have
$$M_P(Q) > c_2^\ast\frac{\sqrt{\log s}}{L_s^{1/2}s}\sum_{(i,j)\in\mathcal{Q}} P_{ij},$$
which implies that
\begin{align*}
\xi   =  \frac{s L_s^{1/2} M_P(Q)}{16 c_1  \sqrt{\log s}} - \sum_{(i,j)\in\mathcal{Q}} P_{ij}
& \geq \sum_{(i,j)\in\mathcal{Q}} P_{ij} > 0.
\end{align*}
\end{proof}


Now with these lemmas, we proceed to prove 
Theorem~\ref{thm:consistency}.

\begin{proof}
Recall that $\widehat{M}_Q(\widehat{P})$ is defined as follows,
\begin{equation}\label{eq: hatMQPhat}
    \begin{split}
        \widehat{M}_Q(\widehat{P}) & =  \| \widehat{P} -  Q^\prime \widetilde{B} Q^{\prime^\top}\|_F^2 - \| \widehat{P} - Q^\ast \widetilde{B} {Q^\ast}^\top\|_F^2, \\
        & = \| \widehat{P} - QBQ^\top \|_F^2 - \| \widehat{P} - B \|_F^2,
    \end{split}
\end{equation}
where $Q^\prime {Q^\ast}^\top = Q.$
Therefore, without loss of generality, we can consider the simpler case where the true permutation matrix is the identity. The graph matching algorithm described in 
Algorithm~\ref{alg1}
should retrieve the identity matrix as the true permutation ($Q^\ast = I_n$). So, for the algorithm to work, we would like $\|\widehat{P} - B\|_F^2$ to be smaller than $\|\widehat{P} - Q B Q^\top\|_F^2$ for any permutation matrix $Q\neq I_n$ with $Q_{\mS,\mS}=I_s$ (due to the constraints). Hence we can write the following,
\begin{equation*}
    \begin{split}
        \left\{ \widehat{Q} \neq Q^* \right\} & = \left\{ \text{There exists} \hspace{5pt} Q \neq I, Q_{\mS,\mS}=I_s, \hspace{5pt} 
        \text{such that}\hspace{5pt} \widehat{M}_Q(P) \leq 0\right\}\\
        & \subseteq \bigcup_{Q \neq I, Q_{\mS,\mS}=I_s} \Theta_{1,Q} \cup \Theta_{2,Q}\cup \Theta_3\\
        & \subseteq  \left( \bigcup_{k=2}^{n-s} \bigcup_{Q \in\Pi_{n,k}} \Theta_{1,Q} \cup \Theta_{2,Q}\right)\cup \Theta_3.
    \end{split}
\end{equation*}
The last line holds by the fact that we can express $\Pi_n$ as $\Pi_n = \{I\}\cup \bigcup_{k=2}^n \Pi_{n,k}$, and the constraint $Q_{\mS,\mS}=I_s$ implies that there are at most $n-s$ nodes permuted by $Q$.
This implies, in terms of probability, 
\begin{equation}\label{eq: hatQneqQ}
\begin{split}
      \mathbbm{P}\left( \widehat{Q} \neq Q^* \right) & \leq \mathbbm{P}\left(\left( \bigcup_{k=2}^{n-s} \bigcup_{Q \in\Pi_{n,k}} \Theta_{1,Q} \cup \Theta_{2,Q}\right)\cup \Theta_3\right) \\ 
      &\leq   \sum_{k=2}^{n-s} \sum_{Q\in\Pi_{n,k}}\Big(\mathbbm{P}(\Theta_{1,Q}) + \mathbbm{P}(\Theta_{2,Q})\Big) + \mathbbm{P}(\Theta_3)\\
      & \leq \sum_{k=2}^{n-s} \sum_{Q\in\Pi_{n,k}} \left[\mathbbm{P} \left( \Delta_1 \geq \frac{M_P(Q)}{2} \right) + \mathbbm{P}  \left(\sum_{(i,j)\in\mathcal{Q}} B_{ij } \geq \frac{M_P(Q)sL_s^{1/2}}{16c_1 \sqrt{\log s}}\right) \right] \\
      & \quad\quad + \mathbbm{P}\left(  \max_{i,j} \left| \widehat{P}_{ij} - P_{ij} \right| \geq c_1  \frac{\sqrt{\log s}}{L_s^{1/2}s} \right).
\end{split}
\end{equation}

Thus, we need to show that the probability $\widehat{M}_Q(\widehat{P}) \leq 0$ is very small. To achieve this, we would like to upper bound the RHS of Equation~\eqref{eq: hatQneqQ} such that it is negligible with growing graph size $n$. Hence, we upper bound the above mentioned probability term as sum of probability terms obtained using Lemma~\ref{ProbtermFirst_lemma}, Lemma~\ref{ProbtermSecond_lemma} and 
Theorem~\ref{thm:estimation-error-P}.

To show that the above probability is upper bounded by a term such that it is very small as the graph size increases, we require 
Assumption~\ref{assump:Diff-P-permutedP} 
which depends on the number of labels permuted denoted by $k$ out of $n$. Under this assumption, we have the following bounds on the right hand side of Equation~\eqref{eq: hatQneqQ}. This leads to
\begin{align}
    \mathbbm{P}\left( \Delta_1 \geq \frac{M_P(Q)}{2} \right) & \leq \exp \left(-\frac{3}{{400}} M_P(Q) \right)\nonumber\\
    & \leq \exp{\left(-2 k \log n \right)} = \frac{1}{n^{2k}}, \label{eq:Pdelta1}
\end{align}
where we assume $c_2^\ast > 800/3$ in 
Assumption~\ref{assump:Diff-P-permutedP}. 
For the second term in Equation~\eqref{eq: hatQneqQ}, we use the fact that
\begin{align}\label{eq: BdProbTerms}
    \mathbbm{P}\left(\sum_{(i,j)\in\mathcal{Q}} B_{ij} \geq \frac{M_P(Q)s L_s^{1/2}}{16c_1  \sqrt{\log s}}\right)  & \leq \exp\left(-\frac{3}{8}\sum_{(i,j)\in\mathcal{Q}} P_{ij}\right)\\ 
   &   \leq \exp\left(-2k\log n \right), \nonumber
\end{align}
where we take $c_1^\ast > 16/3$ in 
Assumption~\ref{assump:Sparsity}.

Therefore, given graphs $A, \widetilde{B}$ and the covariates, under assumptions 
\ref{assump:eigenval-hessian}, \ref{assump:Bdd-norm-covariates}, \ref{assumption:linkfunction}, \ref{assump:Sparsity} and \ref{assump:Diff-P-permutedP},
as well as the assumptions in and 
Theorem~\ref{thm:estimation-error-P},
we have the following using Equations~\eqref{eq:Pdelta1} and \eqref{eq: BdProbTerms} and observing the fact that 
$|\Pi_{n-s,k}| \leq (n-s)^k$, we have



\begin{equation}
\begin{split}
\mathbbm{P}(\widehat{Q} \neq Q^*) 
 & \leq \sum_{k=2}^{n-s} \sum_{Q\in\Pi_{n,k}: Q_{ii} = 1 \, \text{for} \, i \in \mathcal{S}} [\exp \{-2 k \log (n-s)\}+ \exp \{-2 k \log (n-s)\}] + \frac{2(d+1)}{s} \\ 
 & \leq \sum_{k=2}^{n-s} (n-s)^{k}[\exp \{-2 k \log (n-s)\}+ \exp \{-2 k \log (n-s)\}] + \frac{2(d+1)}{s} \\
& \leq \sum_{k=2}^{n-s} \frac{2}{(n-s)^{k}} + \frac{2(d+1)}{s} \\
& \leq \frac{2}{(n-s)(n-s-1)} + \frac{2(d+1)}{s} \\
& \leq \frac{4}{(n-s)^2} + \frac{2(d+1)}{s} \leq 2 \left( \frac{2}{(n-s)^2} + \frac{d+1}{s} \right)\leq c \left(\frac{1}{(n-s)^2} + \frac{1}{s}\right).
\end{split}    
\end{equation}


Hence, $ \mathbbm{P}\left(\widehat{Q} = Q^\ast \right) \geq 1 - c \left( \frac{1}{(n-s)^2} + \frac{1}{s} \right) \implies \mathbbm{P}\left(\widehat{Q} = Q^* \right) = 1 - o(1)$ as the graph size and the number of seeds grow.

\end{proof}




\subsection{Proof of 
Theorem~\ref{thm:consistency_local}}

In the similar way as the previous theorem,
we introduce some lemmas required to prove Theorem~\ref{thm:consistency_local}. 
We define the functions $M_{P_{\bar{\mS}, \mS}}(\cdot), \widehat{M}_{P_{\bar{\mS}, \mS}}(\cdot): \Pi_{n-s} \rightarrow \mathbbm{R}$ using the Frobenius norm as 
\begin{equation*}
    M_{P_{\bar{\mS}, \mS}}(\Psi^*) = \|{P_{\bar{\mS}, \mS}} - \Psi^* {P_{\bar{\mS}, \mS}}\|_F^2,
\end{equation*}
\begin{equation*}
    \widehat{M}_{P_{\bar{\mS}, \mS}}(\Psi^*) = \|{P_{\bar{\mS}, \mS}} - \Psi^*{B_{\bar{\mS}, \mS}}\|_F^2 - \|{P_{\bar{\mS}, \mS}} - {B_{\bar{\mS}, \mS}}\|_F^2,
\end{equation*}
for  $\Psi^* \in \Pi_{n-s}$. Similar to the interpretations of $M_P(\cdot)$ and $\widehat{M}_P(\cdot)$, the function $M_{P_{\bar{\mS}, \mS}}(\Psi^*)$ computes the difference between a matrix ${P_{\bar{\mS}, \mS}}$ and a row permuted version of it. 
Meanwhile, the function $\widehat{M}_{P_{\bar{\mS}, \mS}}(\cdot)$ is used to compare the value of the loss function at ${P_{\bar{\mS}, \mS}}$ and $\widehat{P}_{{\bar{\mS}, \mS}}$, that is,
\begin{equation*}\label{eq: MQ_quantities}
     \widehat{M}_{P_{\bar{\mS}, \mS}}(Q) = \|{P_{\bar{\mS}, \mS}} - Q{B_{\bar{\mS}, \mS}}\|_F^2 - \|{P_{\bar{\mS}, \mS}} - {B_{\bar{\mS}, \mS}}\|_F^2,
\end{equation*}
\begin{equation*}
    \widehat{M}_{\widehat{P}_{{\bar{\mS}, \mS}}}(Q) = \|\widehat{P}_{\bar{\mS}, \mS} - QB_{\bar{\mS}, \mS}\|_F^2 - \|\widehat{P}_{\bar{\mS}, \mS} - {B_{\bar{\mS}, \mS}}\|_F^2.
\end{equation*} 
Observe that the value of $M_{P_{\bar{\mS}, \mS}}(Q)$ is the value of the loss function in 
Equation~\eqref{eq:lin_assign_prob} when $P$ is known and $B$ is substituted by its expectation. Meanwhile, $\widehat{M}_{\widehat{P}_{\bar{\mS}, \mS}}(Q)$ computes the difference between the loss function evaluated at $Q$ and $M_{\widehat{P}_{\bar{\mS}, \mS}}(Q^\ast)$, i.e., the loss function evaluated at the true permutation matrix $Q^\ast$. 
Note that, without loss of generality, we consider $Q^\ast = I$.
Then the following lemma gives the relationship between these quantities. 

\begin{lemma}\label{HatM_Phat_and_P_lemma_local}
The quantities $M_{P_{\bar{\mS}, \mS}}(Q), \widehat{M}_{\widehat{P}_{\bar{\mS}, \mS}}(Q)$ and $\widehat{M}_{P_{\bar{\mS}, \mS}}(Q)$ satisfy 
\begin{align*}
     \widehat{M}_{P_{\bar{\mS}, \mS}}(Q) = M_{P_{\bar{\mS}, \mS}}(Q) - \widetilde{\Delta}_1, 
\end{align*}
\begin{equation*}
    \widehat{M}_{\widehat{P}_{\bar{\mS}, \mS}}(Q) = \widehat{M}_{P_{\bar{\mS}, \mS}}(Q) + \widetilde{\Delta}_2,
\end{equation*}
where
\begin{equation}
    \widetilde{\Delta}_1 = 2 \sum_{(i,j)\in (\bar{\mS}, \mS)}\left( [Q B]_{ij} - [Q P]_{ij} \right)\left( P_{ij} - [QP]_{ij} \right),\label{eq:def_delta_1_local}
\end{equation}
\begin{equation}
    \widetilde{\Delta}_2 = 2 \sum_{(i,j)\in (\bar{\mS}, \mS)}(\widehat{P}_{ij} - P_{ij})\left( B_{ij} - [Q B]_{ij}\right)\label{eq:def_delta_2_local}.
\end{equation}
\end{lemma}

\begin{proof}
First, we prove the relationship between $\widehat{M}_{P_{\bar{\mS}, \mS}}(Q)$ and $M_{P_{\bar{\mS}, \mS}}(Q)$ as follows,
\begin{equation*}
    \begin{split}
        \widehat{M}_{P_{\bar{\mS}, \mS}}(Q) & = \| {P_{\bar{\mS}, \mS}} - Q {B_{\bar{\mS}, \mS}} \|_F^2 - \|{P_{\bar{\mS}, \mS}} - {B_{\bar{\mS}, \mS}}\|_F^2 \\
        & = \|\left({P_{\bar{\mS}, \mS}} - Q{P_{\bar{\mS}, \mS}}\right) + \left(Q{P_{\bar{\mS}, \mS}} - Q{B_{\bar{\mS}, \mS}}\right)\|_F^2 - \|{P_{\bar{\mS}, \mS}} - {B_{\bar{\mS}, \mS}}\|_F^2\\
        & = \|{P_{\bar{\mS}, \mS}} - Q{P_{\bar{\mS}, \mS}}\|_F^2 -2 \sum_{i,j}\left( [Q B]_{ij} - [Q {P_{\bar{\mS}, \mS}}]_{ij} \right)\left( [{P_{\bar{\mS}, \mS}}]_{ij} - [Q{P_{\bar{\mS}, \mS}}]_{ij} \right) \\
        & \ \ \ \ + \|Q({P_{\bar{\mS}, \mS}} -{B_{\bar{\mS}, \mS}})\|_F^2 - \|{P_{\bar{\mS}, \mS}}-{B_{\bar{\mS}, \mS}}\|_F^2\\
        & = M_{P_{\bar{\mS}, \mS}}(Q) - \widetilde{\Delta}_1,
    \end{split}
\end{equation*}
where we used the fact that $\|Q({P_{\bar{\mS}, \mS}}-{B_{\bar{\mS}, \mS}})\|_F^2 = \|{P_{\bar{\mS}, \mS}}-{B_{\bar{\mS}, \mS}}\|_F^2$.

Now, we provide the connection between $\widehat{M}_{\widehat{P}_{\bar{\mS}, \mS}}(Q)$ and $\widehat{M}_{P_{\bar{\mS}, \mS}}(Q)$.
\begin{equation*}
    \begin{split}
        \widehat{M}_{\widehat{P}_{\bar{\mS}, \mS}}(Q) & = \| \widehat{P}_{\bar{\mS}, \mS} - Q {B_{\bar{\mS}, \mS}} \|_F^2 - \|\widehat{P}_{\bar{\mS}, \mS} - {B_{\bar{\mS}, \mS}}\|_F^2 \\
       & = \|(\widehat{P}_{\bar{\mS}, \mS} - {P_{\bar{\mS}, \mS}}) + ({P_{\bar{\mS}, \mS}} - Q {B_{\bar{\mS}, \mS}})\|_F^2 - \| (\widehat{P}_{\bar{\mS}, \mS} - {P_{\bar{\mS}, \mS}}) + ({P_{\bar{\mS}, \mS}} - {B_{\bar{\mS}, \mS}}) \|_F^2 \\
       & = \|{P_{\bar{\mS}, \mS}} - Q {B_{\bar{\mS}, \mS}} \|_F^2 + 2\sum_{i,j}([\widehat{P}_{\bar{\mS}, \mS}]_{ij} - {[P_{\bar{\mS}, \mS}]}_{ij})({[P_{\bar{\mS}, \mS}]}_{ij} - [Q B]_{ij}) - \|{P_{\bar{\mS}, \mS}} - {B_{\bar{\mS}, \mS}}\|_F^2\\
       &\ \ \ \ - 2\sum_{i,j}([\widehat{P}_{{\bar{\mS}, \mS}}]_{ij} - {[P_{\bar{\mS}, \mS}]}_{ij})({[P_{\bar{\mS}, \mS}]}_{ij} - {[B_{\bar{\mS}, \mS}]}_{ij}) \\
       & = \widehat{M}_{P_{\bar{\mS}, \mS}}(Q) + 2 \sum_{i,j}([\widehat{P}_{{\bar{\mS}, \mS}}]_{ij} - {[P_{\bar{\mS}, \mS}]}_{ij})\left( ({[P_{\bar{\mS}, \mS}]}_{ij} - [Q B ]_{ij} - ([{P_{\bar{\mS}, \mS}}]_{ij} - B_{ij}) \right) \\
       & = \widehat{M}_{P_{\bar{\mS}, \mS}}(Q) + 2 \sum_{i,j}([\widehat{P}_{{\bar{\mS}, \mS}}]_{ij} - [{P_{\bar{\mS}, \mS}}]_{ij})\left(  [{B_{\bar{\mS}, \mS}}]_{ij}  - [Q {B_{\bar{\mS}, \mS}}]_{ij} \right)\\
       & = \widehat{M}_{P_{\bar{\mS}, \mS}}(Q) + \widetilde{\Delta}_2.
    \end{split}
\end{equation*}
\end{proof}

Recall that, $Q^\ast$ is the optimal solution to the problem in \eqref{eq:lin_assign_prob} whenever $\widehat{M}_{\widehat{P}_{{\bar{\mS}, \mS}}}(Q)> 0$ for all $Q\neq Q^\ast$. Similarly the next lemmas provide a probabilistic upper bound for the event $\left\{\widehat{M}_{\widehat{P}_{{\bar{\mS}, \mS}}}(Q) \leq 0\right\}$.

\begin{lemma}\label{ProbAsSumProbs_lemma_local}Given a permutation matrix $Q\in\Pi_{n}$, 
the event $\left\{\widehat{M}_{\widehat{P}_{{\bar{\mS}, \mS}}}(Q) \leq 0\right\}$ can be written as a subset of a union of events as follows, 
\begin{equation}\label{eq: ProbAsSumProbs_lemma_local}
     \left\{\widehat{M}_{\widehat{P}_{{\bar{\mS}, \mS}}}(Q) \leq 0 \right\} \subseteq  \Theta_{1, Q} \cup \Theta_{2, Q} \cup \Theta_3,
\end{equation}
where $\Theta_{1,Q}, \Theta_{2,Q}$ and $\Theta_3$ are defined as
$$\Theta_{1,Q} = \left\{ \widetilde{\Delta}_1 \geq \frac{M_P(Q)}{2} \right\},$$ 
$$\Theta_{2, Q} = \left\{\sum_{(i,j) \in(\bar{\mS}, \mS)} B_{ij} \geq \frac{M_P(Q)sL_s^{1/2}}{16c_1\sqrt{\log s} }\right\},$$ 
$$\Theta_3 =  \left\{ \max_{i,j} \left|\widehat{P}_{ij} - P_{ij} \right| \geq c_1  \frac{\sqrt{\log s}}{s L_s^{1/2}} \right\}.$$
\end{lemma}

\begin{proof}
Using Lemma~\ref{HatM_Phat_and_P_lemma_local}, we can conclude that $\widehat{M}_{\widehat{P}_{{\bar{\mS}, \mS}}}(Q) = M_{P_{\bar{\mS}, \mS}}(Q) - \widetilde{\Delta}_1 + \widetilde{\Delta}_2$. Moreover, from the definition of $\Delta^\ell_2$ in Equation~\eqref{eq:def_delta_2_local}, we have
\begin{align*}
    \begin{split}
        |\Delta^\ell_2| & = 2 \left|\sum_{i,j} \left( [\widehat{P}_{{\bar{\mS}, \mS}}]_{ij} - [{P_{\bar{\mS}, \mS}}]_{ij} \right) \left( [{B_{\bar{\mS}, \mS}}]_{ij} - [Q{B_{\bar{\mS}, \mS}}]_{ij} \right)\right| \\
        & \leq 2 \max_{i,j} \left| [\widehat{P}_{{\bar{\mS}, \mS}}]_{ij} - [{P_{\bar{\mS}, \mS}}]_{ij} \right| \sum_{i, j} \left|[{B_{\bar{\mS}, \mS}}]_{ij} - [Q {B_{\bar{\mS}, \mS}}]_{ij} \right| \\
        & \leq 8 \max_{(i,j) \in \mathcal{Q}} \left| \widehat{P}_{ij} - P_{ij} \right|\sum_{i,j} [{B_{\bar{\mS}, \mS}}]_{ij}, 
    \end{split}
\end{align*}
where $\mathcal{Q}$ is defined as in Equation~\eqref{eq:set-Q}. 
The last inequality comes from the following fact,
\begin{align*}
    \sum_{i, j} \left|[{B_{\bar{\mS}, \mS}}]_{ij} - [Q {B_{\bar{\mS}, \mS}}]_{ij} \right| & \leq \sum_{i, j} \left| [{B_{\bar{\mS}, \mS}}]_{ij} \right| + \left| [Q{B_{\bar{\mS}, \mS}}]_{ij} \right| \leq 2 \sum_{i,j} \left| [{B_{\bar{\mS}, \mS}}]_{ij} \right| \leq 4 \sum_{i,j} [{B_{\bar{\mS}, \mS}}]_{ij}.
\end{align*}
Combining these, in the following, we write the event  $\left\{\widehat{M}_{\widehat{P}_{{\bar{\mS}, \mS}}}(Q) \leq 0 \right\}$ as union of multiple events:
\begin{equation*}
    \begin{split}
       \left\{ \widehat{M}_{\widehat{P}_{{\bar{\mS}, \mS}}}(Q) \leq 0 \right\} & = \left\{ M_{P_{\bar{\mS}, \mS}}(Q) - \widetilde{\Delta}_1 + \widetilde{\Delta}_2 \leq 0 \right\}\\
       & \subseteq \left\{-\widetilde{\Delta}_1  + \frac{M_{P_{\bar{\mS}, \mS}}(Q)}{2} \leq 0\right\} \cup \left\{\widetilde{\Delta}_2  + \frac{M_{P_{\bar{\mS}, \mS}}(Q)}{2} \leq 0\right\}\\
       & \subseteq \left\{ \widetilde{\Delta}_1\geq \frac{M_{P_{\bar{\mS}, \mS}}(Q)}{2} \right\} \cup \left\{ |\widetilde{\Delta}_2| \geq \frac{M_{P_{\bar{\mS}, \mS}}(Q)}{2} \right\}\\
       & \subseteq \left\{ \widetilde{\Delta}_1 \geq \frac{M_{P_{\bar{\mS}, \mS}}(Q)}{2} \right\} \cup \left\{\max_{(i,j) \in \mathcal{Q}} \left| \widehat{P}_{ij} - P_{ij} \right| \sum_{i,j} [{B_{\bar{\mS}, \mS}}]_{ij} \geq \frac{M_{P_{\bar{\mS}, \mS}}(Q)}{16}  \right\} \\
       & \subseteq \left\{ \widetilde{\Delta}_1 \geq \frac{M_{P_{\bar{\mS}, \mS}}(Q)}{2} \right\} \cup \left\{ \sum_{(i,j)\in (\bar{\mS}, \mS)} [{B_{\bar{\mS}, \mS}}]_{ij} \geq \frac{M_{P_{\bar{\mS}, \mS}}(Q) s L_s^{1/2}}{16c_1  \sqrt{\log s}} \right\}\\
       & \ \ \ \ \cup \left\{ \max_{(i,j) \in \mathcal{Q}} \left| \widehat{P}_{ij} - P_{ij} \right| \geq c_1 \frac{\sqrt{\log s}}{L_s^{1/2}s}\right\}\\
       & = \Theta_{1, Q} \cup \Theta_{2, Q} \cup \Theta_3.
\end{split}
\end{equation*}
\end{proof}


In the following two lemmas we provide individual upper bounds for the probability of the events $\Theta_{1, Q}$ and $\Theta_{2, Q}$. Specifically, depending on $M_{P_{\bar{\mS}, \mS}}(Q)$, Lemma~\ref{ProbtermFirst_lemma_local} provides an upper bound for 
$\mathbbm{P}\left(\widetilde{\Delta}_1 > \frac{M_{P_{\bar{\mS}, \mS}}(Q)}{2} \right)$.

\begin{lemma}\label{ProbtermFirst_lemma_local}
For the event $\Theta_{1, Q}$ as defined in Lemma~\ref{ProbAsSumProbs_lemma_local}, its probability can be upper bounded as
\begin{equation}\label{eq: ProbtermFirst_local}
        \mathbbm{P}\left(\widetilde{\Delta}_1 \geq \frac{M_{P_{\bar{\mS}, \mS}}(Q)}{2} \right) \leq \exp{\left( -\frac{3}{400} M_{P_{\bar{\mS}, \mS}}(Q) \right)}.
\end{equation}
\end{lemma}

\begin{proof}
    Recall from Equation~\eqref{eq:def_delta_1_local} that 
    $$\widetilde{\Delta}_1 
    = 2 \sum_{(i,j)\in (\bar{\mS}, \mS)}\left( [QB]_{ij} -[QP]_{ij}\right)\left( P_{ij} - [QP]_{ij} \right).$$ The above is a sum of independent random variables $[{U_{\bar{\mS}, \mS}}]_{ij} = 2[Q({B_{\bar{\mS}, \mS}}-{P_{\bar{\mS}, \mS}})]_{ij},$ for $(i, j) \in (\bar{\mS}, \mS)$. Since $\mathbbm{E}(B) = P$, the variables $U_{ij}$ have zero-mean for all $i, j$. This implies $\mathbbm{E}(\widetilde{\Delta}_1) = 0$. Additionally we also have, $-1 \leq [{P_{\bar{\mS}, \mS}}]_{ij} - [Q {P_{\bar{\mS}, \mS}}]_{ij} \leq 1$ and $-1 \leq [Q {B_{\bar{\mS}, \mS}}]_{ij} - [Q {P_{\bar{\mS}, \mS}}]_{ij} \leq 1$ which leads to 
    $$\left|\left( [Q {B_{\bar{\mS}, \mS}}]_{ij} - [Q {P_{\bar{\mS}, \mS}}]_{ij} \right)\left( [{P_{\bar{\mS}, \mS}}]_{ij} - [Q{P_{\bar{\mS}, \mS}}]_{ij} \right)\right| \leq 1.$$ 
    Therefore, by applying the Bernstein's inequality as before, we have 
\begin{equation}\label{firstterm_local}
\begin{split}
    \mathbbm{P}\left(\widetilde{\Delta}_1 \geq \frac{M_{P_{\bar{\mS}, \mS}}(Q)}{2} \right) & \leq \exp{\left( \frac{-\frac {1}{2} \frac{M^2_{P_{\bar{\mS}, \mS}}(Q)}{4}}{\sum_{i,j}\mathbbm{E}\left[ U^2_{ij} \left( [{P_{\bar{\mS}, \mS}}]_{ij} - [Q{P_{\bar{\mS}, \mS}}]_{ij} \right)^2\right] + \frac 43  \frac{M_{P_{\bar{\mS}, \mS}}(Q)}{2}} \right)} \\
    & \leq \exp{\left( \frac{- M^2_{P_{\bar{\mS}, \mS}}(Q)}{ 128 M_{P_{\bar{\mS}, \mS}}(Q) + \frac{16}{3}  M_{P_{\bar{\mS}, \mS}}(Q)}\right)} \\
   & =  \exp{\left( -\frac{3}{400} M_{P_{\bar{\mS}, \mS}}(Q) \right)},
\end{split}
\end{equation}
 where the last inequality holds due to the fact that $[Q{B_{\bar{\mS}, \mS}}]_{ij}$ are Bernoulli random variables with mean $[Q{P_{\bar{\mS}, \mS}}]_{ij}$ and variance $[Q{P_{\bar{\mS}, \mS}}]_{ij}(1 - [Q{P_{\bar{\mS}, \mS}}]_{ij})$. Therefore, the following holds,
\begin{align*}
\begin{split}
    &  \mathbbm{E}\left[ [{U_{\bar{\mS}, \mS}}]^2_{ij}\left( P_{ij} - [Q{P_{\bar{\mS}, \mS}}]_{ij} \right)^2\right] \leq 16\left( [{P_{\bar{\mS}, \mS}}]_{ij} - [Q{P_{\bar{\mS}, \mS}}]_{ij} \right)^2.
\end{split}
\end{align*}
Hence, it follows that
\begin{align*}
    \sum_{i,j} \mathbbm{E}\left[ [{U_{\bar{\mS}, \mS}}]^2_{ij} \left( [{P_{\bar{\mS}, \mS}}]_{ij} - [Q{P_{\bar{\mS}, \mS}}]_{ij} \right)^2\right] & \leq 16\sum_{i,j} \left( [{P_{\bar{\mS}, \mS}}]_{ij} - [Q{P_{\bar{\mS}, \mS}}]_{ij} \right)^2\\
    & = 16\|{P_{\bar{\mS}, \mS}} - Q{P_{\bar{\mS}, \mS}}\|_F^2 = 16M_{P_{\bar{\mS}, \mS}}(Q).
\end{align*}
\end{proof}


Similar to Lemma~\ref{ProbtermFirst_lemma_local},  Lemma~\ref{ProbtermSecond_lemma_local} provides an upper bound for the probability of the event $\Theta_{2, Q}$.

\begin{lemma}\label{ProbtermSecond_lemma_local}
 Under 
 Assumption~\ref{assump:Diff-P-permutedP}, 
 the upper bound of the probability of the event $\Theta_{2, Q}$ is given by the following,
\begin{equation}\label{eq: ProbtermSecond_local}
\begin{split}
     \mathbbm{P}\left( \sum_{(i,j)\in (\bar{\mS}, \mS)} B_{ij} \geq \frac{sL_s^{1/2}M_{P_{\bar{\mS}, \mS}}(Q)}{16c_1\sqrt{\log s}} \right) \leq \exp\left(-\frac{3}{8}\sum_{(i,j)\in(\bar{\mS}, \mS)} {P_{ij}}\right).
\end{split}
    \end{equation}
\end{lemma}

\begin{proof}
Observe that the entries of $B_{ij}$ on the set $(i,j)\in (\bar{\mS}, \mS)$ are independent, as $i>j$. Define $$\xi^\ell =  \frac{s L_s^{1/2} M_{P_{\bar{\mS}, \mS}}(Q)}{16 c_1  \sqrt{\log s}} - \sum_{(i,j)\in(\bar{\mS}, \mS)} P_{ij}.$$ Note that, by 
Assumption~\ref{assump:Diff-P-permutedP},  
$\xi_\ell >0$. Hence, by using Bernstein's inequality,
\begin{equation}\label{eq: sumB_termtwo_local}
\begin{split}\mathbbm{P}\left(\sum_{(i,j)\in (\bar{\mS}, \mS)} B_{ij} \geq  \frac{s L_s^{1/2} M_{P_{\bar{\mS}, \mS}}(Q)}{16 c_1  \sqrt{\log s}} \right) & = \mathbbm{P}\left(\sum_{(i,j)\in(\bar{\mS}, \mS)} B_{ij} \geq \sum_{(i,j)\in(\bar{\mS}, \mS)} P_{ij} + \xi_\ell \right) \\
    & \leq \exp \left( \frac{-\frac 12 {\xi_\ell}^2}{\sum_{(i,j)\in(\bar{\mS}, \mS)} P_{ij}(1-P_{ij}) + \frac 13 \xi_\ell} \right),\\
    & = \exp \left( \frac{- \xi_\ell}{\frac{2}{\xi_\ell}\sum_{(i,j)\in(\bar{\mS}, \mS)} P_{ij}(1-P_{ij}) + \frac 23} \right),\\
    & \leq \exp\left(\frac{-\xi_\ell}{2 + \frac{2}{3}}\right) = \exp\left(-\frac{3}{8}\xi_\ell\right)\\
    & \leq \exp\left(-\frac{3}{8}\sum_{(i,j)\in(\bar{\mS}, \mS)} {P_{ij}}\right).
\end{split}
\end{equation}
The last two lines follows from 
Assumption~\ref{assump:Diff-P-permutedP}. 
In particular, taking $c_2^\ast>32 c_1$, we have
$$M_{P_{\bar{\mS}, \mS}}(Q) > c_2^\ast\frac{\sqrt{\log s}}{L_s^{1/2}s}\sum_{(i,j)\in (\bar{\mS}, \mS)} P_{ij},$$
which implies that
\begin{align*}
\xi_\ell   =  \frac{s L_s^{1/2} M_{P_{\bar{\mS}, \mS}}(Q)}{16 c_1  \sqrt{\log s}} - \sum_{(i,j)\in (\bar{\mS}, \mS)} P_{ij}
& \geq \sum_{(i,j)\in (\bar{\mS}, \mS)} P_{ij} > 0.
\end{align*}
\end{proof}


Now with these lemmas, we proceed to prove 
Theorem~\ref{thm:consistency_local}.

\begin{proof}
Recall that $\widehat{M}_{\widehat{P}_{{\bar{\mS}, \mS}}}(Q)$ is defined as follows,
\begin{equation}\label{eq: hatMQPhat}
    \begin{split}
        \widehat{M}_{\widehat{P}_{{\bar{\mS}, \mS}}}(Q) & =  \| \widehat{P}_{{\bar{\mS}, \mS}} -  Q^\prime \widetilde{B}_{{\bar{\mS}, \mS}} \|_F^2 - \| \widehat{P}_{{\bar{\mS}, \mS}} - Q^\ast \widetilde{B}_{{\bar{\mS}, \mS}}\|_F^2, \\
        & = \| \widehat{P}_{{\bar{\mS}, \mS}} - QB_{{\bar{\mS}, \mS}} \|_F^2 - \| \widehat{P}_{{\bar{\mS}, \mS}} - B_{{\bar{\mS}, \mS}} \|_F^2,
    \end{split}
\end{equation}
where $Q^\prime {Q^\ast}^\top = Q.$
Therefore, without loss of generality, we can consider the simpler case where the true permutation matrix is the identity. The graph matching algorithm described in 
Algorithm~\ref{alg2}
should retrieve the identity matrix as the true permutation ($Q^\ast = I_{n-s}$). So, for the algorithm to work, we would like $\|\widehat{P}_{{\bar{\mS}, \mS}} - {B_{\bar{\mS}, \mS}}\|_F^2$ to be smaller than $\|\widehat{P}_{{\bar{\mS}, \mS}} - Q B_{{\bar{\mS}, \mS}}\|_F^2$ for any permutation matrix $Q\neq I_{n-s}$. Hence we can write the following,
\begin{equation*}
    \begin{split}
        \left\{ \widehat{Q} \neq Q^* \right\} & = \left\{ \text{There exists} \hspace{5pt} Q \neq I, \hspace{5pt} 
        \text{such that}\hspace{5pt} \widehat{M}_{P_{\bar{\mS}, \mS}}(Q) \leq 0\right\}\\
        & \subseteq \bigcup_{Q \neq I_{n-s}} \Theta_{1,Q} \cup \Theta_{2,Q}\cup \Theta_3\\
        & \subseteq  \left( \bigcup_{k=2}^{n-s} \bigcup_{Q \in\Pi_{n-s,k}} \Theta_{1,Q} \cup \Theta_{2,Q}\right)\cup \Theta_3.
    \end{split}
\end{equation*}
The last line holds by the fact that we can express $\Pi_{n-s}$ as $\Pi_{n-s}= \{I\}\cup \bigcup_{k=2}^{n-s} \Pi_{n-s,k}$.
This implies, in terms of probability, 
\begin{equation}\label{eq: hatQneqQ_local}
\begin{split}
      \mathbbm{P}\left( \widehat{Q} \neq Q^* \right) & \leq \mathbbm{P}\left(\left( \bigcup_{k=2}^{n-s} \bigcup_{Q \in\Pi_{n-s,k}} \Theta_{1,Q} \cup \Theta_{2,Q}\right)\cup \Theta_3\right) \\ 
      &\leq   \sum_{k=2}^{n-s} \sum_{Q\in\Pi_{n-s,k}}\Big(\mathbbm{P}(\Theta_{1,Q}) + \mathbbm{P}(\Theta_{2,Q})\Big) + \mathbbm{P}(\Theta_3)\\
      & \leq \sum_{k=2}^{n-s} \sum_{Q\in\Pi_{n-s,k}} \left[\mathbbm{P} \left( \Delta^\ell_1 \geq \frac{M_{P_{\bar{\mS}, \mS}}(Q)}{2} \right) + \mathbbm{P}  \left(\sum_{(i,j)\in (\bar{\mS}, \mS)} B_{ij } \geq \frac{M_{P_{\bar{\mS}, \mS}}(Q)sL_s^{1/2}}{16c_1 \sqrt{\log s}}\right) \right] \\
      & \quad\quad + \mathbbm{P}\left(  \max_{i,j} \left| \widehat{P}_{ij} - P_{ij} \right| \geq c_1  \frac{\sqrt{\log s}}{L_s^{1/2}s} \right).\\
\end{split}
\end{equation}

Thus, we need to show that the probability $\widehat{M}_{\widehat{P}_{\bar{\mS}, \mS}}(Q) \leq 0$ is very small. To achieve this, we would like to upper bound the RHS of Equation~\eqref{eq: hatQneqQ_local} such that it is negligible with growing graph size $n$. Hence, we upper bound the above mentioned probability term as sum of probability terms obtained using Lemma~\ref{ProbtermFirst_lemma}, Lemma~\ref{ProbtermSecond_lemma} and 
Theorem~\ref{thm:estimation-error-P}.

To show that the above probability is upper bounded by a term such that it is very small as the graph size increases, we require 
Assumption~\ref{assump:Diff-P-permutedP} 
which depends on the number of labels permuted denoted by $k$ out of $n$. Under this assumption, we have the following bounds on the right hand side of Equation~\eqref{eq: hatQneqQ_local}. This leads to
\begin{align}
    \mathbbm{P}\left( \widetilde{\Delta}_1 \geq \frac{M_{P_{\bar{\mS}, \mS}}(Q)}{2} \right) & \leq \exp \left(-\frac{3}{{400}} M_{P_{\bar{\mS}, \mS}}(Q) \right)\nonumber\\
    & \leq \exp{\left(-2 k \log (n-s) \right)} = \frac{1}{(n-s)^{2k}}, \label{eq:Pdelta1_local}
\end{align}
where we assume $c_2^\ast > 800/3$ in 
Assumption~\ref{assump:Diff-P-permutedP}. 
For the second term in Equation~\eqref{eq: hatQneqQ_local}, we use the fact that
\begin{align}\label{eq: BdProbTerms_local}
    \mathbbm{P}\left(\sum_{(i,j)\in (\bar{\mS}, \mS)} B_{ij} \geq \frac{M_{P_{\bar{\mS}, \mS}}(Q)s L_s^{1/2}}{16c_1  \sqrt{\log s}}\right)  & \leq \exp\left(-\frac{3}{8}\sum_{(i,j)\in (\bar{\mS}, \mS)} P_{ij}\right)\\ 
   &   \leq \exp\left(-2k\log (n-s) \right), \nonumber
\end{align}
where we take $c_1^\ast > 16/3$ in 
Assumption~\ref{assump:Sparsity}.

Therefore, given graphs $A, \widetilde{B}$ and the covariates, under assumptions 
\ref{assump:eigenval-hessian}, \ref{assump:Bdd-norm-covariates}, \ref{assumption:linkfunction}, \ref{assump:Sparsity} and \ref{assump:Diff-P-permutedP},
as well as the assumptions in 
Theorem~\ref{thm:estimation-error-P},
we have the following using Equations~\eqref{eq:Pdelta1_local} and \eqref{eq: BdProbTerms_local} and observing the fact that $|\Pi_{n-s,k}| \leq (n-s)^k$, we have
\begin{equation}
\begin{split}
 \mathbbm{P}(\widehat{Q} \neq Q^*) 
 & \leq \sum_{k=2}^{n-s} \sum_{Q\in\Pi_{n-s,k}} [\exp \{-2 k \log (n-s)\}+ \exp \{-2 k \log (n-s)\}] + \frac{2(d+1)}{s} \\
 & \leq \sum_{k=2}^{n-s} (n-s)^{k}[\exp \{-2 k \log (n-s)\}+ \exp \{-2 k \log (n-s)\}] + \frac{2(d+1)}{s} \\
& \leq \sum_{k=2}^{n-s} \frac{2}{(n-s)^{k}} + \frac{2(d+1)}{s} \\
& \leq \frac{2}{(n-s)(n-s-1)} + \frac{2(d+1)}{s} \\
& \leq \frac{4}{(n-s)^2} + \frac{2(d+1)}{s} \leq 2 \left( \frac{2}{(n-s)^2} + \frac{d+1}{s} \right)\leq c_\ell \left(\frac{1}{(n-s)^2} + \frac{1}{s}\right).
\end{split}    
\end{equation}


Hence, $ \mathbbm{P}\left(\widehat{Q} = Q^\ast \right) \geq 1 - c \left( \frac{1}{(n-s)^2} + \frac{1}{s} \right) \implies \mathbbm{P}\left(\widehat{Q} = Q^* \right) = 1 - o(1)$ as the graph size and the number of seeds grow.

\end{proof}


\subsection{Proof of Proposition~\ref{prop: eigenval-hessian}}

Before proving the following proposition, we define some asymptotic notations for sequences $a_n$ and $b_n$ that we use for the rest of the appendix. $(i)\, a_n \prec b_n \implies a_n = o(b_n), \ (ii)\, a_n \preceq b_n \implies a_n = O(b_n), \ (iii)\, a_n \succ b_n \implies a_n = \omega(b_n)$.

Before we complete the proof of Proposition~\ref{prop:Bd-Diff-P-permutedP}, we prove an auxiliary result regarding $\rho_s$.

\begin{lemma}\label{lemma:rho-s}
    Suppose that $A\sim \operatorname{ER}(p_n)$, $Y\sim \operatorname{ER}(q_n)$, let $P_{ij} = \theta^\ast_0 + \theta^\ast_1 A_{ij} + \theta^\ast_2 Y_{ij}$ and $\rho_s = \frac{2}{s(s-1)}\sum_{(i,j) \in \widetilde{S}} P_{ij}(1-P_{ij})$. Then, given a sequence $\{t_n\}$ under the assumptions that (i) $t_n s^2\to\infty$ as $n\to\infty$, and (ii) $t_n \geq \operatorname{Var}(P_{ij}(1-P_{ij}))$, we have
    \begin{align}
    \mathbbm{P}\left( \mathbbm{E}(\rho_s)-t_n \leq  \rho_s \leq \mathbbm{E}(\rho_s)+ t_n \right) & \to 1\label{eq:proof-boundrho}
\end{align}
as $n\to\infty$.
\end{lemma}

\begin{proof}
     Define $$H_{ij} = P_{ij}(1-P_{ij}) - \mathbbm{E} \left(P_{ij}(1-P_{ij})\right). $$ 
Observe that $0\leq P_{ij}(1-P_{ij}) \leq 1/2$, therefore
\begin{align*}
    |H_{ij} | & \leq \max\left\{ |P_{ij}(1-P_{ij}) |, \mathbbm{E}(|P_{ij}(1-P_{ij})|) \right\} \leq \frac{1}{2}.
\end{align*}
Therefore, given any $t>0$, by Bernstein's inequality,
\begin{align*}
    \mathbbm{P}\left( \left|\sum_{(i,j) \in \widetilde{S}} H_{ij}\right| \geq t \right) & \leq 2 \exp \left(\frac{ - t^2/2}{\frac{s(s-1)}{2}\operatorname{Var}(H_{ij}) + \frac{t}{6}} \right)
\end{align*}
Letting $t' = \frac{2}{s(s-1)}t$, we can expres the previous inequality in terms of $\rho_s$ as
\begin{align*}
    \mathbbm{P}\left( \left|\rho_s - \mathbbm{E}(\rho_s)\right| \geq t' \right) & \leq 2 \exp \left(\frac{ - \frac{s(s-1)}{4}t'}{\frac{\operatorname{Var}(H_{ij})}{t'} + \frac{1}{3s(s-1)}} \right)
\end{align*}
Letting $t' = t_n$, the right-hand side of the above inequality goes to zero as $n\to\infty$. Therefore,
\begin{align*}
    \mathbbm{P}\left( \mathbbm{E}(\rho_s)-t_n \leq  \rho_s \leq \mathbbm{E}(\rho_s)+t_n \right) & \to 1.
\end{align*}

\end{proof}

Now we are ready to complete the proof of Proposition~\ref{prop: eigenval-hessian}.

\begin{proof}

Under the linear model assumption with a graph $A$ and a single covariate $Y$, the expression for the Hessian matrix $\nabla^2 L(\theta^\ast)$ is given by 
\begin{align*}
    \nabla^2 L(\theta^\ast) =  \begin{pmatrix}
\frac{s(s-1)}{2} & \sum_{i,j \in \mathcal{S}, i>j} A_{ij} & \sum_{i,j \in \mathcal{S}, i>j} Y_{ij}\\
\sum_{i,j \in \mathcal{S}, i>j} A_{ij} & \sum_{i,j \in \mathcal{S}, i>j} A^2_{ij} & \sum_{i,j \in \mathcal{S}, i>j} A_{ij}Y_{ij} \\
\sum_{i,j \in \mathcal{S}, i>j} Y_{ij} & \sum_{i,j \in \mathcal{S}} A_{ij}Y_{ij} & \sum_{i,j \in \mathcal{S}, i>j} Y_{ij}.
\end{pmatrix}.
\end{align*}
In order to proceed with verifying 
Assumption~\ref{assump:eigenval-hessian}, we decompose the Hessian matrix into signal and noise terms as follows, 
\begin{align*}
    \nabla^2 L(\theta^\ast) = \overbrace{\mathbbm{E}\left( \nabla^2 L(\theta^\ast) \right)}^{\substack{\text{Signal}}} + \overbrace{\left\{ \nabla^2 L(\theta^\ast) -  \mathbbm{E}\left(\nabla^2 L(\theta^\ast)\right)\right\}}^{\substack{\text{Noise}}}.
\end{align*}
We focus on the signal term first. Denote the eigenvalues of $\nabla^2L(\theta^\ast)$ by $\lambda_i(\nabla^2L(\theta^\ast))$, for $i = 1, 2, 3$. Then by Weyl's theorem (Theorem $4.3.1$ in \cite{horn2012matrix})
the eigenvalues of $\nabla^2 L(\theta^\ast)$ are bounded as $$\max_{i}\left| \lambda_i \left(\nabla^2 L(\theta^\ast)\right) - \lambda_i \left(\mathbbm{E}\left(\nabla^2 L(\theta^\ast)\right)\right) \right| \leq \| E \|_{\operatorname{op}},$$ 
where $\|E\|_{\operatorname{op}} = \operatorname{inf}\{c \geq 0: \|Ev\| \leq c \|v\| \hspace{5pt} \text{for all} \hspace{5pt} v \in \mathbb{R}^3\}$ is the matrix operator norm of $E$, defined as
$$E = \nabla^2 L(\theta^\ast) -  \mathbbm{E}\left(\nabla^2 L(\theta^\ast)\right). $$
Therefore, 
\begin{align}\label{MinEigHessian}
     \lambda_{\min}\left( \nabla^2 L(\theta^\ast) \right) \geq \lambda_{\min} \left( \mathbbm{E} \left( \nabla^2 L(\theta^\ast)\right)  \right) - \| E \|_{\operatorname{op}},
\end{align}
 Now, consider the signal term given by
\begin{align*}
    \mathbbm{E} \left( \nabla^2 L(\theta) \right) =  \frac{s(s-1)}{2} \begin{pmatrix}
        1 & p_n & q_n \\
        p_n & p_n & p_nq_n \\
        q_n & p_nq_n & q_n\\
    \end{pmatrix}.
\end{align*}
In order to bound its eigenvalues, we write 
$\mathbbm{E} \left( \nabla^2 L(\theta^\ast) \right) = D^{1/2}\Gamma D^{1/2},$ for $\Gamma, D\in\real^{3 \times 3}$ matrices,
where $D$ is a diagonal matrix defined as
$$D = \diag\left(\mathbbm{E} \left( \nabla^2 L(\theta^\ast) \right)\right) = \frac{s(s-1)}{2}\begin{pmatrix}
    1 & 0 & 0\\
    0 & p_n & 0 \\
    0 & 0 & q_n\\
\end{pmatrix}$$
and therefore, we have
$$\Gamma = D^{-1/2} \mathbbm{E} \left( \nabla^2 L(\theta^\ast) \right) D^{-1/2} = \begin{pmatrix}
    1 & \sqrt{p_n} & \sqrt{q_n} \\
    \sqrt{p_n} & 1 & \sqrt{p_nq_n}\\
    \sqrt{q_n} & \sqrt{p_nq_n} & 1\\
\end{pmatrix}.$$
By applying Gershgorin's disc theorem we can bound any eigenvalue of $\Gamma$ which is denoted by $\lambda(\Gamma)$. According to the theorem,  every eigenvalue of $\Gamma$ lies within at least one of the Gershgorin discs defined by $\operatorname{Disc}(1, R_i)$ with center at $1$ and radius $R_i$ for all $i$.  Here, $R_i>0$ is defined as the sum of the absolute values of the non-diagonal entries of $\Gamma$ in row $i$, that is, $R_i = \sum_{j \neq i} \left| \Gamma_{ij}\right|$. That is, $ \left| \lambda(\Gamma) - 1 \right| \leq R_i$ for some $i$ holds for any eigenvalue of $\Gamma$. Then, the minimum eigenvalue of $\Gamma$ (denoted by $\lambda_{\min}(\Gamma)$) satisfies
\begin{align*} 
    1- R_i \leq \lambda_{\min}(\Gamma) \leq 1 + R_i \hspace{5pt} \text{for some} \hspace{5pt} i = 1, 2, 3.
\end{align*}
Note that, in particular, 
$$R_1 = \sqrt{p_n} +\sqrt{q_n},$$ $$R_2 = \sqrt{p_n}+\sqrt{p_nq_n},$$  
$$R_3 =\sqrt{q_n}+\sqrt{p_nq_n}.$$ 
Moreover, $\sqrt{p_n} + \sqrt{q_n} < 1$ implies both $\sqrt{p_n} + \sqrt{p_nq_n} < 1$ and $\sqrt{q_n} + \sqrt{p_nq_n} < 1$. Therefore, given any fixed $n$, the lower bound on $\lambda_{\min}(\Gamma)$ is given by
    $$\lambda_{\min}(\Gamma) \geq \nu > 0,$$  
    where $0< \nu \leq  1 - \sqrt{p_n} - \sqrt{q_n}$ from condition $(i)$ of the proposition. 
Since $\mathbbm{E}\left( \nabla^2 L(\theta) \right) = D^{1/2}\Gamma D^{1/2}$, it follows that,
\begin{align}\label{eq: MinsignalEig}
    \lambda_{\min}\left( \mathbbm{E} (\nabla^2 L(\theta) \right) \geq \lambda_{\min}(\Gamma) \min\left(\diag(D)\right) \geq \nu \frac{s(s-1)}{2} \min \{p_n, q_n\}.
\end{align}

Next we focus on the noise matrix $E$, given by $$E = \nabla^2 L(\theta) - \mathbbm{E}(\nabla^2 L(\theta)) = \begin{pmatrix}
 0 & W_1 & W_2\\
 W_1 & W_1 & W_3\\
 W_2 & W_3 & W_2
\end{pmatrix},$$ where 
$$W_1 = \sum_{i,j \in \mathcal{S}, i>j} (A_{ij} - p_n),$$ 
$$W_2 = \sum_{i,j \in \mathcal{S}, i>j} (Y_{ij} - q_n),$$ 
$$W_3 = \sum_{i,j \in \mathcal{S}, i>j} (A_{ij}Y_{ij} - p_n q_n).$$ 
The spectral norm of $E$ satisfies $$\|E\|_{\operatorname{op}} = \sqrt{\lambda_{\max}(E^\top E)} \leq \sqrt{\operatorname{Tr}(E^\top E)} = \sqrt{4(3W_1^2 + 3W_2^2 + 2W_3^3)},$$ 
since $$E^\top E = \begin{pmatrix}
    W_1^2 + W_2^2 & W_1^2 + W_2W_3 & W_1W_3 + W_2^2\\
    W_1^2 + W_2W_3 & 2W_1^2 + W_3^2 & W_1W_2 + W_1W_3 + W_2W_3 \\
    W_1W_3 + W_2^2 & W_1W_2 + W_1W_3 + W_2W_3 & 2W_2^2 + W_3^2\\
\end{pmatrix}.$$
To check that the lower bound in Equation~\eqref{MinEigHessian} is positive, consider $p_n \leq q_n$ for the rest of the proof (the case $q_n < p_n$ is analogous). Then, using Equation~\eqref{MinEigHessian} and Equation~\eqref{eq: MinsignalEig} we have
\begin{align}\label{eq:lambda_minBdd}
\begin{split}
    \lambda_{\min}(\nabla^2 L(\theta^\ast)) & \geq \nu \frac{s(s-1)}{2}  p_n - \|E\|_{\operatorname{op}}\\
   & \geq \nu \frac{s(s-1)}{2} p_n - \sqrt{3W_1^2 + 3W_2^2 + 2W_3^2}.
\end{split}
\end{align}
We apply the Bernstein's inequality for $W_1, W_2$ and $W_3$  (Theorem $1.6.1$ from \cite{MAL-048}), 
which are defined as sums of mean-zero independent random variables. 
In the first concentration bound for $W_1$, 
let $\varepsilon_1 = C_1 s^2p_n$, where $C_1>0$ is some constant. 
Observe that the condition (ii) of the statement implies that $q_n = o(s^2 p_n^2)$,  which implies $s^2p_n \gg \frac{q_n}{p_n} \geq 1$ and hence, $s^2p_n \to \infty$ as $n\to\infty$. Then, we have,
\begin{align}
\label{concntr1}
\mathbbm{P}\left( \left|\sum_{i,j \in \mathcal{S}, i>j} \left(A_{ij} - p_n\right) \right| > \varepsilon_1\right) & \leq 2\exp \left( \frac{-\frac {1}{2} C_1^2 s^4p^2_n}{\frac{s(s-1)}{2}p_n(1-p_n) + \frac {1}{3} C_1s^2p_n} \right)\nonumber\\
&  \leq 2\exp \left( \frac{-\frac {1}{2} C_1^2s^2p_n}{1 + \frac {C_1}{3}} \right) \to 0.
\end{align}

Similarly as before, in the second concentration bound we set $\varepsilon_2 = C_2s^2p_n$,
where $C_2>0$ is some constant. 
Same as before, the 
condition $(ii)$ of the proposition $\max\{p_n, q_n\} = o(s^2 p_n^2)$ implies that $s^2p_n \to \infty$ as $n \to \infty$ and that $q_n/p_n \ll s^2 p_n$, which together imply that $s^2 p_n /(q_n/p_n + 1)\rightarrow\infty$. So we have
\begin{align}
\label{concntr2}
    \mathbbm{P}\left( \left|\sum_{i,j \in \mathcal{S}, i>j} \left(Y_{ij} - q_n\right)\right| > \varepsilon_2 \right) 
         & \leq 2\exp \left( \frac{-\frac 12 C_2^2 s^4p^2_n}{\frac{s(s-1)}{2}q_n(1-q_n) + \frac 13 C_2 s^2p_n} \right)\nonumber\\
    & \leq2 \exp \left( \frac{-\frac 12 C_2^2 s^2p_n}{\frac{s(s-1)q_n}{2s^2 p_n} + \frac 13 C_2} \right) \to 0.
\end{align}

Finally, in the third concentration inequality for $W_3$, we define $\varepsilon_3 = C_3s^2 p_n $, where $C_3>0$ is some constant. Condition $(ii)$ of the proposition ensures that $s^2p_n\to\infty$. Then, we have
\begin{align}
\label{concntr3}
    \mathbbm{P}\left( \left|\sum_{i,j \in \mathcal{S}, i>j} \left(A_{ij}Y_{ij} - p_nq_n\right)\right| > \varepsilon_3 \right) & \leq 2\exp \left( \frac{-\frac 12 C_3^2 s^4p^2_n}{\frac{s(s-1)}{2}p_nq_n(1-p_nq_n) + \frac 13 C_3 s^2p_n} \right) \nonumber\\
    & \leq 2\exp \left( \frac{-\frac 12 C_3^2s^2p_n}{1 + \frac 13 C_3} \right) \to 0.
\end{align}

Hence, combining Equations~\eqref{concntr1}, \eqref{concntr2} and \eqref{concntr3}, we obtain  
$$ \mathbbm{P}\left(|W_1| \leq \varepsilon_1,  |W_2| \leq \varepsilon_2, |W_3| \leq \varepsilon_3  \right) \geq 1-\mathbbm{P}\left( |W_1| > \varepsilon_1 \right) - \mathbbm{P}\left( |W_2 |> \varepsilon_2 \right) - \mathbbm{P}\left( |W_3| > \varepsilon_3 \right)\rightarrow 1.$$
Take $\varepsilon^\ast = \sqrt{3\varepsilon^2_1 + 3\varepsilon^2_2 + 2\varepsilon^2_3}$. This implies
\begin{align*}
  \mathbbm{P}\left(\|E\|_{\operatorname{op}} \leq \varepsilon^* \right)  & =  \mathbbm{P}\left( \sqrt{3W^2_1 + 3W^2_2 + 2W^2_3} \leq \varepsilon^\ast \right)\\
  & \geq \mathbbm{P}(W_1^2 < \epsilon_1^2,\  W_2^2 < \epsilon_2^2, \ W_3^2 < \epsilon_3^2)\\
 & \geq 1-\mathbbm{P}\left( |W_1| > \varepsilon_1 \right) - \mathbbm{P}\left( |W_2 |> \varepsilon_2 \right) - \mathbbm{P}\left( |W_3| > \varepsilon_3 \right)\rightarrow 1.
\end{align*}
 where  . Then, from Equation~\eqref{eq:lambda_minBdd} we can write, for all $0 \leq \delta \leq 1$ there exists $n_1 \in \mathbb{N}$ such that for $n\geq n_1$,
\begin{align*}
    \mathbbm{P}\left( \lambda_{\min} \left( \nabla^2 L(\theta^\ast)\right) \geq \nu\frac{s(s-1)}{2} p_n - \varepsilon^\ast  \right) > 1 - \delta.
\end{align*}
Recall that $\varepsilon_1 = C_1 s^2p_n, \varepsilon_2 = C_2 s^2p_n$ and $\varepsilon_3 = C_3 s^2p_n$. This implies that
$$\varepsilon^\ast = \left(3C_1^2 s^4p_n + 3 C_2^2s^4p_n^2 + 2C_3^2 s^4 p_n^2\right)^{1/2} = C_4 s^2p_n.$$
Take $C_4 \leq \frac{\nu s(s-1)}{2s^2}$. 
Then,
\begin{align*}
      \mathbbm{P} \left( \lambda_{\min}\left(\nabla^2 L(\theta^\ast) \right) \geq \frac{s^2}{4}\nu p_n \right) & \geq \mathbbm{P} \left( \lambda_{\min}\left(\nabla^2 L(\theta^\ast) \right) \geq \frac{s(s-1)}{2}\nu p_n \right)\\
    & \geq \mathbbm{P}\left( \|E\|_{\operatorname{op}} \leq \varepsilon^\ast\right)\to 1.
\end{align*}

We get similar conclusions for the case $q_n \leq p_n$. 
    Therefore, combining the two case of $p_n \leq q_n$ and $q_n \leq p_n$, 
    we obtain the following which verifies 
    Assumption~\ref{assump:eigenval-hessian}, 
    \begin{align*}
        \mathbbm{P}\left( \lambda_{\min}\left(\nabla^2 L(\theta^\ast) \right) \geq \frac{s^2}{4}\nu\min\{p_n, q_n\} \right) = 1- o(1).
    \end{align*}
This ensures that the minimum eigenvalue of the Hessian is positive and bounded away from zero. Finally, in order to show that $L_n \geq \rho_s$, let $L_n = c \min\{p_n, q_n\}$ for some constant $c>0$. Under the assumption that $\theta_0^\ast + \theta_1^\ast p_n + \theta_2^\ast q_n = O(\min\{p_n, q_n\})$, we have
$$\mathbbm{E}(\rho_s) \leq \mathbbm{E}(P_{ij}) = \theta_0^\ast + \theta_1^\ast p_n + \theta_2^\ast q_n \leq c'\min\{p_n, q_n\},$$
 and hence, if  $s^2\min\{p_n, q_n\}\to\infty$, by Lemma~\ref{lemma:rho-s},
\begin{align*}
    \mathbbm{P}\left( L_n \geq \rho_s\right) \geq \mathbbm{P}\left(\mathbbm{E}[\rho_s] + c'\min\{p_n, q_n\} \geq \rho_s \right)\to 1.
\end{align*}
Therefore, combining everything together,
\begin{align*}
        \mathbbm{P}\left( \lambda_{\min}\left(\nabla^2 L(\theta^\ast) \right) \geq \frac{s^2}{4}\nu\min\{p_n, q_n\}, c\min\{p_n, q_n\}\geq \rho_s  \right) = 1- o(1).
    \end{align*}
\end{proof}


\begin{lemma}\label{lemma: BernStein_ER}
    Suppose that $A\sim \operatorname{ER}(p_n)$, $Y\sim \operatorname{ER}(q_n)$ and let $P_{ij} = \theta^\ast_0 + \theta^\ast_1 A_{ij} + \theta^\ast_2 Y_{ij}$ and $\tau_n = \min\{p_n, q_n\}$. Then, 
    \begin{equation*}
      \mathbbm{P}\left( \frac{1}{\tau_n^{1/2}}\sum_{(i, j) \in \mathcal{Q}} P_{ij} > |\mathcal{Q}|t_n \right) \leq  \exp\left(\frac{-|\mathcal{Q}|\left( \sqrt{\tau_n} t_n - \mathbbm{E}(P_{ij}) \right)}{\frac{1}{2|\mathcal{Q}|} + \frac{1}{6|\mathcal{Q}|\mathbbm{E}(P_{ij})}}\right)
    \end{equation*}
    for any $t_n \geq \mathbbm{E}(P_{ij})/\sqrt{\tau_n}$.
\end{lemma}

\begin{proof}
We have the expected value of $P_{ij}$ given by, $\mathbbm{E}(P_{ij}) = \theta^\ast_0 + \theta^\ast_1 p_n + \theta^\ast_2 q_n$. We then define 
$$R_{ij} = P_{ij} - \mathbbm{E}(P_{ij}) = \theta_1^\ast(A_{ij} - p_n) + \theta_2^\ast(Y_{ij} - q_n).$$
Therefore, the $R_{ij}$'s are mean-zero random variables, and we have the following bound,
$$|R_{ij}| \leq 1.$$
For any $t^\prime >0$, we can evaluate $\mathbbm{P}\left(\sum_{(i, j) \in \mathcal{Q}} P_{ij} > t^\prime \right)$ in terms of $R_{ij}$ as follows,
\begin{align*}
    \begin{split}
        \mathbbm{P}\left( \sum_{(i, j) \in \mathcal{Q}} P_{ij} > t^\prime \right) & =  \mathbbm{P}\left( \sum_{(i, j) \in \mathcal{Q}} \left\{ P_{ij} - \mathbbm{E}(P_{ij})\right\} > t^\prime - \sum_{(i, j) \in \mathcal{Q}} \mathbbm{E}(P_{ij}) \right) \\
        & = \mathbbm{P}\left( \sum_{(i, j) \in \mathcal{Q}} R_{ij} > t^\prime - | \mathcal{Q}| \mathbbm{E}(P_{ij})  \right).
    \end{split}
\end{align*}
We further define $t^{\prime \prime} =t' - |\mathcal{Q}|\mathbbm{E}(P_{ij})$.
Since 
$$\mathbbm{E}(R^2_{ij}) = \operatorname{Var}(P_{ij}) = {\theta^\ast_1}^2p_n(1-p_n) + {\theta^\ast_2}^2 q_n (1-q_n),$$ 
by Bernstein's inequality, we have the following,
\begin{align}
\label{eq: BernSteinER}
\begin{split}
     \mathbbm{P}\left( \sum_{(i, j) \in \mathcal{Q}} P_{ij} > t^\prime \right) & = \mathbbm{P}\left( \sum_{(i, j) \in \mathcal{Q}} R_{ij} > t^{\prime\prime} \right) \\
     & \leq \exp \left( \frac{-\frac{1}{2}{t^{\prime\prime}}}{\frac{|\mathcal{Q}|\operatorname{Var}(P_{ij}) }{t^{\prime\prime}} + \frac{1}{3} }  \right).
\end{split}
\end{align}
In particular, for the choice of 
$t^\prime = \sqrt{\tau_n} | \mathcal{Q}| t_n,$
we have
\begin{align*}
    t^{\prime \prime}  & = |\mathcal{Q}|\left( \sqrt{{\tau_n}} t_n - \mathbbm{E}(P_{ij}) \right),
\end{align*}
we can simplify that
\begin{align}
\label{eq: BernSteinER-2}
\begin{split}
     \mathbbm{P}\left( \sum_{(i, j) \in \mathcal{Q}} P_{ij} > \sqrt{\tau_n} | \mathcal{Q}| t_n \right)      & \leq \exp \left( \frac{- |\mathcal{Q}|\left(\sqrt{\tau_n} t_n - \mathbbm{E}(P_{ij}) \right)}{\frac{\operatorname{Var}(P_{ij}) }{2 |\mathcal{Q}|\left( \sqrt{\tau_n} t_n - \mathbbm{E}(P_{ij}) \right)} + \frac{1}{6 |\mathcal{Q}|\left( \sqrt{\tau_n} t_n - \mathbbm{E}(P_{ij}) \right)} }  \right).
\end{split}
\end{align}
Therefore, under the assumption that
$$\sqrt{\tau_n}t_n > 2\mathbbm{E}(P_{ij}) \geq  \mathbbm{E}(P_{ij}) + \operatorname{Var}(P_{ij}), $$
we have that 
$$  \mathbbm{P}\left( \sum_{(i, j) \in \mathcal{Q}} P_{ij} > \sqrt{\tau_n}t_n \right)\leq \exp\left(\frac{-|\mathcal{Q}|\left( \sqrt{\tau_n} t_n - \mathbbm{E}(P_{ij}) \right)}{\frac{1}{2|\mathcal{Q}|} + \frac{1}{6|\mathcal{Q}|\mathbbm{E}(P_{ij})}}\right).$$ 

\end{proof}

\subsection{Proof of Proposition~\ref{prop:Bd-Diff-P-permutedP}}

\begin{proof}
Let $Q\in\Pi_n$ be a permutation. Given a graph with covariates of size $n$, where the graph $A \sim \operatorname{ER}(p_n)$ and the covariate $Y\sim \operatorname{ER}(q_n),$  under the model 
$$P = \theta^*_0 + \theta^*_1A + \theta^*_2Y,$$ 
with $B|A, Y \sim \operatorname{Ber}(P)$ we can write the expression of $M_P(Q)$ as 
\begin{align*}
\begin{split}
      M_P(Q) & =  \|P - QPQ^\top\|^2_F\\
      & = \|\theta^*_1 \left( A - QAQ^\top \right) + \theta^*_2\left( Y - QYQ^\top \right) \|^2_F \\
      & = {\theta^*_1}^2\|A-QAQ^\top\|^2_F + {\theta^*_2}^2\|Y-QYQ^\top\|^2_F + \operatorname{Tr} \left[ 2\theta^*_1 \theta^*_2 (A - QAQ^\top)(Y - QYQ^\top)  \right].
\end{split}
\end{align*} 
Here, the cross product term $\operatorname{Tr} \left[ 2\theta^*_1 \theta^*_2 (A - QAQ^\top)(Y - QYQ^\top)  \right]$ has expected value equal to zero, as we have 
$$ \mathbbm{E} \left[\Tr\left\{ (A - QAQ^\top)(Y - QYQ^\top)  \right\} \right] = \Tr\left[ p_nq_n \left(J - QJQ^\top \right)^2 \right] = 0,$$ where $J = \mathrm{1}_{n \times n},$ the matrix of all ones. This implies the expected value of $M_P(Q)$ is
\begin{align*}
    \mathbbm{E}\left( M_P(Q) \right) = {\theta^*_1}^2 \mathbbm{E}\left\| A - QAQ^\top \right\|^2_F + {\theta^*_2}^2 \mathbbm{E}\left\| Y - QYQ^\top \right\|^2_F.
\end{align*} 

We now obtain a lower bound on the expected value of $M_P(Q)$ similar to Lemma $4$ from \cite{lyzinski2015graph}. Define 
$$\Delta_A = \left\{(i,j) \in [n]\times [n]: i>j, Q_{ii}+ Q_{jj}\leq 1, \text{ and }Q_{ij}+ Q_{ji}\leq 1\right\},$$
that is, $\Delta_A$ contains the indexes of the lower triangular elements of $A$ for which the permutation $Q$ permutes the element $(i,j)$ but $A_{ij}$ and $[QAQ^\top]_{ij}= A_{k_i, k_j}$ (with $Q_{ik_i} = Q_{j,k_j}=1$) represent different random variables. Note that $Q_{ii}+Q_{jj}\leq 1$ implies that either $k_i=i$ or $k_j=j$ but not both, and conversely for $Q_{ij} + Q_{ji}$, which guarantees that $(k_i,k_j)\neq (i,j), (j,i)$. 
Then, we can write
$$\|A - QAQ^\top\|^2_F = 2\sum_{(i,j) \in \Delta_A} \left(A_{ij} - [QAQ^\top]_{ij} \right)^2,$$ 
which is a sum of Bernoulli random variables. For given any $Q \in \Pi_{n, k}$, the bound on the expected value of $ \left(A_{ij} - [QAQ^\top]_{ij} \right)^2$ where $(i,j) \in \Delta_A$, is obtained as below,
\begin{align*}
    \begin{split}
        \mathbbm{E}\left[ \left(A_{ij} - [QAQ^\top]_{ij} \right)^2 \right] & = \mathbbm{E}\left(A^2_{ij}\right) + \mathbbm{E}\left(\left[QAQ^\top\right]^2_{ij} \right) - 2 \mathbbm{E}\left( A_{ij} \cdot \left[ QAQ^\top\right]_{ij} \right)\\
        & = \mathbbm{E}\left(A^2_{ij}\right) + \mathbbm{E}\left(\left[QAQ^\top\right]^2_{ij} \right) - 2 \mathbbm{E}\left(A_{ij} \right) \cdot \mathbbm{E} \left( \left[QAQ^\top \right]_{ij}\right) \\
        & = 2p_n - 2p_n^2 = 2p_n(1-p_n).
    \end{split}
\end{align*}
 Also, define 
 $$\Delta_A^\prime =  \left\{ (i, j) \in [n] \times [n] \ \text{with} \ i \in \mathcal{K} \ \text{or} \ j \in \mathcal{K} :i>j,  Q_{ij} = Q_{ji} = 1\right\}$$
 to be the labels among the permuted edge pairs with the same values for $A_{ij}$ and $\left[QAQ^\top \right]_{ij}$ (that is, $k_i =j$ and $k_j = i$). Write $N_{\Delta_A} = |\Delta_A|$ and $N_{\Delta_A^\prime} =  |\Delta_A^\prime|$.  
 Note that, $N_{\Delta_A^\prime} \leq k$, and we can write 
 \begin{equation}\label{eq:size-Ndelta}
     2N_{\Delta_A} =  n^2 - (n - k)^2 - N_{\Delta_A^\prime} \geq 2nk - k^2 - k \geq k(n-1)\geq \frac{kn}{2}.
 \end{equation} 
 Then, the expected value of  $ \sum_{(i,j) \in \Delta_A} \left(A_{ij} - [QAQ^\top]_{ij} \right)^2$  can be bounded as 
 $$   \mathbbm{E}\left[\sum_{(i,j) \in \Delta_A} \left(A_{ij} - [QAQ^\top]_{ij} \right)^2 \right] \geq p_n(1-p_n) N_{\Delta_A} \geq p_n(1-p_n) \frac{kn}{4}.$$
Because $Y\sim \operatorname{ER}(q_n)$ with $q_n \in [\tau_n, 1-\tau_n]$, in an exact similar way, we have, 
$$ \mathbbm{E}\left[\sum_{(i,j) \in \Delta_Y } \left(Y_{ij} - [QYQ^\top]_{ij} \right)^2 \right] \geq q_n(1-q_n) \frac{kn}{4}.$$
Therefore, using these two, we obtain the following bound on $\mathbbm{E}(M_P(Q))$ given by
\begin{align}
    \mathbbm{E}(M_P(Q)) & \geq \frac{kn}{4}\left( p_n(1-p_n){\theta^*_1}^2 + q_n(1-q_n){\theta^*_2}^2 \right)\nonumber \\
    & \geq \tau_n\frac{kn}{8}\left({\theta^*_1}^2 + {\theta^*_2}^2 \right). \label{eq:lower-bound-EMpq}
\end{align}

Next, we use Proposition~3.2 of \cite{kim2002asymmetry} in a similar way as in  \cite{lyzinski2015graph}. For a given $Q\in\Pi_{n,k}$,  $M_P(Q)$ is a function of $N$ many independent Bernoulli random variables ($A_{ij}$ and $Y_{ij}$), with  $2kn \geq N \geq \frac{kn}{2}$ (as shown in Equation~\eqref{eq:size-Ndelta}), with parameters $p_n, q_n\in[\tau_n,1/2]$. Moreover, changing the value of any one of the Bernoulli random variables changes the value of $M_P(Q)$ by at most $ 4 \left( |\theta^*_1| + |\theta^*_2| \right)^2$.  Set 
$$\chi^2 = 16 \left( |\theta^*_1| + |\theta^*_2| \right)^4\frac{N}{2}(p_n(1-p_n) + q_n(1-q_n)).$$
Therefore, using Proposition~3.2 of \cite{kim2002asymmetry}, we obtain that, for 
any $t$ such that
$$0 \leq t < \frac{2\chi^2}{4 \left( |\theta^*_1| + |\theta^*_2| \right)^2} = 4 \left( |\theta^*_1| + |\theta^*_2| \right)^2N(p_n(1-p_n) + q_n(1-q_n)),$$
it holds that
\begin{align}\label{eq:bound_mpq_prob}
    \mathbbm{P}\left( \left|M_P(Q)- \mathbbm{E}(M_P(Q)) \right| > t\right) & \leq 2\exp\left(\frac{-t^2}{4\chi^2}\right)\nonumber \\
    & \leq 2 \exp \left( - \frac{t^2}{32 \left( |\theta^*_1| + |\theta^*_2| \right)^4N(p_n(1-p_n) + q_n(1-q_n))} \right)\nonumber\\
    & \leq 2 \exp \left( - \frac{t^2}{64kn\left( |\theta^*_1| + |\theta^*_2| \right)^4(p_n + q_n)} \right).
\end{align}

Subsequently, combining Equation~\eqref{eq:lower-bound-EMpq} with Equation~\eqref{eq:bound_mpq_prob}, the following probabilistic bound is obtained as
\begin{align}\label{eq:IneqLemma2Ly}
\begin{split}
     \mathbbm{P}\left(M_P(Q) < \tau_n \frac{kn}{8}\left( {\theta^*_1}^2 + {\theta^*_2}^2 \right) - t \right) &
    \leq \mathbbm{P}\left(M_P(Q) < \mathbbm{E}(M_P(Q)) - t \right)\\
   &  \leq \mathbbm{P}\left(\left| M_P(Q) - \mathbbm{E}(M_P(Q)) \right| > t \right)\\
   & \leq 2 \exp \left( - \frac{t^2}{64 kn\left(|\theta^*_1| + |\theta^*_2|\right)^4(p_n + q_n)} \right).
\end{split}
\end{align}
In particular, we choose  
    $$t = \frac{kn}{16} {\tau_n}({\theta_1^\ast}^2 +{\theta_2^\ast}^2) < N\tau_n (|\theta_1^\ast| +|\theta_2^\ast|)^2.$$ 
Therefore, using Equation~\eqref{eq:IneqLemma2Ly} we get the following,
\begin{align}\label{eq: largeM}
\begin{split}
      \mathbbm{P} \left( M_P(Q) < \frac{\tau_n kn({\theta^*_1}^2 + {\theta^*_2}^2)}{16}  \right) & \leq 2 \exp \left( - \frac{kn\tau_n^2}{256} \right).
\end{split}
\end{align}
Using this, we proceed to show that Assumption~\ref{assump:Diff-P-permutedP} holds with high probability, which requires proving 
\begin{align*}
    \mathbbm{P} \left( M_P(Q) \geq c_2^\ast \max \left\{ k \log (n-s) 
    \frac{ \sqrt{\log s}}{L^{1/2}_n s}\sum_{(i, j) \in \mathcal{Q}}P_{ij} \right\}\text{ for all }Q\in\Pi_{n}, Q_{ii} =1, i\in \mS \right) \to 1.
\end{align*}

For any given $t^\ast >0$, we can write following,
\begin{align}
\label{eq: M_QER_large}
\begin{split}
     & \mathbbm{P} \left( M_P(Q) \geq c_2^\ast \max \left\{ k \log (n-s), 
     \frac{  \sqrt{\log s}}{L^{1/2}_n s}\sum_{(i, j) \in \mathcal{Q}}P_{ij} \right\} \right) \\
     & \geq \mathbbm{P}\left( \left\{ M_P(Q) \geq t^\ast \right\} \cap \left\{ t^\ast \geq c_2^\ast \max \left\{ k \log (n-s), 
     \frac{ \sqrt{\log s}}{L_n^{1/2} s}\sum_{(i, j) \in \mathcal{Q}}P_{ij} \right\}  \right\} \right)\\
     & = 1 - \mathbbm{P}\left( \left\{ M_P(Q) < t^\ast \right\} \cup \left\{ t^\ast < c_2^\ast \max \left\{ k \log (n-s), 
     \frac{  \sqrt{\log s}}{L_n^{1/2} s}\sum_{(i, j) \in \mathcal{Q}}P_{ij} \right\}  \right\} \right)\\
     & \geq 1 - \mathbbm{P}\left( M_P(Q) < t^\ast \right) - \mathbbm{P} \left(  t^\ast < c_2^\ast \max \left\{ k \log (n-s), 
     \frac{  \sqrt{\log s}}{L_n^{1/2} s}\sum_{(i, j) \in \mathcal{Q}}P_{ij} \right\} \right).
\end{split}
\end{align}

In particular, we choose 
$$t^\ast = \frac{\tau_n k n({\theta^\ast_1}^2 + {\theta^\ast_2}^2)}{16}.$$ 
Therefore,  the condition 
$${\theta^\ast_1}^2 + {\theta^\ast_2}^2 \geq \frac{16 c_2^\ast  \log (n-s)}
{\tau_n n},$$
implies that 
$t^\ast \geq c_2^\ast k \log (n-s)$. Then, in Equation~\eqref{eq: M_QER_large} we obtain the following,
\begin{align}\label{eq:p-mpq-proper}
    \begin{split}
        & \mathbbm{P} \left( M_P(Q) \geq c_2^\ast \max \left\{ k \log (n-s),
        \frac{ \sqrt{\log s}}{L_n^{1/2} s}\sum_{(i, j) \in \mathcal{Q}}P_{ij} \right\} \right) \\
        & \geq 1 - \mathbbm{P}\left( M_P(Q) < \frac{\tau_n kn({\theta^\ast_1}^2 + {\theta^\ast_2}^2)}{16} \right) - \mathbbm{P} \left( \frac{c_2^\ast   \sqrt{\log s}}{L_n^{1/2} s}\sum_{(i, j) \in \mathcal{Q}}P_{ij} > \frac{\tau_n kn({\theta^\ast_1}^2 + {\theta^\ast_2}^2)}{16} \right)
    \end{split}
\end{align}

Now, using Equation~\eqref{eq: largeM},
\begin{align*}
    \mathbbm{P}\left( M_P(Q) < \frac{\tau_n kn({\theta^\ast_1}^2 + {\theta^\ast_2}^2)}{16}, \forall Q\in\Pi_{n,k}, k\geq 2\right) & \leq \sum_{k\geq 2}\sum_{Q\in\Pi_{n,k}, 
    Q_{ii} = 1, i \in \mS}
    \mathbbm{P}\left( M_P(Q) < \frac{\tau_n kn({\theta^\ast_1}^2 + {\theta^\ast_2}^2)}{16}\right)\\
    &  \leq \sum_{k\geq 2}
    (n-s)^k\mathbbm{P}\left( M_P(Q) < \frac{\tau_n kn({\theta^\ast_1}^2 + {\theta^\ast_2}^2)}{16}\right)\\
    &  = \sum_{k\geq 2} 2\exp \left(k\log (n-s) -
    \frac{kn\tau_n^2}{256(p_n + q_n)} \right)\\
    &  \leq \sum_{k\geq 2} 2\exp \left(-k\log (n-s)
    \right)\leq \sum_{k\geq 2} 
    \frac{2}{(n-s)^{k}}
    \leq 
    \frac{4}{(n-s)^2},
\end{align*} 
where we used the assumption that
$$512 \frac{\log (n-s)}{n} \leq \frac{\tau_n^2}{p_n + q_n}.$$
Moreover, using Lemma~\ref{lemma: BernStein_ER} we have 
\begin{align*}
\mathbbm{P} \left( \frac{c_2^\ast  \sqrt{\log s}}{L_n^{1/2} s}\sum_{(i, j) \in \mathcal{Q}} P_{ij} > \frac{\tau_n kn({\theta^\ast_1}^2 + {\theta^\ast_2}^2)}{16} \right) & = 
\mathbbm{P} \left( \frac{1}{L_n^{1/2} }\sum_{(i, j) \in \mathcal{Q}, 
Q_{ii} = 1, i \in \mS}P_{ij} 
> \frac{s\tau_n kn({\theta^\ast_1}^2 + {\theta^\ast_2}^2)}{16c_2^\ast  \sqrt{\log s}} \right)\\
& \leq  \exp\left(-ck
(n-s)
\mathbbm{E}(P_{ij})\right)\\
& \leq \exp(-k\log (n-s))\\
& = 
\frac{1}{(n-s)^k}
\end{align*}
where $c>0$ is a constant. The above is valid as long as 
$$\frac{s \tau_n ({\theta^\ast_1}^2 + {\theta^\ast_2}^2)}{16c_2^\ast  \sqrt{\log s}}\geq \frac{\mathbbm{E}(P_{ij})}{\sqrt{\tau_n}},$$
$$\frac{\log (n-s)}{n} \leq c\mathbbm{E}(P_{ij}).$$
Therefore, taking the union over all permutation matrices, 
\begin{align*}
    \mathbbm{P} \left( \frac{c_2^\ast  \sqrt{\log s}}{L_n^{1/2} s}\sum_{(i, j) \in \mathcal{Q}}P_{ij} > \frac{\tau_n kn({\theta^\ast_1}^2 + {\theta^\ast_2}^2)}{2},\forall Q\in\Pi_n \right) & \leq \sum_{k\geq 2}\sum_{Q\in\Pi_{n,k}} 
    \frac{1}{(n-s)^k}
    \leq \frac{2}{(n-s)^2}.
\end{align*}
Combining the above in Equation~\eqref{eq:p-mpq-proper},
$$\mathbbm{P} \left( M_P(Q) \geq c_2^\ast \max \left\{ k \log (n-s), \frac{ \sqrt{\log s}}{L_n^{1/2} s}\sum_{(i, j) \in \mathcal{Q}}P_{ij} \right\}, \forall Q\in\Pi_n \right) \geq 1 - \frac{6}{(n-s)^2} \to 1.$$
\end{proof}


\subsection{Proof of Theorem~\ref{thm:ERconsistency}}

\begin{proof}
According to Theorem~\ref{thm:consistency}, for a given graphs $A, \Tilde{B}$ and $d$ many covariate information, under Assumptions~\ref{assump:eigenval-hessian}, \ref{assump:Bdd-norm-covariates}, \ref{assumption:linkfunction}, \ref{assump:Sparsity} and \ref{assump:Diff-P-permutedP}, we have $\widehat{Q} = Q^*$ with high probability for large $n$. 

Under the  Erd\H{o}s-R\'enyi model where $A \sim \operatorname{ER}(p_n)$ and single covariate $Y\sim\operatorname{ER}(q_n)$, 
Assumption~\ref{assump:eigenval-hessian} 
is satisfied from Proposition~\ref{prop: eigenval-hessian} with probability $1-o(1)$. On the other hand, Proposition~\ref{prop:Bd-Diff-P-permutedP} ensures $P$ and permuted version of $P$ are different enough as in Assumption~\ref{assump:Diff-P-permutedP} with probability $1-o(1)$. 

For the other assumptions, observe that Assumption~\ref{assump:Bdd-norm-covariates} holds for a binary covariate $Y_{ij}$, and Assumption~\ref{assumption:linkfunction} holds for the model with identity link function.  Finally, using Bernstein's inequality, Assumption \ref{assump:Sparsity} holds with probability $1-o(1)$ if 
$$\frac{\log (n-s)}{n} \lesssim \theta_0^\ast + \theta_1^\ast p_n + \theta_2^\ast q_n.$$

Hence, all the  are satisfied under the Erd\H{o}s-R\'enyi model with probability at least $1-o(1)$. Therefore,
\begin{align*}
    \mathbbm{P}(\widehat{Q} = Q^*)  & = \mathbbm{P}(\widehat{Q} = Q^*|\text{assumptions hold})\mathbbm{P}(\text{assumptions hold}) \\
    & \quad\quad +  \mathbbm{P}(\widehat{Q} = Q^*|\text{assumptions do not hold})\mathbbm{P}(\text{assumptions do not hold})\\
    & \leq  \left( 1- c\left(\frac{1}{(n-s)^2} + \frac{1}{s}\right)\right)(1 - o(1)) + o(1)\\
    & \to 1,
\end{align*}
as $s,n\to \infty$.

\end{proof}

\section{Additional simulation}\label{AddFig}

\begin{figure}[ht!]
  \centering
  
  \begin{subfigure}{0.8\textwidth}
    \centering
    \includegraphics[width=\textwidth]{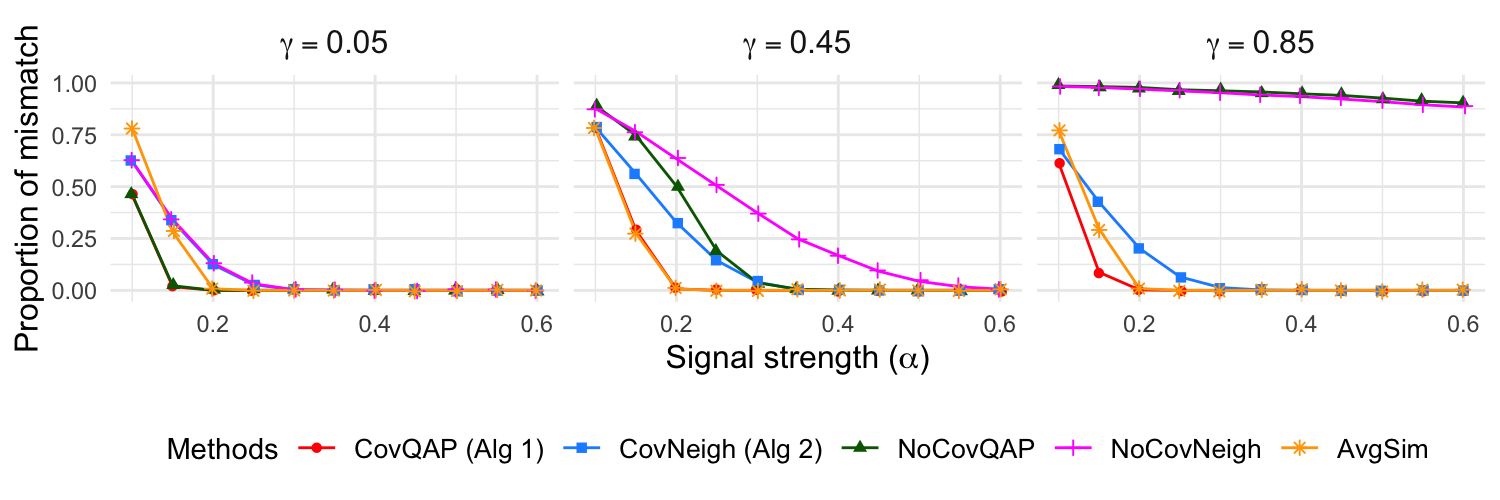}
    \caption{Both the given network and the covariate have an effect in the same direction.}
    \label{fig:easy250}
  \end{subfigure}

  \begin{subfigure}{0.8\textwidth}
    \centering
    \includegraphics[width=\textwidth]
    {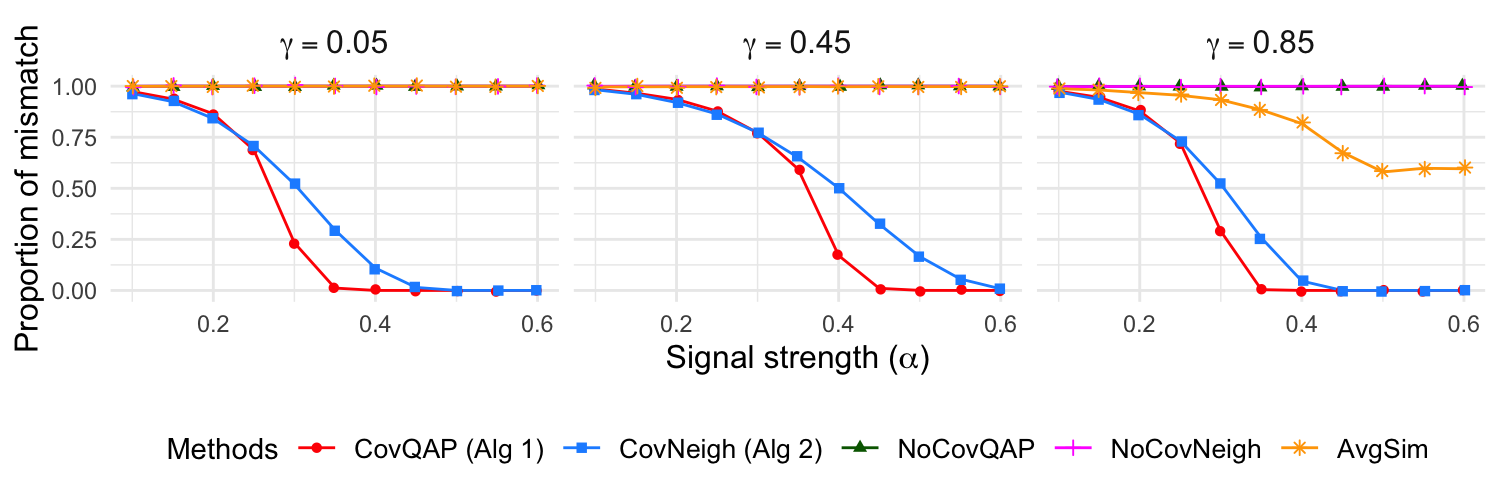}
    \caption{The given network and the covariate have an effect in the opposite direction.}
    \label{fig:difficult250}
  \end{subfigure}
  
  \caption{Comparison of different graph matching algorithms with varying covariate effect $(\gamma)$ and signal strength $(\alpha)$.
  Using $250$ seeds out of $500$ vertices.}
  \label{fig:setup250}
\end{figure}

 \newpage

\section{Additional figure and table for the real data}\label{real}

\begin{figure}[hbt!]
    \centering
    \includegraphics[width=0.8\textwidth]{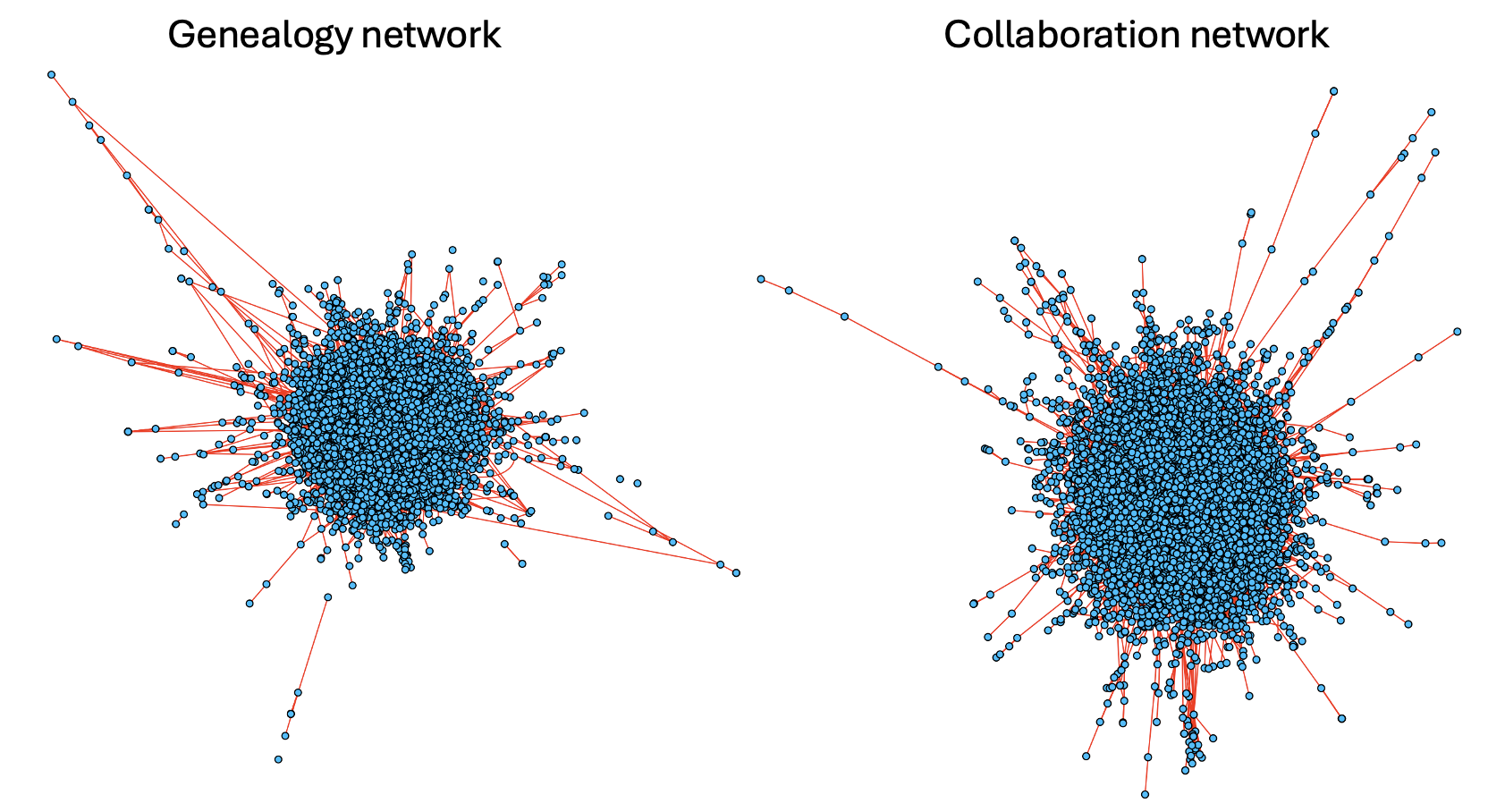}
    \caption{Visualization of the genealogy network (left) obtained from the Mathematics Genealogy Project, and the collaboration network (right) from the MADStat dataset \citep{DVN_V7VUFO_2022}. The networks represent the relationships among $8,627$ statisticians, either by doctoral advisor-advisee (genealogy) or co-autorship of a journal paper (collaboration). As these networks come from different sources, a crucial step in their joint analysis is to find the correspondence between their nodes.}
    \label{fig: networksAB}
\end{figure}


\subsection{Summary result for percentage of matched nodes} 

\begin{table}[hbt!]
\begin{center}
    \begin{tabular}{|c|c|c|c|}
\hline
\begin{tabular}[c]{@{}c@{}}No. of \\ non-seeds \end{tabular} & Methods  & \begin{tabular}[c]{@{}c@{}}Estimated difference \\ in \% of matched nodes\end{tabular} & 95 \% CI     \\ \hline
\multirow{2}{*}{500}  & CovQAP vs NoCovQAP     & 3.94  & (3.44, 4.44) \\ \cline{2-4} 
 & CovNeigh vs NoCovNeigh & 5.16   & (4.80, 5.52) \\ \hline
\multirow{2}{*}{1000}   & CovQAP vs NoCovQAP     & 3.57  & (3.23, 3.92) \\ \cline{2-4}  & CovNeigh vs NoCovNeigh & 4.26  & (3.88, 4.65) \\ \hline
\multirow{2}{*}{2000}   & CovQAP vs NoCovQAP  & 2.41  & (2.14, 2.68) \\ \cline{2-4}   & CovNeigh vs NoCovNeigh & 3.48   & (3.22, 3.74) \\ \hline
\multirow{2}{*}{4000}   & CovQAP vs NoCovQAP  & 1.34  & (0.96, 1.73) \\ \cline{2-4}   & CovNeigh vs NoCovNeigh & 2.54   & (2.07, 2.10) \\ \hline
\end{tabular}
\end{center}
\caption{Comparison of the matching performance between our algorithms and their analog versions with no covariates. The table shows the results of a paired t-test for the experiments reported in Table~\ref{RealDataTab}, including the average accuracy difference along with the $95\%$ confidence interval.}
\label{TestTab}
\end{table}

\end{document}